\documentclass[a4paper,11pt]{article}
\usepackage[utf8]{inputenc}
\usepackage{textcomp}
\usepackage{newunicodechar}
\newunicodechar{−}{-}
\usepackage[a4paper, margin=1in]{geometry}
\usepackage{amsmath}
\usepackage{amssymb}
\usepackage{amsthm}
\usepackage{mathtools}
\usepackage{graphicx}
\usepackage{hyperref}
\usepackage{caption}
\usepackage{cleveref}
\usepackage{subcaption}
\usepackage{algorithm}
\usepackage{algpseudocode}
\usepackage{enumerate}
\usepackage{rotating}
\usepackage{xcolor}
\usepackage{dsfont}
\usepackage{comment}
\usepackage[braket]{qcircuit}
\usepackage{environ}
\newcommand{\myqctmp}[2][0.25]{\Qcircuit @C=#2em @R=#1em @!R}
\NewEnviron{myqcircuit}[1][0.25]{\vcenter{\myqctmp[#1]{0.5} {\BODY}}}

\newtheorem{theorem}{Theorem}[section]
\newtheorem{proposition}[theorem]{Proposition}
\newtheorem{lemma}[theorem]{Lemma}
\newtheorem{corollary}[theorem]{Corollary}
\newtheorem{conjecture}[theorem]{Conjecture}

\theoremstyle{definition}

\newtheorem{remark}[theorem]{Remark}
\newtheorem{assumption}[theorem]{Assumption}

\usepackage[
backend=biber,
style=alphabetic,
sorting=ynt
]{biblatex}

\begin{document}
	
\title{Coordinate-energy transformation and the one-point function for the Heisenberg-Ising XXZ spin-$\tfrac{1}{2}$ chain on the ring}
\author{Eric I.~Corwin, Nikolaus Elsaesser, Axel Saenz}

\date{}
\maketitle

\begin{abstract}
We provide explicit formulas to diagonalize the Hamiltonian for the Heisenberg-Ising XXZ spin-1/2 chain on a discrete ring. Two distinguished bases for the Hilbert space include the basis labeled by the coordinates of the particle configurations and the basis obtained from the eigenvectors of the Hamiltonian. We diagonalize the Hamiltonian by providing an explicit transformation between these two distinguished bases. The transformation from the coordinate basis to the eigenbasis is given by the well-known coordinate Bethe Ansatz. Our contribution is the transformation from the eigenbasis to the coordinate basis, which we call the inverse coordinate Bethe Ansatz transformation/formula. We prove that the inverse coordinate Bethe Ansatz transformation is indeed the inverse of the transformation obtained from the Bethe Ansatz for the case of $N=2$ particles and a ring of odd length $L$ with $|\Delta|<\frac{L-1}{2L}$ and $\Delta$ outside some exceptional finite set. The case of $N>2$ particles and a ring of odd length $L$ is numerically confirmed for different arbitrary choices of parameters and is left as a conjecture. Additionally, assuming that the conjecture is true, we derive an exact formula for the one-point function of the system through special identities for the Izergin-Korepin determinant. Moreover, if the conjecture is true, this implies that the Bethe Ansatz is complete. 

\end{abstract}

\tableofcontents

\section{Introduction}

\emph{Spin} is a discrete conserved quantity for particles in quantum mechanics arising from rotational symmetry. In classical mechanics, rotational symmetry leads to conservation of angular momentum. In particular, angular momentum arises from an infinite dimensional representation of $SO(3)$, the group of spatial rotations. In quantum mechanics, spin arises from finite representations of $SU(2)$, a double cover of $SO(3)$\footnote{Due to Noether's theorem (see \cite[Sec.~4.5]{schwichtenberg_physics_2018}), a rotational symmetry of spacetime along a given axis leads to conserved total angular momentum along that axis. For a particle like an electron, orbital angular momentum and spin are what comprise its total angular momentum.}. Spin interacts with magnetic fields as if the particle were rotating about its center of mass, creating an analogy between bar magnets and spin particles. In 1921, the Stern-Gerlach experiment verified that spin is quantized, i.e.~discrete. For a more detailed description of quantum spin, see \cite[Sec.~3.7, 4.5, 8.5]{schwichtenberg_physics_2018}.

The \emph{Heisenberg-Ising XXZ spin-1/2 chain} is a nearest neighbour spin-spin interaction model with an anisotropy for the spin in the z-direction. It is a well-known and well-studied model in the math and physics literature due to its rich mathematical structure. Each particle has spin $1/2$ or $-1/2$, and we refer to these particles, respectively, as \emph{up-spin} and \emph{down-spin} particles. The XXZ spin-1/2 chain is a specific example of a general class of quantum spin chain models. See \cite{parkinson_introduction_2010} for a general introduction of quantum spin chains. 

In 1931, Hans Bethe developed an Ansatz -- a method -- to obtain a collection of eigenvectors for the Hamiltonian of the XXX spin-1/2 chain, i.e.~the isotropic limit of the XXZ spin-1/2 chain, on a one-dimensional lattice \cite{bethe_zur_1931}. This method is called the \emph{Bethe Ansatz}, and it generalizes to the XXZ spin-1/2 chain and many other models with similar symmetries. Intense study on the XXZ and related models such as the six-vertex model followed Hans Bethe's original ansatz, with much effort put towards algebraic techniques developed in the late 20th century (\cite{Kulish:1981bi}, \cite{Jimbo1994-mc}, \cite{Jimbo1992-zh}, \cite{Davies1993-su}, \cite{Baxter1972-xm}, \cite{Korepin1993-yu}). See \cite{Faddeev1995-io} for a history on one of the most widely used and studied algebraic approaches, called the Quantum Inverse Scattering Method or Algebraic Bethe Ansatz\footnote{Given the vast amount of research on the XXZ spin chain and Bethe Ansatz, we cannot hope to capture all of the work done. Interested readers should look at references within those cited here as well.}. See \cite{baxter_exactly_1982, caux_bethe_2014, Korepin1993-zf,Arutyunov_2026} for a thorough introduction of the Bethe Ansatz.

It is not clear, a priori, that the Bethe Ansatz produces all of the eigenvectors for the Hamiltonian of the XXZ spin-1/2 chain. The \emph{Bethe eigenvectors}, i.e.~eigenvectors produced by the Bethe Ansatz, are determined by solutions of a system of coupled algebraic equations, called \emph{Bethe equations}, where the degree of the equations depends on the length of the spin chain and the size of the system of equations depends on the number of up-spin particles. Some \emph{Bethe roots}, i.e.~solutions of the Bethe equations, may lead to trivial Bethe vectors, Bethe eigenvectors may be linearly dependent, or the Bethe vectors may not generate the entire vector space. 

There are precise arguments to show completeness of the Bethe Ansatz in the literature. In \cite{tarasov_completeness_1995}, the authors give an existence proof for completeness by showing completeness for special parameters and showing that deformations are algebraic and non-degenerate. In \cite{borodin_spectral_2015, borodin_spectral_2015-1}, the authors provide a Plancherel-type isomorphism between the canonical coordinate basis and the Bethe vectors for interacting particle systems on the infinite integer lattice, including the asymmetric simple exclusion process and the XXZ spin-1/2 chain. In this context, our transformations can be viewed as Plancherel bijection formulas, an invertible correspondence between functions and their spectral data. In \cite{GORBOUNOV2017282}, completeness and solutions for the Bethe Ansatz equations are studied in relation to quantum cohomology for a degeneration of the asymmetric six vertex model. Under the \emph{string hypothesis}, the assumption that solutions to the Bethe Ansatz Equations (BAE) have fixed real parts, there are proofs for completeness in the thermodynamic limit ($N,L\to\infty$ with $N/L=constant$) \cite{kirillov_completeness_1997}, \cite{Takahashi1971-uq}. See \cite{PhysRevE.88.052113} and \cite{Jiang2018-vm} and references therein for further context on completeness of the Bethe ansatz. Lastly, we would like to mention the very recent work of Shinsuke Iwao and Kohei Motegi \cite{Iwao2025-vv} where completeness of the Bethe ansatz for the Totally Asymmetric Simple Exclusion Process (TASEP) is proven rigorously via a counting argument with the BAE realized as algebraic equations on a Riemann surface. This approach is similar to the ideas from Sylvain Prolhac (\cite{Prolhac2024-fr}, \cite{Prolhac2020-nn}, \cite{Prolhac2022-xq}).

We consider a constructive approach to the completeness of the Bethe Ansatz. The coordinate Bethe Ansatz provides a basis transformation, assuming the Bethe Ansatz is complete, from the coordinate basis to the energy (i.e.~eigenvector) basis. We provide a closed formula for the inverse transformation, from the energy basis to the coordinate basis, in \Cref{c:main_v3}. We prove the case of $N=2$ up-spin particles in \Cref{s: N=2 proof} for $\Delta < (2L)/(L-1)$ and some additional conditions given in \Cref{a:N=2 proof assumptions} that exclude a finite number of values of $\Delta$. For several cases with $N>2$, we have numerical evidence that our conjecture is true. As a direct consequence, if \Cref{c:main_v3} is true, the following conjecture follows.

\begin{conjecture}\label{c:main_v2}
    The Bethe Ansatz completely determines the spectrum of the Heisenberg-Ising XXZ spin-$\tfrac{1}{2}$ chain on a 1D periodic lattice of odd length $L$  with $N$ up-spins and anisotropy parameter $\Delta \in \mathbb{R}$ for all but a finite number of values of $\Delta$ such that $|\Delta| < \varepsilon$ for some $\varepsilon(N,L) = \varepsilon>0$.
\end{conjecture}

We determine a new closed formula for the one-point function, i.e.~the probability of measuring an up-spin particle at a specific point in space-time, assuming the validity of \Cref{c:main_v3}. Special identities and simplification lead to the following theorem for the one-point function.

\begin{theorem}\label{t:one-point function main}
    Assume the Bethe Ansatz is complete. Then, the one-point function, given by \eqref{e:one-point}, has the following formula 
\begin{equation}
    \rho(x,t) = \sum_{[\boldsymbol{\xi}], [\boldsymbol{\zeta}] \in \Xi} \ell(\mathbf{y}, \boldsymbol{\xi}) \mathcal{F}(x; \boldsymbol{\xi}, \boldsymbol{\zeta})\ell(\mathbf{y}, \boldsymbol{\zeta})
\end{equation}
where $[\mathbf{\boldsymbol{\xi}}]$ and $[\boldsymbol{\zeta}]$ are energy basis vectors,  $\ell(\mathbf{y},\mathbf{\boldsymbol{\xi}})$ is the inverse transformation, and $\mathcal{F}$ involves a double sum over power sets of $\{1,2,\ldots, N\}$.
\end{theorem}

This formula generalizes the gap probability function in \cite{saenz_domain_2022} for the case of the Heisenberg-Ising spin-1/2 XXZ chain on the integer line. Moreover, for the domain wall (i.e.~step) initial conditions and zero anisotropy (i.e.~$\Delta = 0$), asymptotic analysis of the gap probability formula led to Tracy-Widom fluctuations for the location of the edge of the up-spin profile. The result in \Cref{t:one-point function main} may provide an avenue for a precise asymptotic analysis of the spin-up profile of the periodic XXZ for the initial out-of-equilibrium conditions.

\subsection{Overview}

We give a precise definition of the \emph{Heisenberg-Ising (XXZ) spin-1/2 chain} in \Cref{s:model}. We define the Hilbert space, the Hamiltonian and the dynamics of the model. In the following section, \Cref{s:bethe}, we introduce the Bethe Ansatz and provide our main results and conjectures. 

We show \Cref{c:main_v2} to hold true for $\Delta =0$ in \Cref{s:complete}, and provide numerical evidence for $\Delta \neq 0$. We introduce a novel formula to decompose the \emph{configuration basis}, the natural basis of the Hilbert space given by the different configurations of the system, as a linear combination of the Bethe vectors; see \Cref{c:main_v3}. Along with \Cref{p:cardinality}, \Cref{c:main_v3} leads to \Cref{c:main_v2}. We provide a proof of \Cref{c:main_v3} for $\Delta = 0$ in \Cref{s:linear_independence} right after its statement. Thus, we establish \Cref{c:main_v2} for $\Delta = 0$. In \Cref{s: N=2 proof} we give a proof of \Cref{c:main_v2} for $N=2$ under the hypotheses stated in \Cref{a:N=2 proof assumptions}.

The rest of the paper is dedicated to numerical verification and applications of \Cref{c:main_v3}. In \Cref{s:functions}, we give a formula for the wave function that solves the Schr\"odinger equation for the XXZ spin-1/2 chain with deterministic initial conditions as a linear combination of Bethe vectors. Then, we give formulas for the probability functions for measuring a specific configuration, i.e.~the joint probability function, and for measuring spin-up particle at a given site, i.e.~the one-point function. We also provide a method to numerically solve the Bethe equations for small $\Delta$. In \Cref{s:numerical}, we present numerical verification of \Cref{c:main_v3} and plot the one-point function introduced in the previous section. In \Cref{s:one-point function} we re-state and prove \Cref{t:one-point function main} with new and old special identities for the Izergin-Korepin determinant.

\section{The Heisenberg-Ising XXZ spin-1/2 chain}\label{s:model}

We consider the \emph{Heisenberg-Ising XXZ spin-$\tfrac{1}{2}$} chain on a periodic lattice. The length of the lattice is given by a positive integer $L \in \mathbb{Z}_{>0}$. Each site of the lattice is occupied by a \emph{spin-up} particle, $\mid\uparrow\rangle$, or a \emph{spin-down} particle, $\mid \downarrow \rangle$. As the system evolves, the labels of the sites are updated due to local, pairwise interactions. Additionally, due to the dynamics of the model, the number of up-spin labels are conserved as time evolves. We let $N$ denote the number of up-spin particles in the system. We give the details of the model in the following sections.

\subsection{Notation}

We use the following notation throughout the paper:

\begin{itemize}
    \item Given $L \in \mathbb{Z}_{>0}$, we denote the possible locations on the ring as follows
    \begin{equation}
    [L] := \{0, 1, 2 , \dots, L-1 \}
    \end{equation}

    \item Given $N, L \in \mathbb{Z}_{>0}$, $L\geq N$, we denote and index the solutions to the equation $\eta^L = (-1)^{N-1}$ as follows
    \begin{equation}\label{e:roots_special}
        \eta(k) = \begin{cases}
            e^{\pi i (2 k +1)/L}, &\quad N \text{ is even}\\
            e^{2 \pi i k/L}, &\quad N \text{ is odd}
        \end{cases}
    \end{equation}
    for $k \in [L]$.

    \item Given $N\in \mathbb{Z}_{>0}$, we denote the group of permutations on the set $\{1, 2, \dots, N\}$ by $S_N$ so that $\sigma \in S_N$ is a bijection on $\{1, 2, \dots, N\}$ and composition is the group action.
    
\end{itemize}

\subsection{Configuration, states, and probability}\label{s:configuration}

We consider the evolution of $N$ spin-up particles on a lattice ring of length $L$. A configuration is given by the location of the $N$ spin-up particles $\mathbf{x} = (x_1, x_2, \dots, x_N)$ so that $x_i \in [L]$ and $x_i < x_{i+1}$. We denote the space of configurations by $\mathcal{X}$, which is given by
\begin{equation}\label{e:coordinate_config}
    \mathcal{X} = \{(x_1 < x_2 < \cdots < x_N) \mid x_i \in [L], \, i=1, 2, \dots, N \}.
\end{equation}
The evolution of the particle is given by the Schr\"odinger equation from quantum mechanics. Thus, we take the set of configurations $\mathcal{X}$ to be an orthonormal basis set. For each $\mathbf{x} \in \mathcal{X}$, we have a basis vector $| \mathbf{x} \rangle$ and we introduce the vector space generated by these vectors
\begin{equation}\label{e:config_space}
    \mathbb{H}:= \mathrm{Span}\{|\mathbf{x}\rangle \mid \mathbf{x} \in \mathcal{X}\}.
\end{equation}

We introduce an inner product structure for vector space $\mathbb{H}$,
\begin{equation}
    \langle \cdot | \cdot \rangle: \mathbb{H}^2 \rightarrow \mathbb{C},
\end{equation}
by taking the coordinate basis to be orthonormal. We set 
\begin{equation}\label{e:lp_rel1}
    \langle \mathbf{y} | \mathbf{x} \rangle : = \delta_{\mathbf{y}, \mathbf{x}}
\end{equation}
for any $\mathbf{x}, \mathbf{y} \in \mathcal{X}$. We extend the inner product linearly on the second argument to follow the typical bra-ket convention,
\begin{equation}\label{e:lp_rel2}
    \langle \mathbf{y} | \alpha \,\mathbf{x}_1 + \beta \, \mathbf{x}_2 \rangle = \alpha\langle \mathbf{y} | \mathbf{x}_1 \rangle +\beta \langle \mathbf{y} | \mathbf{x}_2 \rangle,
\end{equation}
for any $\mathbf{x}_1, \mathbf{x}_2, \mathbf{y} \in \mathcal{X}$ and $\alpha, \beta \in \mathbb{C}$ where the bar denotes the complex conjugate of a complex number. Lastly, we make the inner product conjugate symmetric,
\begin{equation}\label{e:lp_rel3}
    \langle \mathbf{y} | \mathbf{x} \rangle = \overline{\langle \mathbf{x} | \mathbf{y} \rangle},
\end{equation}
for any $|\mathbf{x}\rangle , |\mathbf{y}\rangle \in \mathbb{H}$. The properties \eqref{e:lp_rel1} -- \eqref{e:lp_rel3} completely determine the inner product $\langle \cdot | \cdot \rangle$. 

The vector space $\mathbb{H}$ with the inner product $\langle \cdot | \cdot \rangle$ is a Hilbert space. A vector $ | \Psi \rangle \in \mathbb{H}$ in the Hilbert space is called a \emph{state (of the configuration space)} if it is normalized $\langle \Psi | \Psi \rangle =1$. Otherwise, a non-zero vector is called an unnormalized state. States will remain normalized under the dynamics described below, and we will only consider normalized states. Given a state $|\Psi \rangle$, the probability that the state is \emph{measured} in configuration $\mathbf{x} \in \mathcal{X}$ is given by the \emph{Born rule} from the postulates of quantum mechanics
\begin{equation}\label{e:prob}
    \mathbb{P}(|\Psi \rangle = | \mathbf{x} \rangle) = \langle \Psi | \mathbf{x} \rangle \langle \mathbf{x} | \Psi \rangle = \vert \langle \mathbf{x} | \Psi \rangle \vert^2.
\end{equation}

\subsection{Spin configuration and local operators}

The evolution of the system is generated by the action of local Hamiltonian operators. We define the local Hamiltonian operators in an auxiliary vector space, given by the spin configurations. We identify the vector space $\mathbb{H}$, given by the particle configurations, as a subspace of the auxiliary vector space. Then, we define the local Hamiltonian operators on $\mathbb{H}$ by restricting the local Hamiltonian operators on the auxiliary vector space.

Consider a two dimensional complex vector space with a basis given by spin-up and spin-down labels,
\begin{equation}
    \mathbb{V} = \mathrm{Span}\{ \mid \uparrow \rangle, \mid \downarrow \rangle \}.
\end{equation}
Let $\mathbb{V}_{L} = \mathbb{V}^{\otimes L}$ be the $L$-fold tensor product of $\mathbb{V}$. A basis for $\mathbb{V}_{L}$ is given by taking $L$ tensors of basis elements of $\mathbb{V}$. For instance, the basis of $\mathbb{V}_2$ is given and denoted as follows
\begin{equation}
    \mid \uparrow \rangle \otimes \mid \uparrow \rangle = \mid \uparrow \uparrow \rangle,\, \mid \uparrow \rangle \otimes \mid \downarrow \rangle = \mid \uparrow \downarrow \rangle,\, \mid \downarrow \rangle \otimes \mid \uparrow \rangle = \mid \downarrow \uparrow \rangle,\, \mid \downarrow \rangle \otimes \mid \downarrow \rangle = \mid \downarrow \downarrow \rangle.
\end{equation}
We identify $\mathbb{H}$, given by \eqref{e:config_space}, as a subspace of $\mathbb{V}_L$. Each basis element $|X \rangle \in \mathbb{H}$ is identified with a basis element of $\mathbb{V}_L$ where the position of the up-spins are given by $\mathbf{x}= (x_1, \dots, x_N) \in \mathcal{X}$. For instance, we have the identification
\begin{equation}
    \mid \uparrow \downarrow \rangle = | 1 \rangle,\,  \mid \downarrow \uparrow \rangle = | 2 \rangle
\end{equation}
when $N=1$ and $L=2$. It is clear that this identification is unique and well-defined. Thus, we have $\mathbb{H} \subseteq \mathbb{V}_L$.

The local Hamiltonian operators are given by tensor products of the \emph{Pauli matrices} on $\mathbb{V}$,
\begin{equation}
    \sigma^x = \left(\begin{array}{cc}
         0& 1 \\
         1& 0
    \end{array}\right),\quad
    \sigma^y = \left(\begin{array}{cc}
         0& -i \\
         i& 0
    \end{array}\right),\quad
    \sigma^z = \left(\begin{array}{cc}
         1& 0 \\
         0& -1
    \end{array}\right),
\end{equation}
with respect to the basis given by $\mid \uparrow \rangle$ and $\mid \downarrow \rangle$. Then, we introduce the following operators on $\mathbb{V}_L$
\begin{equation}
    S^{\alpha}_iS^{\alpha}_j = \frac{1}{4}\left(\mathrm{Id} \otimes\cdots \otimes \mathrm{Id} \otimes \sigma^{\alpha} \otimes \mathrm{Id} \otimes\cdots \otimes \mathrm{Id} \otimes \sigma^{\alpha} \otimes \mathrm{Id} \otimes\cdots \otimes \mathrm{Id}\right)
\end{equation}
where $\mathrm{Id}$ is the identity operator on $\mathbb{V}$ and the location of the operators $\sigma^{\alpha}$ is given by $i,j \in [L]$, respectively, with $i \neq j$. In particular, the operator $\frac12\left(\sigma^x \otimes \sigma^x+ \sigma^y \otimes \sigma^y +\Delta( \sigma^z\otimes \sigma^z - \mathrm{Id})\right)$, with $\Delta \in \mathbb{R}$, has the following representation
\begin{equation} \label{e: local hamiltonian matrix}
    \left( \begin{array}{cccc}
         0 & 0 & 0 &0  \\
         0 & -\Delta & 1 &0  \\
         0 & 1 & -\Delta &0  \\
         0 & 0 & 0 &0  \\
    \end{array}\right)
\end{equation}
with respect to the basis $\mid \uparrow \uparrow \rangle, \mid \uparrow \downarrow \rangle, \mid \downarrow \uparrow \rangle, \mid \downarrow \downarrow \rangle$. Note that these operators conserve the number of up-spins. Moreover, it follows that we may restrict the following \emph{local Hamiltonian operators}
\begin{equation}\label{e:local_operator}
    h_{i,j}(\Delta) = 2S^x_{i}S^x_{j}+ 2S^y_{i}S^y_{j} +2\Delta (S^z_{i}S^z_{j}-1/4),
\end{equation}
for $i, j \in [L]$ with $i \neq j$, to the vector space $\mathbb{H}$ since it maps $\mathbb{H}$ to itself.

\subsection{Dynamics}

The evolution of the system is given by a Schr\"odinger equation. The Hamiltonian for the system is given by a sum of local operators
\begin{equation}\label{e:hamiltonian}
    H_{\mathrm{XXZ}} = \sum_{i=0}^{L-1} h_{i,i+1}(\Delta), \quad \Delta \in \mathbb{R}
\end{equation}
where $h_{i,i+1}(\Delta)$ is given by \eqref{e:local_operator} and the indices are taken modulo $L$ so that the indexes $i =L$ and $i=0$ are identified. The \emph{anisotropy} parameter $\Delta$ is required to be real so that the operator is Hermitian, a sufficient condition to have real eigenvalues. Then, we consider the evolution of the system given by the following Schr\"odinger equation with deterministic initial conditions
\begin{equation}\label{e:schrodinger}
    \begin{cases}
    i \frac{d}{dt} |\Psi(t) \rangle = H_{\mathrm{XXZ}} | \Psi(t) \rangle,\\
    | \Psi(0) \rangle = | \mathbf{y} \rangle
    \end{cases}
\end{equation}
with $|\Psi(t)\rangle \in \mathbb{H}$ and $\mathbf{y} \in \mathcal{X}$.

\section{Bethe Ansatz}\label{s:bethe}

The Bethe Ansatz is a method to construct a collection of eigenvectors, called the \emph{Bethe eigenvectors}, for certain generators of one-dimensional interacting particle systems with enough symmetries -- we only consider the XXZ spin-1/2 chain. In the periodic case, each Bethe eigenvector depends on a solution for a system of algebraic equations, called the \emph{Bethe equations}. A priori, it is not clear if the Bethe Ansatz produces all the eigenvectors for the XXZ spin-1/2 chain, i.e.~if the Bethe Ansatz is \emph{complete}. In the following sections, we will describe the Bethe Ansatz in more detail. Additionally, we provide \Cref{c:main_v3} that if true, along with some other existing results, implies that the Bethe Ansatz is complete.  

\subsection{Bethe vectors}
We introduce a collection of vectors, called \emph{Bethe eigenvectors}, given as follows
\begin{equation}\label{e:bethe_vector}
    | \mathbf{\boldsymbol{\xi}} \rangle = \sum_{\mathbf{x} \in \mathcal{X}} u(\mathbf{\boldsymbol{\xi}}, \mathbf{x}) | \mathbf{x} \rangle
\end{equation}
where $\mathbf{\boldsymbol{\xi}} = (\xi_1, \dots, \xi_N) \in \mathbb{C}^N$ is a solution to the system of equations, called the \emph{Bethe equations}
\begin{equation}\label{e:bethe_equation}
    \xi_i^L = (-1)^{N-1} \prod_{j=1}^N \frac{1 + \xi_i \xi_j - 2 \Delta \xi_i}{1 + \xi_i \xi_j - 2 \Delta \xi_j}, \quad i =1, 2, \dots, N
\end{equation}
and the coefficients are given as follows
\begin{equation}\label{e:coefficients}
    u(\mathbf{\boldsymbol{\xi}}, \mathbf{x}) = \sum_{\sigma \in S_N} A_{\sigma}(\mathbf{\boldsymbol{\xi}}) \prod_{i=1}^N \xi_{\sigma(i)}^{x_i}
\end{equation}
with the \emph{amplitudes}
\begin{equation}\label{e:amplitude}
    A_{\sigma}(\mathbf{\boldsymbol{\xi}}) = \prod_{\substack{i < j \\ \sigma(i) > \sigma(j)}} \left(- \frac{1 + \xi_{\sigma(i)} \xi_{\sigma(j)} - 2 \Delta \xi_{\sigma(i)}}{1 + \xi_{\sigma(i)} \xi_{\sigma(j)} -2  \Delta \xi_{\sigma(j)}} \right).
\end{equation}
In particular, a Bethe vector $| \mathbf{\boldsymbol{\xi}} \rangle$ is determined and labeled by a solution $\mathbf{\boldsymbol{\xi}} \in \mathbb{C}^N$ to the Bethe equations given by \eqref{e:bethe_equation}. We will consider solutions of the Bethe equations so that the Bethe vectors are well-defined.

\begin{assumption}\label{a:generic}
Take $\mathbf{\boldsymbol{\xi}} = (\xi_1, \dots, \xi_N) \in \mathbb{C}^N$ so that it is a solution of the Bethe equations \eqref{e:bethe_equation} and
\begin{equation}\label{e:generic}
    1 + \xi_i \xi_j - 2 \Delta \xi_i \neq 0
\end{equation}
for any $i, j \in \{1, 2,\dots, N\}$.
\end{assumption}

\begin{remark}
    \Cref{a:generic} is a generic assumption. Take the Bethe equations and at least one pair $(i,j)$ so that the left side of \eqref{e:generic} is equal to zero. This is an overdetermined system of equations. Then, for fixed $N$ and $L$, it suffices to find one value of $\Delta$ where \Cref{a:generic} is true to show that it is a generic condition. For $\Delta =0$, the Bethe equations and the negation of \Cref{a:generic} become
    \begin{equation}
        \begin{cases}
            \xi_k^L = (-1)^{N-1}, &\quad k=1, 2, \dots, N\\
            1 + \xi_i \xi_j = 0,
        \end{cases}
    \end{equation}
    for some $i, j \in \{1, \dots, N\}$. This system of equations doesn't have a solution for $L$ an odd number but it does have some solutions for $L$ an even number. Thus, it's not clear that \Cref{a:generic} is generically true for $L$ an even number but we do know that \Cref{a:generic} is generically true for $L$ an odd number.  
\end{remark}

\begin{proposition}\label{p:eigenvector}
If $\mathbf{\boldsymbol{\xi}} \in \mathbb{C}^N$ is a solution of the Bethe equations \eqref{e:bethe_equation} so that \Cref{a:generic} is true and $| \mathbf{\boldsymbol{\xi}} \rangle \neq0$, then $| \mathbf{\boldsymbol{\xi}} \rangle$ is an eigenvector of the Hamiltonian for the XXZ spin-$\tfrac{1}{2}$ chain with eigenvalue
\begin{equation}\label{e:eigenvalue}
    E(\mathbf{\boldsymbol{\xi}}) = \sum_{i=1}^N \left(\xi_i + \xi_i^{-1} - 2\Delta \right)
\end{equation}
\end{proposition}

This is a well-known result dating back to the work of Hans Bethe in 1931 \cite{bethe_zur_1931}; see \cite{caux_bethe_2014} for a modern treatment and \cite{gwa_bethe_1992} in the context of a 1D interacting particle Markov process. We sketch the idea of the proof. We first extend the configuration space and the corresponding vector space. Let
\begin{equation}
    \tilde{\mathcal{X}} = \mathbb{Z}^{N}
\end{equation}
be the space of configurations where particles may occupy the same position in space and may be in any order. Then, we have a natural inclusion of the configuration spaces
\begin{equation}
    \mathcal{X} \hookrightarrow \tilde{\mathcal{X}}.
\end{equation}
We also introduce the corresponding Hilbert space
\begin{equation}
    \tilde{\mathbb{H}} = \mathrm{Span} \{|\mathbf{x} \rangle \mid \mathbf{x} \in \tilde{\mathcal{X}} \}
\end{equation}
with the same inner product structure as for $\mathbb{H}$, defined in \Cref{s:configuration}. Then, we have a natural projection map
\begin{equation}
    P: \tilde{\mathbb{H}} \rightarrow \mathbb{H}
\end{equation}
induced from the inclusion map $\mathcal{X} \hookrightarrow \tilde{\mathcal{X}}$. Lastly, we introduce a Hamiltonian $H: \tilde{\mathbb{H}} \rightarrow \tilde{\mathbb{H}}$ given as follows
\begin{equation}
    H: |\mathbf{x} \rangle \mapsto \sum_{i =1}^N \bigg(| x_i +1\rangle + | x_i -1 \rangle - 2 \Delta |\mathbf{x} \rangle\bigg)
\end{equation}
where $\mathbf{x}= (x_1, x_2, \dots, x_N)$ and $|x_i \pm 1 \rangle$ is the basis vector corresponding to the same coordinates as $\mathbf{x}$ except that the $x_i$-coordinate is replaced by $x_i \pm 1$. This Hamiltonian corresponds to the spin dynamics of the Hamiltonian $H_{XXZ}$ without interactions.

\begin{lemma}\label{l:bethe_conditions}
Take $|\tilde{V} \rangle \in \tilde{\mathbb{H}}$ and $|V \rangle \in \mathbb{H}$. Assume the following conditions
\begin{equation}\label{e:bethe_conditions}
    (i)\, H |\tilde{V} \rangle = E |\tilde{V} \rangle, \quad (ii)\, P |\tilde{V} \rangle = |V \rangle, \quad (iii)\, H_{XXZ} | V \rangle = P H | \tilde{V} \rangle , \quad (iv)\, | V\rangle \neq 0
\end{equation}
for some $E\in \mathbb{C}$. Then, $|V \rangle \in \mathbb{H}$ is an eigenvector of the XXZ Hamiltonian,
\begin{equation}
    H_{XXZ} |V\rangle = E | V \rangle
\end{equation}
where $E \in \mathbb{C}$ as in \eqref{e:bethe_conditions}.
\end{lemma}

\begin{proof}
This is a straightforward result. Taking the first three conditions, we have
\begin{equation}
    H_{XXZ} | V \rangle = P H | \tilde{V} \rangle = P E | \tilde{V} \rangle = E P | \tilde{V} \rangle = E | V \rangle.
\end{equation}
The last condition assures that the vector $|V \rangle$ is not trivial. Thus, the result follows.
\end{proof}

\Cref{l:bethe_conditions} is the key ingredient for the proof of \Cref{p:eigenvector}. Finding eigenvectors for the operator $H$ is simple since the particles don't interact. The other conditions in \eqref{e:bethe_conditions}, except condition (iii), are straightforward to check. Then, the problem of finding eigenvectors for the Hamiltonian $H_{XXZ}$ may be shifted to condition (iii) in \eqref{e:bethe_conditions}. Note that, for now, this is just a convenient method to find eigenvectors for the XXZ Hamiltonian. 

\begin{proof} [Proof of \Cref{p:eigenvector}]
We claim that the Bethe vector $|\mathbf{\boldsymbol{\xi}} \rangle$ satisfies the conditions in \Cref{l:bethe_conditions} which is sufficient. The last condition is satisfied by assumption. We check the rest of the conditions below.

Suppose  $|\tilde{V} \rangle \in \tilde{\mathbb{H}}$ has the form
\begin{equation}
    |\tilde{V} \rangle = \sum_{\mathbf{x} \in \tilde{\mathcal{X}}} u(\mathbf{\boldsymbol{\xi}}, \mathbf{x}) | \mathbf{x} \rangle
\end{equation}
so that $\mathbf{\boldsymbol{\xi}} = (\xi_1 , \xi_2, \dots, \xi_N) \in \mathbb{C}^N$ and
\begin{equation}
    u(\mathbf{\boldsymbol{\xi}}, \mathbf{x}) = \sum_{\sigma \in S_N} A_{\sigma}(\mathbf{\boldsymbol{\xi}}) \prod_{i=1}^N \xi_{\sigma(i)}^{x_i}
\end{equation}
for $\mathbf{x} = (x_1 , x_2, \dots, x_N)$ and some amplitudes $A_{\sigma} \in \mathbb{C}$ to be determined. One can determine the eigenvalue to be
\begin{equation}
    E = \sum_{i=1}^N \left(\xi_i + \xi_i^{-1} -2 \Delta \right)
\end{equation}
by noticing that the eigenvalue equation $H|\tilde{V}\rangle=E|\tilde{V}\rangle$ becomes a sum of difference equations in the $x_i$ variables. 
Note that $u(\mathbf{\boldsymbol{\xi}}, \cdot): \tilde{\mathcal{X}} =\mathbb{Z}^N \rightarrow \mathbb{C}$, for fixed $\boldsymbol{\xi} \in \mathbb{C}^N$, is a function mapping from the set of integer $N$-tuples to the complex numbers.

The third condition in \eqref{e:bethe_conditions} is now true if and only if the following relations\footnote{Obtaining these relations may not be immediately clear. It becomes clear by treating the cases for $N=2,3$ by hand. One also notes that the third condition leads to a large number of relations, but all of these relations reduce to those given in \eqref{e:u_relations}. This is a special symmetry of the current model that makes it \emph{exactly solvable}.} 
\begin{equation}\label{e:u_relations}
    \begin{split}
    &(i)\, 2 \Delta u(x_k =x, x_{k+1} = x+1)=\\
    &\hspace{20mm} u(x_k =x, x_{k+1} = x) + u(x_{k}=x+1, x_{k+1} = x+1) \\
    &(ii)\, u(x_1, \dots, x_{N-1}, x_N) = u(x_{N}-L, x_1, \dots, x_{N-1}).
    \end{split}
\end{equation}
for $\mathbf{x} = (x_1, \dots, x_N) \in \tilde{\mathcal{X}}$ and $k=1, 2, \dots, N-1$ hold. These two relations determine the amplitude coefficients \eqref{e:amplitude} and the Bethe equations \eqref{e:bethe_equation}. 

The first relation in \eqref{e:u_relations} is a condition on the $u$-function for configuration where consecutive coordinates are equal or adjacent. This is where the action of the operators $H$ and $H_{XXZ}$ differ, and the first relation is forcing the operators to agree on the candidate eigenvector $|V\rangle$. The first relation leads to the following relation for the amplitude coefficients\footnote{This is where we use \Cref{a:generic}.}
\begin{equation}\label{e:amp_relation}
    A_{\sigma \circ \tau} = A_{\sigma}\left( - \frac{1 + \xi_{\sigma(k)} \xi_{\sigma(k+1)} - 2 \Delta \xi_{\sigma(k+1)}}{1 + \xi_{\sigma(k)} \xi_{\sigma(k+1)} - 2 \Delta \xi_{\sigma(k)}}\right)
\end{equation}
for any $\sigma \in S_N$ and the transposition $\tau = (k\, k+1)$. Note that this relation determines all the values of the amplitudes, up to the value of $A_{\mathrm{id}}$, since the symmetric group is generated by transpositions. The amplitude $A_{\mathrm{id}} \neq 0$. Otherwise, $A_{\sigma} = 0$ for all $\sigma \in S_N$, by \eqref{e:amp_relation} and \Cref{a:generic}, meaning that $|V \rangle = 0$ is a trivial eigenvector. Then, we may take $A_{\mathrm{id}} =1$ since eigenvectors are determined up to a multiplicative constant. Moreover, we may directly check that the amplitude coefficients $A_{\sigma}$ in \eqref{e:amplitude} satisfy the relation \eqref{e:amp_relation} and $A_{\mathrm{id}} =1$, making them the unique choice for the amplitude coefficients.

The second relation in \eqref{e:u_relations} is a periodicity condition for the $u$-function. This relation directly leads to the Bethe equation \eqref{e:bethe_equation}.

We have now completely determined, through conditions (i) and (iii) in \eqref{e:bethe_conditions}, that
\begin{equation}
    |\tilde{V} \rangle = \sum_{\mathbf{x} \in \tilde{\mathcal{X}}} u(\boldsymbol{\xi}, \mathbf{x}) |\mathbf{x} \rangle
\end{equation}
with the $u$-function given by \eqref{e:coefficients} and $\boldsymbol{\xi} = (\xi_1, \dots, \xi_N) \in \mathbb{C}^N$ satisfying the Bethe equations \eqref{e:bethe_equation}. Lastly, we have
\begin{equation}
    |\boldsymbol{\xi} \rangle = |V \rangle = P |\tilde{V} \rangle = \sum_{\mathbf{x} \in \mathcal{X}} u(\boldsymbol{\xi}, \mathbf{x}) |\mathbf{x} \rangle.
\end{equation}
satisfying condition (ii) in \eqref{e:bethe_conditions}. This establishes the result. 
\end{proof}

\begin{remark}
Note the following subtleties with \Cref{p:eigenvector}: (1) it is not a priori guaranteed that all eigenvectors for the XXZ spin-$\tfrac{1}{2}$ chain are Bethe vectors, and (2) it is not a priori clear which solutions of the Bethe equations lead to non-trivial Bethe vectors (i.e.~non-zero vectors). In particular, the first point means that we don't know if the set of Bethe vectors span the entire Hilbert space $\mathbb{H}$. 
\end{remark}

\subsection{Bethe space}

We construct a vector space generated by the Bethe vectors \eqref{e:bethe_vector} corresponding to the solutions of the Bethe equations \eqref{e:bethe_equation}. We are only interested in solutions to the Bethe equations with pair-wise distinct entries, due to upcoming technical reasons. Note that any permutation of the entries of a solution to the Bethe equations is also a solution to the Bethe equations; i.e.~if $\boldsymbol{\xi} = (\xi_1, \dots, \xi_N)$ is a solution of the Bethe equations, then 
\begin{equation}\label{e:sym_action}
    \sigma \cdot \boldsymbol{\xi} = (\xi_{\sigma(1)}, \dots, \xi_{\sigma(N)}),
\end{equation}
is also a solution of the Bethe equations for any permutation $\sigma \in S_N$. Moreover, we may show that two solutions of the Bethe equations that are equal up to permutation of their components lead to the same Bethe vector.

\begin{lemma}\label{l:amps_rel}
Take any $\boldsymbol{\xi} = (\xi_1, \dots, \xi_N) \in \mathbb{C}^N$ so that \Cref{a:generic} is true. Let $\tau = (k \, k+1) \in S_N$ be a transposition. Then, the amplitudes given by \eqref{e:amplitude} satisfy the relation
\begin{equation}
    A_{\sigma} (\tau \cdot \boldsymbol{\xi}) =  \left(-\frac{1 + \xi_k \xi_{k+1} - 2 \Delta \xi_k}{1 + \xi_k \xi_{k+1} - 2 \Delta \xi_{k+1}}\right) A_{\tau \circ \sigma}(\boldsymbol{\xi}) 
\end{equation}
for any $\sigma \in S_N$ and $\tau \cdot \boldsymbol{\xi}$ given by \eqref{e:sym_action}.
\end{lemma}

\begin{proof}
    The result follows from a direct computation:
    \begin{equation}
        \begin{split}
        A_{\sigma} (\tau \cdot \boldsymbol{\xi}) &= \prod_{\substack{i < j \\ \sigma(i)> \sigma(j)}}\left(- \frac{1 + \xi_{\tau \circ \sigma (i)}\xi_{\tau \circ \sigma (j)} - 2 \Delta \xi_{\tau \circ \sigma (i)}}{1 + \xi_{\tau \circ \sigma (i)}\xi_{\tau \circ \sigma (j)} - 2 \Delta \xi_{\tau \circ \sigma (j)}} \right)\\
        &= \left(-\frac{1 + \xi_k \xi_{k+1} - 2 \Delta \xi_k}{1 + \xi_k \xi_{k+1} - 2 \Delta \xi_{k+1}}\right) \prod_{\substack{i < j \\ \sigma(i)> \sigma(j)}}\left(- \frac{1 + \xi_{ \sigma (i)}\xi_{ \sigma (j)} - 2 \Delta \xi_{ \sigma (i)}}{1 + \xi_{ \sigma (i)}\xi_{ \sigma (j)} - 2 \Delta \xi_{ \sigma (j)}} \right)\\
        &= \left(-\frac{1 + \xi_k \xi_{k+1} - 2 \Delta \xi_k}{1 + \xi_k \xi_{k+1} - 2 \Delta \xi_{k+1}}\right) A_{\tau \circ \sigma}(\boldsymbol{\xi}).
        \end{split}
    \end{equation}
    The transposition $\tau = (k \, k+1)$ only affects the order of the elements $k$ and $k+1$. So, all the terms remain the same except for one term, which is the additional prefactor term.
\end{proof}

\begin{lemma}\label{l:well_def}
Take any $\boldsymbol{\xi} \in \mathbb{C}^N$, so that \Cref{a:generic} is true, and $\mathbf{x}\in \mathbb{Z}^N$. Then, we have
    \begin{equation}
    u(\boldsymbol{\xi}, \mathbf{x}) = \lambda u(\tau \cdot \boldsymbol{\xi}, \mathbf{x}),
    \end{equation}
for any $\tau \in S_N$, where $u$ is given by \eqref{e:coefficients}, $\tau \cdot \boldsymbol{\xi}$ is given by \eqref{e:sym_action}, and $\lambda \in \mathbb{C}$ with $\lambda \neq 0$.
\end{lemma}

\begin{proof}
    This follows from re-indexing the summation in the definition of the $u$-function and applying \Cref{l:amps_rel}. Moreover, it suffices to take $\tau \in S_N$ to be a transposition since transpositions generate the symmetric group $S_N$. First, note that
    \begin{equation}
    A_{\tau \circ \sigma}(\boldsymbol{\xi}) = \lambda A_{\sigma}(\tau \cdot \boldsymbol{\xi}),
    \end{equation}
    with $\lambda \neq 0$ independent of $\sigma$, due to \Cref{l:amps_rel} and \Cref{a:generic}. Then,
    \begin{equation}
    \begin{split}
    u(\boldsymbol{\xi}, \mathbf{x}) &= \sum_{\nu \in S_N} A_{\nu}(\boldsymbol{\xi}) \prod_{i=1}^N \xi_{\nu(i)}^{x_i}\\
    &= \sum_{\sigma \in S_N} A_{\tau \circ \sigma}(\boldsymbol{\xi}) \prod_{i=1}^N \xi_{\tau \circ \sigma(i)}^{x_i}\\
    &= \sum_{\sigma' \in S_N} \lambda A_{\sigma'}(\tau \cdot \boldsymbol{\xi}) \prod_{i=1}^N (\tau\cdot\boldsymbol{\xi})_{\sigma'(i)}^{x_i}\\
    &= \lambda u(\tau \cdot \boldsymbol{\xi}, \mathbf{x}),
    \end{split}
    \end{equation}
    where we re-index $\nu = \tau \circ \sigma$ and sum over $\sigma$ in the second equality. Thus, we obtain the desired result. 
\end{proof}

\begin{corollary}\label{c:sym_bethe_vec}
Take any $\boldsymbol{\xi} \in \mathbb{C}^N$, so that \Cref{a:generic} is true, and $\mathbf{x}\in \mathbb{Z}^N$. Then, the Bethe vectors, given by \eqref{e:bethe_vector}, are equal up to a multiplicative constant
    \begin{equation}
    |\boldsymbol{\xi} \rangle = \lambda |\sigma \cdot \boldsymbol{\xi} \rangle,
    \end{equation}
for any $\sigma \in S_N$, where $\sigma \cdot \boldsymbol{\xi}$ is given by \eqref{e:sym_action}, and $\lambda \in \mathbb{C}$ with $\lambda \neq 0$.
\end{corollary}

\begin{proof}
Follows directly from \Cref{l:well_def}
\end{proof}

We take any two solutions to the Bethe equations to be equivalent if one is a permutation of the other, i.e.~if they are equal up to order. Additionally, we disregard Bethe vectors $|\boldsymbol{\xi}\rangle$ with any two coordinates in $\boldsymbol{\xi} \in \mathbb{C}^N$ equal. In that case, the Bethe vector is identically equal to zero due to the following result. 

\begin{lemma}
    Take $\boldsymbol{\xi} = (\xi_1, \dots, \xi_N) \in \mathbb{C}^N$, so that \Cref{a:generic} is true, and let $|\boldsymbol{\xi} \rangle$ be the vector given by \eqref{e:bethe_vector}. Then,
    \begin{equation}
        |\boldsymbol{\xi} \rangle = 0
    \end{equation}
    if $\xi_i = \xi_j$ for some $i \neq j$.
\end{lemma}

\begin{proof}
    We may assume that $\xi_1 = \xi_{2}$ due to \Cref{c:sym_bethe_vec}. Then, we have
    \begin{equation}
        |\boldsymbol{\xi} \rangle = - | \tau \cdot \boldsymbol{\xi} \rangle = - | \boldsymbol{\xi} \rangle,
    \end{equation}
    for $\tau = (1\, 2)$. The first equality is due to \Cref{c:sym_bethe_vec} with $\lambda = -1$ due to \Cref{l:amps_rel} and $\xi_1 = \xi_2$. The second identity follows from $\boldsymbol{\xi} = \tau \cdot \boldsymbol{\xi}$ since $\xi_1 = \xi_2$. Thus, $2 |\boldsymbol{\xi} \rangle =0$ and the result follows.
\end{proof}

Let $\Xi$ denote the set of solutions to the Bethe equations, up to order, with pair-wise distinct entries:
\begin{equation}\label{e:bethe_sols}
    \Xi = \{ \boldsymbol{\xi}  \mid \boldsymbol{\xi}\in \widetilde{\mathbb{C}^N} \text{ is a solution to \eqref{e:bethe_equation}}\}/ S_N,
\end{equation}
where $\widetilde{\mathbb{C}^N}$ is the set of $N$-tuples of complex numbers with pair-wise distinct entries so that any two solutions are identified by the action of the symmetric group $S_N$. For a solution $\boldsymbol{\xi} \in \widetilde{\mathbb{C}^N}$ of the Bethe equation, we denote its image on the set of solutions by $[\boldsymbol{\xi}] \in \Xi$ so that $[\boldsymbol{\xi}] = [\sigma \cdot \boldsymbol{\xi}]$ for any permutation $\sigma \in S_N$.

We introduce the vector space generated by the Bethe vectors
\begin{equation}\label{e:bethe_space}
    \mathbb{W} = \mathrm{Span}\{ | \boldsymbol{\xi} \rangle \mid [\boldsymbol{\xi}] \in\Xi\}
\end{equation}
where $|\boldsymbol{\xi} \rangle$ is given by \eqref{e:bethe_vector}. Note that $\mathbb{W}$ is well-defined due to \Cref{l:well_def}\footnote{In the following, we assume we have an arbitrary fixed choice of representative $\boldsymbol{\xi}$ for each $[\boldsymbol{\xi}] \in \Xi$.}. Then, we have $\mathbb{W} \subseteq \mathbb{H}$ and, in the coming sections, we aim to show that the vector spaces are equal to show that the Bethe vectors generate the same space generated by the coordinate basis.

\begin{proposition}\label{p:cardinality}
    For generic $\Delta$, the cardinality of the set $\Xi$, given by \eqref{e:bethe_sols}, is equal to $\binom{L}{N}$.
\end{proposition}

\begin{proof}
For $\Delta =0$, it is straightforward to check that the proposition above is true as the solutions to the Bethe Equations are roots of unity. If $\Delta \neq 0$, then the counting argument becomes more complicated. The proof of \Cref{p:cardinality} follows from the counting argument in \cite[Sec.~4.3]{brattain_completeness_2017}, where the  Bethe Ansatz for the asymmetric simple exclusion process (ASEP) is considered. In particular, our Bethe equations become
\begin{equation}\label{e:BE-transformed}
    \left(\frac{q}{p}\right)^{L/2}(\xi_i')^L = (-1)^{N-1} \prod_{j=1}^N \frac{p + q \xi_i' \xi_j' - \xi_i'}{p + q \xi_i' \xi_j' - \xi_j'}, \quad i=1, 2, \dots, N
\end{equation}
after the change of variables \begin{equation}
    \xi_i=\xi_i'\sqrt{\frac{q}{p}}, \quad 2\Delta=\frac{1}{\sqrt{pq}}.
\end{equation}
The system of equations \eqref{e:BE-transformed} is the same as the system of equations for the ASEP, except for a multiplicative factor on the left side of the equations. The counting argument relies on a generalization of the Lefschetz fixed point number for the coincidence number of two functions, which correspond to the functions on the left and right side of the equality in \eqref{e:BE-transformed}. Moreover, the counting of solutions remains the same for a homotopic deformation of the functions. Thus, the counting argument remains the same despite the difference in the multiplicative constant, since scaling by a non-zero multiplicative constant gives a homotopic function. The result follows.
\end{proof}

Note that, by \Cref{p:cardinality}, it follows that the set of Bethe vectors $|\boldsymbol{\xi} \rangle$, for $\boldsymbol{\xi} \in \Xi$, are linearly independent if and only if $\mathbb{H} = \mathbb{W}$ since $\dim \mathbb{H} = \binom{L}{N}$.

\subsection{Linear Independence}\label{s:linear_independence}

We claim that the set of Bethe vectors generated by $\Xi$, given by \eqref{e:bethe_sols}, are linearly independent. Due to \Cref{p:cardinality}, it suffices to show
\begin{equation}
    \dim \mathbb{W} = \binom{L}{N}.
\end{equation}
In particular, we show that $\mathbb{H} \subseteq \mathbb{W}$ since, by definition, we know $\dim \mathbb{H} = \binom{L}{N}$ and $\mathbb{W} \subseteq \mathbb{H}$. We do this by representing each coordinate basis in $|\mathbf{x} \rangle \in \mathbb{H}$, for $x \in \mathcal{X}$, as a linear combination of Bethe vectors $| \boldsymbol{\xi} \rangle \in \mathbb{W}$, with $[\boldsymbol{\xi}] \in \Xi$. In the case $\Delta =0$, we give a precise proof. In the case $\Delta\neq0$, we verify our claim numerically. In \Cref{s: N=2 proof}, we give a proof for the two particle case ($N=2$) for $|\Delta|<\frac{L-1}{2L}$ and additional conditions on $\Delta$ stated in \Cref{a:N=2 proof assumptions}. 

We introduce a map $\ell: \mathbb{Z}^N \times \mathbb{C}^N \rightarrow \mathbb{C}$. Let $\Lambda (\boldsymbol{\xi})$ be a $N\times N$-matrix \footnote{A similar matrix appears in \cite{korepin_calculation_1982},\cite{gaudin_normalization_1981} in relation to the norm of the Bethe vectors.} with entries given by 
\begin{equation}\label{e:lambda_mat}
\begin{split}
    [\Lambda(\boldsymbol{\xi})]_{i, j} 
    &= \partial_{\xi_i}\left[\ln\left(S(\xi_i,\xi_j) \right)\right]\\
    &= \frac{-2\Delta(1 + \xi_j^2 - 2\Delta \xi_j)}{(1 + \xi_i \xi_j - 2\Delta \xi_i)(1 + \xi_i \xi_j - 2\Delta \xi_j)}, \hspace{10mm} i \neq j \\
    [\Lambda(\boldsymbol{\xi})]_{i, i}&=\partial_{\xi_i}\left[\ln\left(\xi_i^L \right) - \ln\left(\prod_{k\neq i}S(\xi_i,\xi_k) \right) \right]\\
    &= \frac{L}{\xi_i} - \sum_{j \neq i} [\Lambda(\boldsymbol{\xi})]_{i, j}, \hspace{41mm} i=1, \dots, N,
\end{split}
\end{equation}
for any $\boldsymbol{\xi} = (\xi_1, \dots, \xi_N) \in \mathbb{C}^N$ and $S(\xi_i,\xi_j) \coloneq -\frac{1+\xi_i\xi_j-2\Delta\xi_i}{1+\xi_i\xi_j-2\Delta\xi_j}$. Then, we define the map
\begin{equation}
    \tilde{\ell}(\mathbf{x}, \boldsymbol{\xi}) := \det(\Lambda(\boldsymbol{\xi}))^{-1} \prod_{i=1}^{N} \xi_{i}^{-x_i-1}
\end{equation}
and its twisted-symmetrization on the $\boldsymbol{\xi}$-variable
\begin{equation}\label{e:l_coeff}
    \ell(\mathbf{x}, \boldsymbol{\xi}) := \sum_{\sigma \in S_N} \frac{\tilde{\ell}(\mathbf{x}, \sigma \cdot \boldsymbol{\xi} )}{A_{\sigma}(\boldsymbol{\xi})} = \sum_{\sigma \in S_N}\left( A_{\sigma}(\boldsymbol{\xi}) \det[\Lambda(\boldsymbol{\xi})] \prod_{i =1}^N \xi_{\sigma(i)}^{x_i +1} \right)^{-1}
\end{equation}
for any $\mathbf{x} = (x_1, \dots, x_N) \in \mathbb{Z}^N$ and any $\boldsymbol{\xi} = (\xi_1, \dots, \xi_N) \in \mathbb{C}^N$. Note that the $\ell$-map is well-defined on $\mathbb{Z}^N \times \Xi$ since it is symmetric on the second argument. Moving forward, we assume that $\det(\Lambda(\boldsymbol{\xi}))\neq0$. For $\Delta=0$, this assumption holds, since the $\xi_k$ are roots of unity. We expect this to hold for $\Delta\neq0$ and verify this numerically, but we do not have a precise argument.  

\begin{conjecture}\label{c:main_v3}
For any coordinate basis element $|\mathbf{x}\rangle \in \mathbb{H}$ with $\mathbf{x} \in \mathcal{X}$, given by \eqref{e:coordinate_config} and \eqref{e:config_space}, we have
    \begin{equation}\label{e:v_tranformation}
    | \mathbf{x} \rangle = \sum_{[\boldsymbol{\xi}] \in \Xi} \ell(\mathbf{x}, \boldsymbol{\xi}) | \boldsymbol{\xi} \rangle
    \end{equation}
where $|\boldsymbol{\xi} \rangle \in \mathbb{W}$, with $\boldsymbol{\xi} \in \Xi$, is given by \eqref{e:bethe_vector}. Moreover, the set of Bethe vectors $\{|\boldsymbol{\xi} \rangle \mid [\boldsymbol{\xi}] \in \Xi \}$ are linearly independent.
\end{conjecture}

The formula \eqref{e:v_tranformation}, for $\Delta \neq 0$, is new (to the best of the authors' knowledge). The formula was obtained through non-rigorous contour integral formulas, similar to the formulas from \cite{tracy_asymptotics_2009},\cite{liu_integral_2020},\cite{li_contour_2023}, and recognizing a pattern for the $N \leq 3$ cases. \Cref{c:main_v3} is equivalent to 
\begin{equation}\label{e:identity}
    \mathds{1}(\mathbf{x} = \mathbf{y}) = \sum_{[\boldsymbol{\xi}] \in \Xi}  \ell(\mathbf{y}, \boldsymbol{\xi}) u(\boldsymbol{\xi}, \mathbf{x})
\end{equation}
for all $\mathbf{x}, \mathbf{y} \in \mathcal{X}$ where the $u$-function is given by \eqref{e:coefficients} and $\mathds{1}(\mathbf{x} = \mathbf{y})$ is the indicator function\footnote{A formula for the indicator function based on Plancherel measures in the case of the XXZ on $\mathbb{Z}$ can be found in \cite{Gutkin2000-gr}.}. We check \eqref{e:identity} numerically using a \emph{Python} program for several $\Delta \neq 0$ cases; see \Cref{s:numerical}. We have a precise proof for $\Delta =0$ and for $N=2$ under the conditions of \Cref{a:N=2 proof assumptions} in \Cref{s: N=2 proof}.

\begin{lemma}\label{l:main_v3}
\Cref{c:main_v3} is true for $\Delta =0$.
\end{lemma}

\begin{proof}[Proof of \Cref{c:main_v3} for $\Delta =0$.]

Given $\Delta=0$, we have that
\begin{equation}
    u(\boldsymbol{\xi},\mathbf{x})=\sum_{\sigma\in S_N}\mathrm{sgn}(\sigma)\prod_{i=1}^N\xi_{\sigma(i)}^{x_i}, \quad \ell(\mathbf{y},\boldsymbol{\xi})=\sum_{\sigma\in S_N}\frac{1}{\mathrm{sgn}(\sigma)}\prod_{i=1}^N\frac{\xi_{\sigma(i)}^{-y_i}}{L}.
\end{equation}
Then,
\begin{equation}
\begin{split}
    \ell(\mathbf{y},\boldsymbol{\xi})u(\boldsymbol{\xi},\mathbf{x}) &= \sum_{\sigma, \mu\in S_N} \mathrm{sgn}(\sigma)\mathrm{sgn}(\mu)\prod_{i=1}^N\frac{\xi_{\mu(i)}^{x_i}\xi_{\sigma(i)}^{-y_i}}{L}\\
    &= \sum_{\sigma, \mu\in S_N} \mathrm{sgn}(\sigma)\mathrm{sgn}(\mu) \prod_{i=1}^N\frac{\xi_{i}^{x_{\mu^{-1}(i)}-y_{\sigma^{-1}(i)}}}{L}\\
    &= \sum_{\sigma, \mu\in S_N} \mathrm{sgn}(\sigma)\mathrm{sgn}(\mu) \prod_{i=1}^N \frac{\xi_{i}^{x_{\mu(i)}-y_{\sigma(i)}}}{L}.
\end{split}
\end{equation}

The solutions to the Bethe equations \eqref{e:bethe_equation} when $\Delta=0$ are the $\binom{L}{N}$ possible $N$-tuples of $L$'th roots of unity (up to permutation of the entries). That is, setting $\eta=e^{\frac{2\pi i}{L}}$ and $\phi=e^{\frac{\pi i}{L}}$ when $N$ is even and $\phi=1$ when $N$ is odd, we have that 
\begin{equation}
    \Xi=\left\{\phi\left(\eta^{k_1},\eta^{k_2},\ldots,\eta^{k_N}\right)=\boldsymbol{\xi} \mid \{k_1,\ldots,k_N\}\subset [L] \right\}.
\end{equation}
Note that $\ell(y,\boldsymbol{\xi})= u(\boldsymbol{\xi},x)=0$ when $\boldsymbol{\xi}$ has repeated entries. Then, when we sum over Bethe roots we may sum over all roots of unity\footnote{It may be curious to the reader that we extend our sum to include terms that we know are zero, but including these terms is crucial to obtain the finite geometric sums allowing for great simplification.} 
\begin{equation}
\begin{split}
\sum_{[\boldsymbol{\xi}] \in \Xi} \ell(\mathbf{y},\boldsymbol{\xi})u(\boldsymbol{\xi},\mathbf{x}) &= \frac{1}{N!}\sum_{k_1=1}^L \cdots \sum_{k_N=1}^L \sum_{\sigma, \mu\in S_N} \mathrm{sgn}(\sigma)\mathrm{sgn}(\mu) \prod_{i=1}^N\frac{\left(\phi\eta^{k_i}\right)^{(x_{\mu(i)}-\,y_{\sigma(i)})}}{L}\\
&= \frac{1}{N!} \sum_{\sigma, \mu \in S_N} \mathrm{sgn}(\sigma)\mathrm{sgn}(\mu)\prod_{i=1}^N\mathds{1}(x_{\mu(i)}=y_{\sigma(i)})\\
&=\mathds{1}(\mathbf{x} = \mathbf{y}).
\end{split}
\end{equation}
The factor $1/N!$ in the first equality is due to the fact that we only consider Bethe roots up to permutation of the coordinates. Periodicity is satisfied in the indicator function equality due to the periodicity of the $L$'th roots of unity. We obtain the second identity by taking the summations inside the product which gives us 
\begin{equation}
    \begin{split}
        \prod_{i=1}^N\phi^{(x_{\mu_{(i)}}-\,y_{\sigma_{(i)}})}\left(\sum_{k_i=1}^L\frac{\eta^{k_i(x_{\mu(i)}-\,y_{\sigma(i)})}}{L}\right)  &= \prod_{i=1}^N\mathds{1}(x_{\mu(i)}=y_{\sigma(i)}) 
    \end{split}
\end{equation}
where it is clear the finite geometric sum over $k_i$ is $1$ if $x_{\mu(i)}=y_{\sigma(i)}$ and the sum is $0$ otherwise since $\frac{1-\eta^{L(x_{\mu(i)}-y_{\sigma(i)})}}{1-\eta^{(x_{\mu(i)}-y_{\sigma(i)})}} =0$ because $\eta$ is an $L$'th root of unity.  
The last identity holds because: (1) the product is non-zero only if the set of $x$-coordinates is equal to the set of $y$-coordinates, (2) the $x$- and $y$-coordinates are ordered, and (3) if $\mathbf{x} = \mathbf{y}$, the product is equal to one if and only if $\sigma = \mu$.

\end{proof}

\subsection{Completeness of the Bethe Ansatz}\label{s:complete}

The Bethe Ansatz is said to be \emph{complete} if the Bethe vectors, $| \boldsymbol{\xi} \rangle \in \mathbb{W}$ with $[\boldsymbol{\xi}] \in \Xi$ given by \eqref{e:bethe_vector}, \eqref{e:bethe_space} and \eqref{e:bethe_sols} respectively,  
\begin{itemize}
    \item [(i)] are eigenvectors for the Hamiltonian $H_{\mathrm{XXZ}}$ given by \eqref{e:hamiltonian},
    \item [(ii)] generate the vector space $\mathbb{H}$ given by \eqref{e:config_space}, and
    \item [(iii)] are linearly independent.
\end{itemize}
We have shown that the Bethe vectors are eigenvectors of the Hamiltonian $H_{\mathrm{XXZ}}$ if the Bethe vector is not identically equal to zero; see \Cref{p:eigenvector}. We have introduced a collection of Bethe vectors $\Xi$ with pairwise-distinct entries and cardinality equal to $\dim \mathbb{H}$; see \Cref{p:cardinality}. We have also introduced a representation of the coordinate basis vectors $| \mathbf{x} \rangle \in \mathbb{H}$ as a linear combination of Bethe vectors; see \Cref{c:main_v3} and \eqref{e:v_tranformation}. These are sufficient to make the following claim.

\begin{proposition}
    Assume \Cref{c:main_v3} is true. Then, the Bethe Ansatz is complete. This means that the set of Bethe vectors $\{ |\boldsymbol{\xi}\rangle \mid [\boldsymbol{\xi}] \in \Xi \}$, given by \eqref{e:bethe_vector} and \eqref{e:bethe_sols} respectively, is a complete basis for the vector space $\mathbb{H}$, given by \eqref{e:config_space}, so that each Bethe vector is an eigenvector of the Hamiltonian $H_{XXZ}$, given by \eqref{e:hamiltonian}.
\end{proposition}

\begin{proof}
The result follows from the discussion before the statement of the result. To show that the set of Bethe vectors is a complete basis for the vector space $\mathbb{H}$, it suffices to show that the cardinality of the set of Bethe vectors is equal to the dimension of the vector space $\mathbb{H}$ and that the Bethe vectors generate the entire space $\mathbb{H}$. This follows from $\Cref{p:cardinality}$ and $\Cref{c:main_v3}$. The latter result we take to be true by assumption, but it is also verified numerically in \Cref{s:numerical} and proven for $N=2$ in \Cref{s: N=2 proof} under certain assumptions on the values of $\Delta$ stated in \Cref{a:N=2 proof assumptions}. Moreover, we know that the Bethe vectors are not identically equal to zero since we have just argued that the Bethe vectors give a complete basis for the vector space $\mathbb{H}$. Then, it follows that the Bethe vectors are eigenvectors of the Hamiltonian $H_{XXZ}$ by \Cref{p:eigenvector}. Thus, we have established the desired result. 
\end{proof}

\section{Computations and plots}\label{s:functions}

In the previous sections, we have provided an explicit construction of the completeness of the Bethe Ansatz for the periodic XXZ model. With this construction, namely the $\ell$ and $u$ functions, we can provide an exact solution for the Schr\"odinger equation \eqref{e:schrodinger} and compute observables of interest, like the one-point function. For the statements in previous sections that remain conjectured (like \eqref{c:main_v3}), we can provide numerical evidence for their validity.

\subsection{Solution to the Schr\"odinger equation}

\begin{proposition}\label{p:solution}
Assume \Cref{c:main_v3} is true. Then, the evolution of the Heisenberg-Ising (XXZ) spin-$\tfrac{1}{2}$ chain with deterministic initial conditions given by any basis vector $|\mathbf{y} \rangle$, i.e.~the solution of the initial value problem \eqref{e:schrodinger}, is given by
    \begin{equation}
    |\Psi(t) \rangle  = \sum_{\mathbf{x} \in \mathcal{X}} \left(\sum_{[\boldsymbol{\xi}] \in \Xi} \ell(\mathbf{y}, \boldsymbol{\xi})u(\boldsymbol{\xi}, \mathbf{x}) e^{-i t E(\boldsymbol{\xi})}\right) | \mathbf{x} \rangle
    \end{equation}
with the functions $\ell(\mathbf{y}, \boldsymbol{\xi})$ and $u(\boldsymbol{\xi},\mathbf{x})$ given by \eqref{e:l_coeff} and $\eqref{e:coefficients}$, respectively, and the function $E(\boldsymbol{\xi})$ given by \eqref{e:eigenvalue}.
\end{proposition}

\begin{proof}
First let us show that $|\Psi(t)\rangle$ satisfies the Schr\"odinger equation. We have
\begin{equation}
    \begin{split}
    i\frac{d}{dt}|\Psi(t)\rangle &= \sum_{\mathbf{x}\in \mathcal{X}}\left(\sum_{[\boldsymbol{\xi}]\in\Xi}E(\boldsymbol{\xi})\ell(\mathbf{y},\boldsymbol{\xi})u(\boldsymbol{\xi},\mathbf{x})e^{-itE(\boldsymbol{\xi})} \right)|\mathbf{x}\rangle \\
    &=\sum_{[\boldsymbol{\xi}]\in\Xi}\ell(\mathbf{y},\boldsymbol{\xi})e^{-itE(\boldsymbol{\xi})}E(\boldsymbol{\xi})\left(\sum_{x\in \mathcal{X}}u(\boldsymbol{\xi},\mathbf{x})|\mathbf{x}\rangle\right) \\ 
    &=\sum_{[\boldsymbol{\xi}]\in\Xi}\ell(\mathbf{y},\boldsymbol{\xi})e^{-itE(\boldsymbol{\xi})}H_{XXZ}|\boldsymbol{\xi}\rangle \\
    &= H_{XXZ}|\Psi(t)\rangle.
    \end{split}
\end{equation}
Now let's inspect $|\Psi(0)\rangle$. We have
\begin{equation}
    \begin{split}
        |\Psi(0)\rangle &= \sum_{\mathbf{x}\in \mathcal{X}}\left(\sum_{[\boldsymbol{\xi}]\in\Xi}\ell(\mathbf{y},\boldsymbol{\xi})u(\boldsymbol{\xi},\mathbf{x}) \right)|x\rangle \\
        &= \sum_{\mathbf{x}\in \mathcal{X}} \mathds{1}(\mathbf{x}=\mathbf{y}) |\mathbf{x}\rangle = |\mathbf{y}\rangle
    \end{split}
\end{equation}
where the second equality holds since we are assuming that \eqref{c:main_v3} is true. Thus, we have shown that $|\Psi(t)\rangle$ satisfies the initial value problem $\eqref{e:schrodinger}$.  
\end{proof}

\subsection{Probability function}

With an exact solution to the Schr\"odinger equation satisfying deterministic initial conditions, we can, in principle, compute observables explicitly. Of particular interest is the probability function \eqref{e:prob} from which we can extract statistics and potentially observe universal behavior under certain scaling limits. 

\begin{lemma}
Assume \Cref{c:main_v3} is true. Then, the probability function for the Heisenberg-Ising XXZ spin-$\tfrac{1}{2}$ chain, given by \eqref{e:schrodinger} with deterministic initial conditions $\mathbf{y} \in \mathcal{X}$, is given by 
\begin{equation}\label{e:prob_ba}
    \mathbb{P}\left(| \Psi(t) \rangle = | \mathbf{x} \rangle\right) = \left|\sum_{[\boldsymbol{\xi}] \in \Xi} \ell(\mathbf{y}, \boldsymbol{\xi})u(\boldsymbol{\xi},\mathbf{x}) e^{-i t E(\boldsymbol{\xi})} \right|^2
\end{equation}
where $u(\boldsymbol{\xi}, \mathbf{x})$ is given by \eqref{e:coefficients} and $\ell(\mathbf{y}, \boldsymbol{\xi})$ is given by \eqref{e:l_coeff}.
\end{lemma}

\begin{proof}
The probability of a configuration $\mathbf{x}$ at time $t \geq 0$ is given by
\begin{equation}
    \mathbb{P}\left(| \Psi(t) \rangle = | \mathbf{x}\rangle\right) = \langle \Psi(t)| \mathbf{x} \rangle \langle \mathbf{x} | \Psi(t) \rangle = \vert \langle \mathbf{x} | \Psi(t) \rangle  \vert^2;
\end{equation}
recall \eqref{e:prob}. Then, the result follows from \Cref{p:solution}.
\end{proof}

\subsection{One-point probability function}

The one-point probability function $\rho(x,t)$ is the probability of finding an up-spin at position $x \in [L]$ at time $t\geq 0$. We obtain the one-point probability function by summing the probability function \eqref{e:prob_ba} over all the configurations that contain an up-spin at the desired position. That is,
\begin{equation}\label{e:one-point}
\begin{split}
    \rho(x,t) &:= \sum_{\mathbf{x} \in X(x)} \mathbb{P}\left(| \Psi(t) \rangle = | \mathbf{x} \rangle\right)\\
    &= \sum_{\mathbf{x} \in X(x)} \left|\sum_{[\boldsymbol{\xi}] \in \Xi} \ell(\mathbf{y}, \boldsymbol{\xi})u(\boldsymbol{\xi},\mathbf{x}) e^{-i t E(\boldsymbol{\xi})} \right|^2
\end{split}
\end{equation}
where $X(x) = \{(x_1, x_2, \dots, x_N) \in X \mid x \in\{x_1, x_2, \dots, x_N\} \}$. The following plots are created using a \emph{Python} program that numerically finds all the solutions to the Bethe equations \eqref{e:bethe_equation}.

\subsection{Numerical solutions to the Bethe equations}

We solve the Bethe equations numerically to test \Cref{c:main_v3} for $\Delta \neq 0$. Recall that \Cref{l:main_v3} proves that \Cref{c:main_v3} is true for $\Delta=0$. In that case, we have an explicit solution for the Bethe equations, i.e.~roots of unity with a phase shift. Then, we employ Newton's method to solve the Bethe equations numerically for small $\Delta \neq 0$ using the solutions for $\Delta =0$ as our initial guess. For small $\Delta \neq 0$, we keep a one-to-one correspondence of the solutions of the Bethe equations with the solutions for $\Delta =0$.

We define a sequence $\boldsymbol{\xi}(k) \in \mathbb{C}^n$, $k=0, 1, \dots$, so that it has a direct numerical implementation and the limit is a solution of the Bethe equations. We use the solutions of the Bethe equations at $\Delta =0$ as the base point, i.e.~$\boldsymbol{\xi}(0) \in \mathbb{C}^N$. Note that the Bethe equations decouple in the case $\Delta =0$, 
\begin{equation}
    \xi_i^L = (-1)^{N-1}, \quad i = 1, 2, \dots, N.
\end{equation}
In this case, we can find all the solutions to the system of equations
\begin{equation}\label{e:delta_zero_solutions}
    \eta_i(k) = 
    \begin{cases}
    e^{2 k_i  \pi i /L}, &\quad N \text{ is odd}\\
    e^{(2k_i +1 ) \pi i /L}, &\quad N \text { is even}
    \end{cases}
\end{equation}
for any $k= (k_1, k_2, \dots, k_N)$ with $k_i \in [L]$. Thus, when $\Delta$ is close to zero, we expect the solutions to the Bethe equations to be close to the solutions given by \eqref{e:delta_zero_solutions}. More precisely, by Rouch\'e's Theorem, we expect
\begin{equation}
    |\xi_i - \eta_i| < C \epsilon
\end{equation}
for some constant $C> 0$ if $|\Delta| = \epsilon >0$. In our numerical solution to the Bethe equations, we use the solutions $\eta (k) = (\eta(k_i))_{i=1}^{N}$, given by \eqref{e:delta_zero_solutions}, as the base points of our sequence.

We now define our sequence that will give a solution of the Bethe equations; see \Cref{l:BE_limit}. For a fixed $k \in [L]^N$, we set the first element of the sequence as follows
\begin{equation}\label{e:base_point}
    \boldsymbol{\xi}(0) := \eta(k).
\end{equation}
Then, we define the rest of the terms in the sequence recursively as follows
\begin{equation}\label{e:recursion}
\begin{split}
    \boldsymbol{\xi}(n, 0) &= \boldsymbol{\xi}(n) \\
    \boldsymbol{\xi}(n, i+1) &= \boldsymbol{\xi}(n, i) - \left(\frac{F_{i+1}(\boldsymbol{\xi}(n,i))}{\Vert J_{i+1}(\boldsymbol{\xi}(n, i)) \Vert^2 }\right) J^{\dagger}_{i+1}(\boldsymbol{\xi}(n, i)), \quad i=0, 1, \dots, N-1,\\
    \boldsymbol{\xi}(n+1) &= \boldsymbol{\xi}(n, N)
\end{split}
\end{equation}
with the functions $F_i$, $i = 1, \dots, N$, given by the Bethe equations
\begin{equation}
    F_i(\boldsymbol{\xi}) = \xi_i^{L} + (-1)^N\prod_{j=1}^N\left(\frac{1 + \xi_i \xi_j - 2 \Delta \xi_j}{1 + \xi_i \xi_j - 2 \Delta \xi_i} \right)
\end{equation}
and $J_i = (\partial F_i / \partial\xi_1, \dots, \partial F_i / \partial\xi_N)$. 

\begin{lemma}\label{l:BE_limit}
For any $k \in [L]^N$, let $\{\boldsymbol{\xi}(n)\in \mathbb{C}^N \mid  n=0, 1, \dots\}$ be a sequence given by the recursion formula \eqref{e:recursion} and with a base point $\boldsymbol{\xi}(0) \in \mathbb{C}^N$ given by \eqref{e:base_point}. Then, there exists some $\epsilon >0$ so that the sequence converges for $\Delta$ if $|\Delta| < \epsilon$ and the limit point, $\boldsymbol{\xi}(n) \rightarrow \boldsymbol{\xi}(\infty)$ as $n \rightarrow \infty$, has all pairwise distinct entries and solves the Bethe equation given by \eqref{e:bethe_equation}.
\end{lemma}

\begin{proof}
The sequence defined by \eqref{e:recursion} is the so-called \emph{Newton-Kaczmarz} method to solve a system of non-linear equations. Our proof relies on satisfying the conditions of \cite[Theorem 2.1]{martinez1986method}. The key condition is that the vectors $J_i$, $i= 1, 2, \dots, N$, of partial derivatives are linearly independent. We know this condition to be true for $\Delta = 0$. Moreover, we know that this condition depends algebraically on $\Delta$ since this is equivalent to the condition that the determinant of the Jacobian is non-zero, i.e.~a polynomial on $\Delta$ is non-zero. By taking $\varepsilon >0$ small enough, this and the other conditions are satisfied. This establishes that the sequence is well-defined, converges and the limit is a solution of the Bethe equations.   

Recall, $\boldsymbol{\xi}(0) \in \mathbb{C}^N$ is the solution of the Bethe equations for $\Delta =0$. By Rouch\'e's theorem, we may assume that there is a small neighborhood around $\boldsymbol{\xi}(0) \in \mathbb{C}^N$ so that the Bethe equations have a unique solution with pairwise distinct entries if $|\Delta| < \varepsilon$ for $\varepsilon >0$ small enough. This establishes that the limit of the sequence has pairwise distinct entries.
\end{proof}

\begin{remark}
    We opted to use the Newton-Kaczmarz method \eqref{e:recursion} to solve the Bethe equations since it has a direct numerical implementation. An alternative approach was to use the regular Newton method to numerically solve the Bethe equations, but this required inverting an $N\times N$ matrix at each step of the sequence. 
\end{remark}

In \Cref{l:BE_limit}, the value $\Delta = 0$ is only special because we know how to solve the Bethe equations for that value. The result may be extended to find solutions of the Bethe equations near other base points $\boldsymbol{\xi}(0) \in \mathbb{C}^N$ if we know the base point to be a solution of the Bethe equations. In principle, we may find solutions for generic values of $\Delta$ by starting from the solutions for $\Delta =0$ and repeatedly applying the sequence \eqref{e:recursion} to find solutions that are closer and closer to $\Delta$. The key takeaway from this section is that we have a precise numerical method to solve the Bethe equations for non-trivial $\Delta \neq 0$ to verify \Cref{c:main_v3}.

\section{Numerical verification and code}\label{s:numerical}
To verify that the equations presented in the previous sections are accurate, we used Python to generate solutions and check that the output transition matrix at $t = 0$ was an identity matrix. The pseudocode for this is shown below, and the actual Python code can be found \href{https://github.com/elsaessern/Bethe-Functions}{\underline{here}}.

\begin{algorithm}[ht]
\caption{Transition Matrix Generation}
\begin{algorithmic}
    \For{all possible ordered n-tuples of $0 \dots (L-1)$} 
        \State Generate roots of unity as initial iterates
        \For{$t \in \{0, \dots, \text{numIters}\}$}
            \State Run an iteration of \Cref{e:recursion} for finding Bethe equations solutions
        \EndFor
    \EndFor
    \State Calculate $\sum_{\mathbf{x} \in \mathcal{X}} \left(\sum_{[\boldsymbol{\xi}] \in \Xi} \ell(\mathbf{y}, \boldsymbol{\xi})u(\boldsymbol{\xi}, \mathbf{x})\right) | \mathbf{x} \rangle$ for every configuration pair $(\mathbf{x}, \mathbf{y})$
    \State Check that this resulting matrix is an identity matrix
\end{algorithmic}
\end{algorithm}

The functions $\ell(\mathbf{y}, \boldsymbol{\xi})$ and $u(\boldsymbol{\xi}, \mathbf{x})$ are given by \eqref{e:l_coeff} and $\eqref{e:coefficients}$, respectively. We omitted the $e^{-itE(\boldsymbol{\xi})}$ because this code was only checking initial values, so this always evaluates to 1. Code for the above algorithm can be found in the \textit{DetFormula.ipynb}. 

\begin{figure}[ht]
    \centering
    \includegraphics[width=\linewidth]{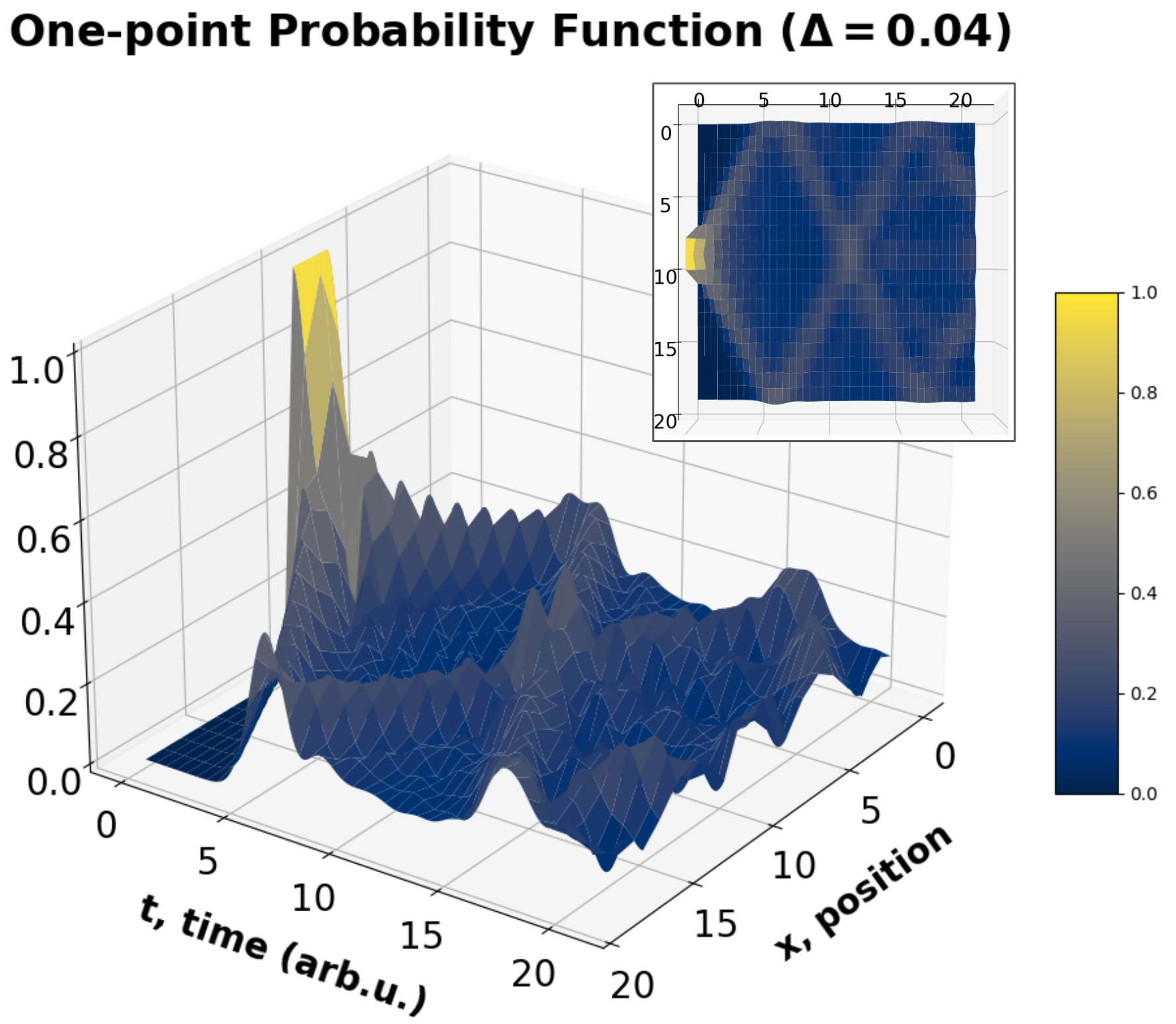}   
    \label{fig:N3-L21-d0.04-yi[20,21]}
    \caption{One-point function for $N=3$, $L=21$, $\Delta=0.04$, and initial condition $|\mathbf{y}\rangle=|8,9,10\rangle$}
\end{figure}

We also wrote code to generate plots for the one-point function for a given $N$, $L$, $\Delta$, and initial condition. Below are a few examples. The Python file \textit{OnePointFuncConsole.py} contains this code. To generate the plots for larger $N$ and $L$, we used Oregon State University's College of Engineering High Performance Computing Cluster.

\section{One-point function}\label{s:one-point function}
In the work of Saenz, Tracy, and Widom for XXZ on the line ($\mathbb{Z}$) \cite{saenz_domain_2022}, the marginal distribution of the left-most up-spin (often called the one-point function) with domain wall initial conditions ($|\mathbf{y}\rangle=|1, 2, 3,\ldots\rangle$) converges to the GUE Tracy-Widom distribution $F_2$ when $\Delta=0$ in the large $N$ limit (with proper scaling). Moreover in that work, there is an asymptotic formula for the one-point function when $\Delta \neq 0$, but it is unclear that the limit is also the GUE Tracy-Widom distribution. This connection to the KPZ universality class prompts us to investigate the similar statistics for the XXZ spin-1/2 chain on the ring. Here we present a simplification of a different kind of one-point function, assuming the completeness of the Bethe Ansatz but independent of the specific form of the $\ell$-function. In this setting, our one-point function could also be called the occupation probability, as it gives the probability of having an up-spin at a particular site $x$ and time $t$. 

\begin{theorem}\label{t:one-point}
Assume the Bethe Ansatz is complete. Then, the one-point function, given by \eqref{e:one-point}, has the following formula 
\begin{equation}
    \rho(x,t) = \sum_{[\boldsymbol{\xi}], [\boldsymbol{\zeta}] \in \Xi} \ell(\mathbf{y}, \boldsymbol{\xi}) \mathcal{F}(x; \boldsymbol{\xi}, \boldsymbol{\zeta})\ell(\mathbf{y}, \boldsymbol{\zeta})
\end{equation}
with the function $\mathcal{F}(x; \boldsymbol{\xi}, \boldsymbol{\zeta})$ given by
\begin{equation}
\begin{split}
    &\mathcal{F}(x;\boldsymbol{\xi}, \boldsymbol{\zeta}) :=\\
    &\frac{e^{-it(E(\boldsymbol{\xi})-E(\boldsymbol{\zeta}))}F_{1,N}(\boldsymbol{\xi}, \boldsymbol{\zeta})^x }{\prod_{1 \leq i <j \leq N}(1 + \xi_i \xi_j - 2 \Delta \xi_i)(1 + \zeta_i \zeta_j - 2 \Delta \zeta_i)}\\
    &\times\sum_{s = 1}^{N} \sum_{\substack{I_1, I_2 \subset [N]\\|I_1|=|I_2| = s}} \bigg[ ( F_{1, s}(\boldsymbol{\xi}, \boldsymbol{\zeta}; I_1, I_2)^{-1}-1)\Gamma(\boldsymbol{\xi}, \boldsymbol{\zeta}; I_1, I_2)G(\boldsymbol{\xi}; I_1^c, I_1)\\
    &\hspace{15mm}\times G(\boldsymbol{\zeta}; I_2^c, I_2)\Gamma(\boldsymbol{\xi}^{-1}, \boldsymbol{\zeta}^{-1}; I_1^c, I_2^c) F_{1, N-s}(\boldsymbol{\xi}, \boldsymbol{\zeta}; I_1^c, I_2^c)^{N-s-2}\bigg].
\end{split}
\end{equation}
\end{theorem}

The definition of the one-point function includes a summation over a large subset of the configuration space. Without simplification, it still far from a closed simplification amenable to asymptotic analysis, and the computational complexity of the one-point function is on the order of $\binom{L}{N}^2\binom{L-1}{N-1}\left(N! \right)^4$ which is much too large for any kind of numerical implementation of a chain with large $N$ and $L$. In what follows we will apply and extend identities from \cite{liu_integral_2020},\cite{tracy_asymptotics_2009}, \cite{saenz_domain_2022}, \cite{Cantini:2019zmz} and establish some new identities to simplify the one-point function to the form in \Cref{t:one-point}.
\begin{remark}
    Due to the rotational invariance of the lattice, we may express the one-point function at a point $x\in [L]$ in terms of the one-point function at $x=0$ by taking $\mathbf{y}\mapsto \mathbf{y}-x$. Our work to show \Cref{t:one-point} will be done for $\rho(0,t)$ for simplicity but easily extends for a general $x\neq 0$. 
\end{remark}
We will first establish some notation and then focus on simplifying $\sum_{\mathbf{x} \in X(0)}u(\boldsymbol{\xi}, \mathbf{x})u(\boldsymbol{\zeta}, \mathbf{x})$. 

\subsection{Notation and preliminary identities}

We use the following notation. For a fixed subset $I \subset \{1, \dots, N \}$ and a collection of functions $f_k(\boldsymbol{\xi}) = f_k(\xi_1, \dots, \xi_k)$, we denote the function restricted to a set of indices
\begin{equation}
    f(\boldsymbol{\xi}; I) = f_{n}(x_1, \dots, x_n)
\end{equation}
where $|I|=n$, the elements of $I$ are ordered so that $I = \{i_1 < \cdots < i_n\}$, and the variables are identified $x_{\alpha} = \xi_{i_{\alpha}}$. Additionally, when $I$ is the whole set, we use the following identification $f(\boldsymbol{\xi}) := f(\boldsymbol{\xi}; \{1, \dots, N\})$. Similarly, we denote the following function on two variables restricted on two sets
\begin{equation}
    f(\boldsymbol{\xi}, \boldsymbol{\zeta}; I_1, I_2) = f_n(x_1, \dots, x_n , z_1,\dots, z_n)
\end{equation}
where $|I_1|= |I_2|=n$, the elements of $I_1$ and $I_2$ are ordered so that $I_1 = \{i_1 < \cdots < i_n\}$ and $I_2 = \{j_1 < \cdots < j_n\}$, and the variables are identified $x_{\alpha} = \xi_{i_{\alpha}}$ and $z_{\alpha} = \zeta_{j_{\alpha}}$. Additionally, when $I_1$ and $I_2$ are the whole set, we use the following identification $f(\boldsymbol{\xi}, \boldsymbol{\zeta}) := f(\boldsymbol{\xi}, \boldsymbol{\zeta}; \{1, \dots, N\}, \{1, \dots, N\})$. In the case when $I$ is empty, $f(\boldsymbol{\xi};I)=1$ and $f(\boldsymbol{\xi},\boldsymbol{\zeta};I,J)=1$. 

We define the following functions
\begin{equation}
    \begin{split}
    B_{\sigma} (\xi_1, \dots, \xi_N) &= \mathrm{sgn}(\sigma) \prod_{1 \leq i <j \leq N}\left(1 + \xi_{\sigma(i)}\xi_{\sigma(j)} - 2 \Delta \xi_{\sigma(i)}\right) \prod_{j=1}^N\xi_{\sigma(j)}^j \\
    &= A_\sigma(\boldsymbol{\xi})\prod_{1\leq i<j\leq N}\left( 1+\xi_i\xi_j-2\Delta\xi_i \right) \prod_{j=1}^N\xi_{\sigma(j)}^j
    \end{split}
\end{equation}
for any permutation $\sigma \in S_N$ and 
\begin{equation}
    C(\boldsymbol{\xi}, \boldsymbol{\zeta}; \mathbf{x}) = \sum_{\sigma, \mu \in S_N} B_{\sigma}(\boldsymbol{\xi}) B_{\mu}(\boldsymbol{\zeta}) \prod_{j=1}^N \xi_{\sigma(j)}^{x_j-j}\zeta_{\mu(j)}^{x_j -j}.
\end{equation}

With these functions we can now write
\begin{equation}\label{e:oneptsimplificationstep1}
\begin{split}
    \sum_{\mathbf{x} \in X(0)}u(\boldsymbol{\xi}, \mathbf{x})u(\boldsymbol{\zeta}, \mathbf{x}) &= \frac{\sum_{\mathbf{x} \in X(0)}\sum_{\sigma ,\mu \in S_N} B_{\sigma}(\boldsymbol{\xi}) B_{\mu}(\boldsymbol{\zeta}) \prod_{j=1}^N \xi_{\sigma(j)}^{x_j-j}\zeta_{\mu(j)}^{x_j-j}}{\prod_{1 \leq i <j \leq N}(1 + \xi_i \xi_j - 2 \Delta \xi_i)(1 + \zeta_i \zeta_j - 2 \Delta \zeta_i)}\\
    &=\frac{\sum_{\mathbf{x} \in X(0)}C(\boldsymbol{\xi}, \boldsymbol{\zeta}; \mathbf{x})}{\prod_{1 \leq i <j \leq N}(1 + \xi_i \xi_j - 2 \Delta \xi_i)(1 + \zeta_i \zeta_j - 2 \Delta \zeta_i)}.\\
\end{split}
\end{equation}

We define the following family of functions \begin{equation}
    F_{m,n}(\boldsymbol{\xi}, \boldsymbol{\zeta}) = \prod_{j=m}^n \xi_j \zeta_j, \quad F_{m,n}^{\sigma, \mu}(\boldsymbol{\xi},\boldsymbol{\zeta}) = \prod_{j=m}^n \xi_{\sigma(j)} \zeta_{\mu(j)}
\end{equation}
for any pair of permutations $\sigma, \mu \in S_N$ and $F_{m, m-1} =1$. Then, we have the following lemma from \cite{liu_integral_2020} and \cite{Baik2019-nd}:

\begin{lemma}(\cite{liu_integral_2020}, Lemma 6.1; \cite{Baik2019-nd}, Lemma 5.3)
    \begin{equation}
    \begin{split}
    &\sum_{\substack{\mathbf{x} \in X\\x_1=0}} \left( \prod_{j=1}^N \xi_{j}^{x_j-j}\right)\\
    &=\sum_{k=0}^{N-1}(-1)^k \sum_{1 < s_1 < \cdots < s_k < N+1} \frac{1 - F_{1, s_1-1}}{F_{1, s_1-1}}(F_{s_1, N})^{L-N} \prod_{i=0}^{k} \frac{1}{\prod_{j=s_i}^{s_{i+1}-1}(1 - F_{j, s_{i+1}-1}) }
    \end{split}
    \end{equation}
\end{lemma}

We define the following function
\begin{equation}
\begin{split}
    &C_{s_1, \cdots, s_k}(\boldsymbol{\xi}, \boldsymbol{\zeta}) \\
    &= \sum_{\sigma, \mu \in S_N}\left( B_{\sigma}(\boldsymbol{\xi})B_{\mu}(\boldsymbol{\zeta})
   \frac{1 - F_{1, s_1-1}^{\sigma, \mu}}{F_{1, s_1-1}^{\sigma, \mu}}(F_{s_1, N}^{\sigma, \mu})^{L-N} \prod_{i=0}^{k} \frac{1}{\prod_{j=s_i}^{s_{i+1}-1}(1 - F_{j, s_{i+1}-1}^{\sigma,\mu}) }\right)
\end{split}
\end{equation}
for $1 < s_1 < \cdots < s_k < N+1$. By convention, we set $s_0 =1$. By the prior lemma, we have that

\begin{equation}\label{e:CDecomp}
    \sum_{0=x_1<x_2<\cdots<x_N<L}C(\boldsymbol{\xi},\boldsymbol{\zeta},\mathbf{x})=\sum_{k=0}^{N-1}(-1)^k\sum_{1<s_1<\cdots<s_k<N+1}C_{\vec{s}}(\boldsymbol{\xi},\boldsymbol{\zeta}).
\end{equation}

Then we can write \eqref{e:oneptsimplificationstep1} as 
\begin{equation}\label{e:oneptsimplificationstep2}
\begin{split}
    &\frac{\sum_{\mathbf{x} \in X(0)}C(\boldsymbol{\xi}, \boldsymbol{\zeta}; \mathbf{x})}{\prod_{1 \leq i <j \leq N}(1 + \xi_i \xi_j - 2 \Delta \xi_i)(1 + \zeta_i \zeta_j - 2 \Delta \zeta_i)}\\
    =&\frac{\sum_{k=0}^{N-1}(-1)^k \sum_{1< s_1 < \cdots < s_k < N+1}C_{\vec{s}}(\boldsymbol{\xi}, \boldsymbol{\zeta})}{\prod_{1 \leq i <j \leq N}(1 + \xi_i \xi_j - 2 \Delta \xi_i)(1 + \zeta_i \zeta_j - 2 \Delta \zeta_i)}.
\end{split}
\end{equation}
Now our aim is to break down the sum over the $C_{\vec{s}}(\boldsymbol{\xi},\boldsymbol{\zeta})$ terms.
We begin by decomposing the permutation group based on the image of the set $\{1,2, \dots, s_1-1\}$. For a fixed subset $I \subset [N]=\{1,\ldots,N\}$, we denote the set of permutations that map $\{1,2, \dots, s_1-1\}$ to $I$ as follows
\begin{equation}
    S(I) = \{\sigma \in S_N  \mid \sigma(\{1, 2, \dots, s_1-1 \}) =I\}.
\end{equation}
Then, we have the following decomposition of the permutation group
\begin{equation}\label{e:decomp1}
    S_{N} = \bigcup_{\substack{I \subset [N]\\|I|=s_1-1}} S(I).
\end{equation}
Furthermore, we have the following decomposition
\begin{equation}\label{e:decomp2}
    S(I) = S_1(I) \times S_2(I)
\end{equation}
where we identify $\sigma \in S(I)$ with the restriction to the subset $\{1, \dots, |I|\}$ and its complements,
\begin{equation}
    S_1(I) = \{\sigma|_{\{1, \dots, |I|\}} \mid \sigma\in S(I) \}, \quad S_2(I) = \{\sigma|_{\{|I|+1, \dots, N\}} \mid \sigma\in S(I) \}.
\end{equation}
Note that $S_1(I)$ and $S_2(I)$ are isomorphic to the permutation groups $S_{|I|}$ and $S_{N-|I|}$. We identify $\lambda_1 \in S_{|I|}$ and $\lambda_2 \in S_{N-|I|}$ with $\sigma_1(\lambda_1) \in S_1(I)$ and $\sigma_2(\lambda_2) \in S_2(I)$ as follows
\begin{equation}\label{e:per_decomp}
\begin{split}
\sigma_1(\lambda_1)&: k \in \{1, \dots, |I|\} \mapsto i_{\lambda_1(k)} \in I\\
\sigma_2(\lambda_2)&: k\in \{|I|+1, \dots, N\} \mapsto j_{\lambda_2(k)} \in I^c
\end{split}
\end{equation}
where we have ordered the elements $I = \{ i_1 < \cdots < i_{|I|}\}$ and $I^c = \{ j_1 < \cdots< j_{N-|I|}\}$. A more basic discussion on this decomposition can be found in section \eqref{section: S_N decomp}.

With this kind of decomposition, we can now write
\begin{equation}\label{e:CsDecomp1}
\begin{split}
    &C_{\vec{s}}(\boldsymbol{\xi},\boldsymbol{\zeta})= \sum_{\substack{I_1,I_2 \subset [N]\\|I_1| = |I_2| = s_1-1}} \left(\frac{1-F_{1,s_1-1}^{\sigma,\mu}(\boldsymbol{\xi},\boldsymbol{\zeta};I_1,I_2)}{F_{1,s_1-1}^{\sigma,\mu}(\boldsymbol{\xi},\boldsymbol{\zeta};I_1,I_2)}\right) \sum_{\sigma \in S(I_1), \mu \in S(I_2)}  \\
    &B_\sigma(\vec{\boldsymbol{\xi})}B_\mu(\boldsymbol{\zeta})\left(F_{s_1,N}^{\sigma,\mu}(\boldsymbol{\xi},\boldsymbol{\zeta})\right)^{L-N} \prod_{i=0}^k\frac{1}{\prod_{j=s_i}^{s_{i+1}-1}\left(1-F_{j,s_{i+1}-1}^{\sigma,\mu}(\boldsymbol{\xi},\boldsymbol{\zeta})\right)}
\end{split}
\end{equation}
where we have used the decomposition given by \eqref{e:decomp1} twice, once per permutation group. 

Now we define the following functions:
\begin{equation}\label{oneptGfxn}
    G(\boldsymbol{\xi};J_1, \dots, J_k) = \prod_{1 \leq i < j \leq k} \left(\prod_{\alpha \in J_i, \beta\in J_j}(1 + \xi_{\alpha} \xi_{\beta} - 2\Delta \xi_{\alpha}) \prod_{\substack{\alpha>\beta\\ \alpha \in J_i, \beta \in J_j}}(-1)\right)
\end{equation}
and
\begin{equation} \label{oneptEfxn}
    E(\boldsymbol{\xi}, \boldsymbol{\zeta}; I_1, I_2) = \left(F(\boldsymbol{\xi},\boldsymbol{\zeta};I_1,I_2)^{-1} -1\right) \sum_{\lambda_1, \tau_1 \in S_{s}} \frac{B_{\lambda_1}(\boldsymbol{\xi}; I_1)B_{\tau_1}(\boldsymbol{\zeta}; I_2)}{\prod_{j=1}^{|I|} (1 - F_{j, |I|}^{\lambda_1, \tau_1}(\boldsymbol{\xi}, \boldsymbol{\zeta}; I_1, I_2))}
\end{equation}
and
\begin{equation} \label{oneptDfxn}
\begin{split}
    &D_{s_1, \dots, s_k}(\boldsymbol{\xi}, \boldsymbol{\zeta}; I_1^c, I_2^c)\\
    &= \sum_{\lambda_2, \tau_2 \in S_{N-s}}B_{\lambda_2}(\boldsymbol{\xi}; I_1^c) B_{\tau_2}(\boldsymbol{\zeta}; I_2^c) \prod_{i=1}^k \frac{(F_{s_i, s_{i+1}-1}^{\lambda_2, \tau_2}(\boldsymbol{\xi}, \boldsymbol{\zeta}; I_1^c, I_2^c))^{L- |I^c|}}{\prod_{j=s_i +1}^{s_{i+1}-1}(1- F_{j, s_{i+1}-1}^{\lambda_2, \tau_2}(\boldsymbol{\xi}, \boldsymbol{\zeta}; I_1^c, I_2^c))}
\end{split}
\end{equation}
where $|I_1| = |I_2| = s$.

\subsection{Further simplification and identities}
We simplify $C_{\vec{s}}(\boldsymbol{\xi},\boldsymbol{\zeta})$ with the following lemma whose proof can be found in \Cref{s:L6.2 and 6.5 proofs}. 
\begin{lemma}\label{l:BDecomp} 
Fix $I\subset\{1,\ldots,N\}$, $|I| = s_1-1$. Let $\sigma\in S(I)$ such that $\sigma=(\sigma_1(\lambda_1),\sigma_2(\lambda_2))$ with $\lambda_1\in S_{|I|}$ and $\lambda_2\in S_{|I^c|}$; see \eqref{e:per_decomp}. Then
    \begin{equation}\label{BDecomp}
        B_\sigma(\vec{\boldsymbol{\xi})} = B_{\lambda_1}(\boldsymbol{\xi};I)G(\boldsymbol{\xi};I,I^c)B_{\lambda_2}(\boldsymbol{\xi};I^c)\left(F_{s_1,N}^{\sigma}(\boldsymbol{\xi})\right)^{|I|}.
    \end{equation}
\end{lemma}

With the above lemma we can now write the following proposition
\begin{proposition}
    \begin{equation}
C_{\vec{s}}(\boldsymbol{\xi}, \boldsymbol{\zeta}) = \sum_{\substack{I_1,I_2 \subset [N]\\|I_1| = |I_2| = s_1-1}} E(\boldsymbol{\xi}, \boldsymbol{\zeta}; I_1, I_2) G(\boldsymbol{\xi},I_1, I_1^c)G(\boldsymbol{\zeta},I_2, I_2^c)D_{\vec{s}}(\boldsymbol{\xi}, \boldsymbol{\zeta}; I_1^c, I_2^c) 
\end{equation}
\end{proposition}

\begin{proof}Using the decomposition given by \eqref{e:decomp2} and the identification given by \eqref{e:per_decomp}, we have that
\begin{equation}
\begin{split}
    &C_{\vec{s}}(\boldsymbol{\xi},\boldsymbol{\zeta}) =\sum_{\substack{I_1,I_2 \subset [N]\\|I_1| = |I_2| = s_1-1}} \left(\frac{1-F_{1,s_1-1}^{\sigma,\mu}(\boldsymbol{\xi},\boldsymbol{\zeta}; I_1, I_2)}{F_{1,s_1-1}^{\sigma,\mu}(\boldsymbol{\xi},\boldsymbol{\zeta}; I_1, I_2)}\right) \sum_{\lambda_1, \tau_1 \in S_{s_1-1}}\sum_{\lambda_2, \tau_2 \in S_{N+1-s_1}}\\
    &B_{\lambda_1}(\boldsymbol{\xi};I_1)B_{\tau_1}(\boldsymbol{\zeta};I_2)G(\boldsymbol{\xi};I_1,I_1^c)G(\boldsymbol{\zeta};I_2,I_2^c)B_{\lambda_2}(\boldsymbol{\xi};I_1^c)B_{\tau_2}(\boldsymbol{\zeta};I_2^c)\\
    &\times\left(F_{s_1,N}^{\lambda_2,\tau_2}(\boldsymbol{\xi},\boldsymbol{\zeta}; I_1^c, I_2^c)\right)^{L-N+s_1-1} \prod_{i=1}^k\frac{1}{\prod_{j=s_i}^{s_{i+1}-1}\left(1-F_{j,s_{i+1}-1}^{\lambda_2,\tau_2}(\boldsymbol{\xi},\boldsymbol{\zeta}; I_1^c, I_2^c)\right)}\\
    &\times \left(\frac{1}{\prod_{j=1}^{s_{1}-1}\left(1-F_{j,s_{1}-1}^{\lambda_1,\tau_1}(\boldsymbol{\xi},\boldsymbol{\zeta};I_1, I_2)\right)} \right).
\end{split}
\end{equation}
Then by \Cref{l:BDecomp}
\begin{equation}\label{e:CsToEGGDs}
C_{\vec{s}}(\boldsymbol{\xi}, \boldsymbol{\zeta}) = \sum_{\substack{I_1 \subset [N]\\|I_1| = |I_2| = s_1-1}} E(\boldsymbol{\xi}, \boldsymbol{\zeta}; I_1, I_2) G(\boldsymbol{\xi},I_1, I_1^c)G(\boldsymbol{\zeta},I_2, I_2^c)D_{\vec{s}}(\boldsymbol{\xi}, \boldsymbol{\zeta}; I_1^c, I_2^c) .
\end{equation} 
\end{proof}

Adapted to our notation, we use the following identity 
\begin{lemma}(\cite{saenz_domain_2022} equation $(23)$, \cite{Cantini:2019zmz} Proposition $6$)\label{l:STW(23)}
\begin{equation}
    \sum_{\sigma, \mu \in S_N} \frac{B_{\sigma}(\boldsymbol{\xi}) B_{\mu}(\boldsymbol{\zeta})}{\prod_{j=1}^N \left(1 - F_{j, N}(\boldsymbol{\xi}, \boldsymbol{\zeta})\right)} = F_{1, N}(\boldsymbol{\xi}, \boldsymbol{\zeta}) \Gamma(\boldsymbol{\xi}, \boldsymbol{\zeta})
\end{equation}
where
\begin{equation}
    \Gamma(\boldsymbol{\xi},\boldsymbol{\zeta}) =\det \left[\left( \frac{\prod_{k \neq j}(\xi_i + \zeta_k - 2 \Delta \xi_i \zeta_k)}{1-\xi_i \zeta_j}\right)_{i, j=1}^N \right].
\end{equation}
\end{lemma}

We now simplify \eqref{oneptEfxn} and \eqref{oneptDfxn} via \eqref{l:STW(23)}. For \eqref{oneptEfxn}, we have
\begin{equation}
    E(\boldsymbol{\xi},\boldsymbol{\zeta};I_1, I_2)=(1-F_{1,s_1-1}(\boldsymbol{\xi},\boldsymbol{\zeta};I_1, I_2))\Gamma(\boldsymbol{\xi},\boldsymbol{\zeta};I_1, I_2)
\end{equation}

and for \eqref{oneptDfxn} we first sum over all possible partitions of $I_1^c$ and $I_2^c$ as we did in \eqref{e:CDecomp}, and then we apply \Cref{l:STW(23)}
\begin{equation}\label{e:DsDecomp}
\begin{split}
    &D_{s_1,\ldots,s_k}(\boldsymbol{\xi},\boldsymbol{\zeta};I_1^c, I_2^c)=\sum_{\substack{J_{1,1}\cup\cdots\cup J_{1,k} = I_1^c \\J_{2,1}\cup\cdots\cup J_{2,k} = I_2^c \\
    |J_{m, l}|=s_{l+1}-s_l}}G(\boldsymbol{\xi};J_{1,1},\ldots,J_{1,k})G(\boldsymbol{\zeta};J_{2,1},\ldots,J_{2,k})\\
    &\times \prod_{l=1}^k\left( \sum_{\lambda_{(l)}\in S_{|J_{1,l}|}, \tau_{(l)} \in S_{|J_{2,l}|}}\frac{B_{\lambda_{(l)}}(\boldsymbol{\xi};J_{1,l})B_{\tau_{(l)}}(\boldsymbol{\zeta};J_{2,l})\left(F_{1,|J_l|}^{\lambda_{(l)},\tau_{(l)}}(\boldsymbol{\xi},\boldsymbol{\zeta};J_{1,l}, J_{2, l}) \right)^{L-N+s_l-1}}{\prod_{j=1}^{|J_l|}\left(1-F_{j,|J_l|}^{\lambda_{(l)},\tau_{(l)}}(\boldsymbol{\xi},\boldsymbol{\zeta};J_{1, l}, J_{2,l}) \right)}  \right) \\
    &= \left(F_{1,N+1-s_1}(\boldsymbol{\xi},\boldsymbol{\zeta};I_1^c, I_2^c) \right)^{L-N+s_1}\sum_{\substack{J_{1,1}\cup\cdots\cup J_{1,k} = I_1^c \\J_{2,1}\cup\cdots\cup J_{2,k} = I_2^c \\
    |J_{m, l}|=s_{l+1}-s_l}}G(\boldsymbol{\xi};J_{1,1},\ldots,J_{1,k})G(\boldsymbol{\zeta};J_{2,1},\ldots,J_{2,k})\\
    &\times \prod_{l=1}^k\left(F_{1,|J_l|}(\boldsymbol{\xi},\boldsymbol{\zeta};J_{1,l}, J_{2,l}) \right)^{s_l-s_1}\Gamma(\boldsymbol{\xi},\boldsymbol{\zeta};J_{1,l}, J_{2,l}).
\end{split}
\end{equation}
To handle this last sum we introduce the following lemma (adapted from Lemma 6.3 in \cite{liu_integral_2020}). The proof can be found in \Cref{s:L6.2 and 6.5 proofs} 
\begin{lemma}\label{c:LSW20L6.3}
    \begin{equation}
    \begin{split}
    &\sum_{k=1}^n(-1)^k\sum_{\substack{I_{1,1}\cup\cdots\cup I_{1,k}=\{1,\ldots,n\}\\I_{2,1}\cup\cdots\cup I_{2,k}=\{1,\ldots,n\}}} G(\boldsymbol{\xi};I_{1,1},\ldots,I_{1,k})G(\boldsymbol{\zeta};I_{2,1},\ldots,I_{2,k})\\&\times\prod_{\alpha=1}^kF_{1,e_{\alpha}}(\boldsymbol{\xi},\boldsymbol{\zeta};I_{1,\alpha}, I_{2,\alpha})^{\left(\sum_{i=1}^{\alpha-1}e_i \right)}\Gamma(\boldsymbol{\xi},\boldsymbol{\zeta}; I_{1, \alpha}, I_{2, \alpha})\\
    &= \left(F_{1,n}(\boldsymbol{\xi},\boldsymbol{\zeta}) \right)^{2n-3}\Gamma(\boldsymbol{\xi}^{-1},\boldsymbol{\zeta}^{-1}).
    \end{split}
    \end{equation}
where $e_{\alpha} = |I_{1, \alpha}| = |I_{2, \alpha}| $, $\boldsymbol{\xi}^{-1} = (\xi_1^{-1}, \dots, \xi_N^{-1})$ and $\boldsymbol{\zeta}^{-1} = (\zeta_1^{-1}, \dots, \zeta_N^{-1})$.
\end{lemma}

Then, we have
\begin{equation}\label{e:DsDecompFinal}
\begin{split}
    &\sum_{k=1}^{N-s} (-1)^k \sum_{1+s =s_1< s_2 <\cdots < s_k \leq N} D_{s_1, \dots, s_k}(\boldsymbol{\xi}, \boldsymbol{\zeta}; I_1^c, I_2^c)\\
    &= F_{1, N+1-s_1}(\boldsymbol{\xi}, \boldsymbol{\zeta}; I_1^c, I_2^c)^{L+N-s_1-1} \Gamma(\boldsymbol{\xi}^{-1}, \boldsymbol{\zeta}^{-1}; I_1^c, I_2^c).
\end{split}    
\end{equation}

\begin{lemma}\label{l:GGF^L to GG}
Assume $[\boldsymbol{\xi}],[\boldsymbol{\zeta}] \in \Xi$. Then,
\begin{equation}
    G(\boldsymbol{\xi}; I_1,I_1^c )G(\boldsymbol{\zeta}; I_2,I_2^c ) F_{1, s}(\boldsymbol{\xi}, \boldsymbol{\zeta}; I_1, I_2)^L = G(\boldsymbol{\xi}; I_1^c, I_1) G(\boldsymbol{\zeta}; I_2^c, I_2) 
\end{equation}
where $s = |I_1|= |I_2|$ and $I_1, I_2 \subset [N]$.
\end{lemma}

\begin{proof}
It is sufficient to show
\begin{equation}
\begin{split}
    &\prod_{\beta\in I^c}\xi_{\beta}^L\prod_{\alpha \in I, \beta \in I^c} (1+ \xi_{\alpha} \xi_{\beta} - 2 \Delta \xi_{\alpha})\prod_{\substack{\alpha >\beta \\\alpha \in I, \beta \in I^c}}(-1)\\
    &= \prod_{\alpha \in I^c, \beta \in I} (1+ \xi_{\alpha} \xi_{\beta} - 2 \Delta \xi_{\alpha})\prod_{\substack{\alpha >\beta \\\alpha \in I^c, \beta \in I}}(-1).
\end{split}
\end{equation}
Due to the Bethe equation \eqref{e:bethe_equation}, we have
\begin{equation}
\begin{split}
    \prod_{\beta \in I^c} \xi_{\beta}^L &= (-1)^{(N-1)|I^c|}\prod_{\beta \in I^c} \prod_{j =1}^N \frac{1 + \xi_{\beta} \xi_j - 2 \Delta \xi_{\beta}}{1 + \xi_{\beta} \xi_j - 2 \Delta \xi_{j}}\\
    &= (-1)^{|I||I^c|}\prod_{\alpha \in I, \beta \in I^c}\frac{1 + \xi_{\beta} \xi_{\alpha} - 2 \Delta\xi_{\beta}}{1 + \xi_{\beta} \xi_{\alpha} - 2 \Delta\xi_{\alpha}}.
\end{split}
\end{equation}
Then, the result follows.
\end{proof}

\subsection{One-point function simplification}
We summarize the preceding simplifications as follows.
the one-point function reads\footnote{We omit the complex conjugation in the expansion of the probabilities, since the Bethe roots come in complex conjugate pairs and it is straightforward to check that $\overline{u(\boldsymbol{\xi},\textbf{x})}=u\left(\overline{\boldsymbol{\xi}},\textbf{x}\right)$ and $\overline{\ell(\textbf{y},\boldsymbol{\xi})}=\ell\left(\textbf{y},\overline{\boldsymbol{\xi}}\right)$.}
\begin{equation}
\begin{split}
    \rho(0;t) &= \sum_{\mathbf{x} \in X(0)} \mathbb{P}(|\Psi(t) \rangle = | \mathbf{x}\rangle)\\
    &=\sum_{\mathbf{x} \in X(0)} \sum_{[\boldsymbol{\xi}], [\boldsymbol{\zeta}] \in \Xi}\ell(\mathbf{y}, \boldsymbol{\xi})\overline{\ell(\mathbf{y}, \boldsymbol{\zeta})} u(\boldsymbol{\xi}, \mathbf{x})\overline{u(\boldsymbol{\zeta}, \mathbf{x})} e^{-it(E(\boldsymbol{\xi})-E(\boldsymbol{\zeta}))}\\
    &=\sum_{[\boldsymbol{\xi}], [\boldsymbol{\zeta}] \in \Xi}\ell(\mathbf{y}, \boldsymbol{\xi})\ell(\mathbf{y}, \boldsymbol{\zeta}) \left(\sum_{\mathbf{x} \in X(0)}u(\boldsymbol{\xi}, \mathbf{x})u(\boldsymbol{\zeta}, \mathbf{x})\right) e^{-it(E(\boldsymbol{\xi})-E(\boldsymbol{\zeta}))}.
\end{split}
\end{equation}

We focus on the summation inside the parenthesis. From \eqref{e:oneptsimplificationstep2} we have 
\begin{equation}
\begin{split}
    \sum_{\mathbf{x} \in X(0)}u(\boldsymbol{\xi}, \mathbf{x})u(\boldsymbol{\zeta}, \mathbf{x}) &= \frac{\sum_{\mathbf{x} \in X(0)}\sum_{\sigma ,\mu \in S_N} B_{\sigma}(\boldsymbol{\xi}) B_{\mu}(\boldsymbol{\zeta}) \prod_{j=1}^N \xi_{\sigma(j)}^{x_j-j}\zeta_{\mu(j)}^{x_j-j}}{\prod_{1 \leq i <j \leq N}(1 + \xi_i \xi_j - 2 \Delta \xi_i)(1 + \zeta_i \zeta_j - 2 \Delta \zeta_i)}\\
    &=\frac{\sum_{\mathbf{x} \in X(0)}C(\boldsymbol{\xi}, \boldsymbol{\zeta}; \mathbf{x})}{\prod_{1 \leq i <j \leq N}(1 + \xi_i \xi_j - 2 \Delta \xi_i)(1 + \zeta_i \zeta_j - 2 \Delta \zeta_i)}\\
    &=\frac{\sum_{k=0}^{N-1}(-1)^k \sum_{1< s_1 < \cdots < s_k < N+1}C_{\vec{s}}(\boldsymbol{\xi}, \boldsymbol{\zeta})}{\prod_{1 \leq i <j \leq N}(1 + \xi_i \xi_j - 2 \Delta \xi_i)(1 + \zeta_i \zeta_j - 2 \Delta \zeta_i)}.
\end{split}
\end{equation}
Moreover, by \eqref{e:CsToEGGDs}, \eqref{e:DsDecompFinal}, and \Cref{l:GGF^L to GG}
\begin{equation}
\begin{split}
    &\sum_{k=0}^{N-1}(-1)^k \sum_{1< s_1 < \cdots < s_k < N+1}C_{\vec{s}}(\boldsymbol{\xi}, \boldsymbol{\zeta})\\
    &= \sum_{k=0}^{N-1}(-1)^k\sum_{s = 1}^{N} \sum_{\substack{I_1, I_2 \subset [N]\\|I_1|=|I_2| = s}} (1 - F_{1, s}(\boldsymbol{\xi}, \boldsymbol{\zeta}; I_1, I_2))\Gamma(\boldsymbol{\xi}, \boldsymbol{\zeta}; I_1, I_2)\\
    &\times G(\boldsymbol{\xi}; I_1^c, I_1)G(\boldsymbol{\zeta}; I_2^c, I_2)F_{1, N-s}(\boldsymbol{\xi}, \boldsymbol{\zeta}; I_1^c, I_2^c)^{N-s-2} \Gamma(\boldsymbol{\xi}^{-1}, \boldsymbol{\zeta}^{-1}; I_1^c, I_2^c)
\end{split}
\end{equation}

Therefore, we have
\begin{equation}
\begin{split}
    &\rho(0;t)\\
    &= \sum_{[\boldsymbol{\xi}], [\boldsymbol{\zeta}] \in \Xi}\ell(y, \boldsymbol{\xi}) \ell(y, \boldsymbol{\zeta})\sum_{s = 1}^{N} \sum_{\substack{I_1, I_2 \subset [N]\\|I_1|=|I_2| = s}}  ( F_{1, s}(\boldsymbol{\xi}, \boldsymbol{\zeta}; I_1, I_2)^{-1}-1)F_{1, N-s}(\boldsymbol{\xi}, \boldsymbol{\zeta}; I_1^c, I_2^c)^{N-s-2}\\
    &\times \frac{\Gamma(\boldsymbol{\xi}, \boldsymbol{\zeta}; I_1, I_2)G(\boldsymbol{\xi}; I_1^c, I_1)G(\boldsymbol{\zeta}; I_2^c, I_2)\Gamma(\boldsymbol{\xi}^{-1}, \boldsymbol{\zeta}^{-1}; I_1^c, I_2^c)}{\prod_{1 \leq i <j \leq N}(1 + \xi_i \xi_j - 2 \Delta \xi_i)(1 + \zeta_i \zeta_j - 2 \Delta \zeta_i)} e^{-it(E(\boldsymbol{\xi})-E(\boldsymbol{\zeta}))}\\
     &= \sum_{[\boldsymbol{\xi}],[\boldsymbol{\zeta}]\in\Xi}\ell(y,\boldsymbol{\xi})\mathcal{F}(0;\boldsymbol{\xi},\boldsymbol{\zeta})\ell(y,\boldsymbol{\zeta}).
\end{split}
\end{equation}

\subsection{Miscellaneous}
\subsubsection{Proofs for \texorpdfstring{\Cref{l:BDecomp}}{Lemma 6.4} and \texorpdfstring{\Cref{c:LSW20L6.3}}{Lemma 6.6}}
\label{s:L6.2 and 6.5 proofs}
\begin{proof}[Proof of \Cref{l:BDecomp}]
    We will proceed by splitting the three terms in $B_\sigma(\boldsymbol{\xi})$ into components depending on indices from $I=\{i_1,\ldots,i_{s_1-1}\}$ and $I^c=\{j_1,\ldots,j_{N-s_1+1}\}$ separately, which will be the terms in $B_{\lambda_1}(\boldsymbol{\xi};I)$ and $B_{\lambda_2}(\boldsymbol{\xi};I^c)$, respectively, and mixed terms involving indices from both $I$ and $I^c$ which will be the terms in $G(\boldsymbol{\xi};I,I^c)$. First, we have that \begin{equation}
        \mathrm{sgn}(\sigma)=\mathrm{sgn}(\lambda_1)\mathrm{sgn}(\lambda_2)\prod_{\substack{\alpha>\beta\\\alpha\in I ,\space \beta\in I^c}}(-1).
    \end{equation}
    Moreover, 
    \begin{equation}
        \prod_{j=1}^N\xi_{\sigma(j)}^j = \prod_{j\in I}\xi_{\sigma(j)}^j\prod_{k\in I^c}\xi_{\sigma(k)}^k = \prod_{\alpha=1}^{|I|}\xi_{i_{\lambda_1(\alpha)}}^\alpha\left(\prod_{\beta=1}^{|I^c|}\xi_{j_{\lambda_2(\beta)}}^\beta \prod_{k=|I|+1}^N\xi_{\sigma(k)}^{|I|}\right).
    \end{equation}
    For the product over the set $\{1\leq i < j \leq N\}$ in $B_\sigma(\boldsymbol{\xi})$, we partition this into products over $\{i<j\mid \sigma(i),\sigma(j)\in I\}$, $\{i<j\mid \sigma(i),\sigma(j)\in I^c\}$, and $\{i<j \mid \sigma(i)\in I, \sigma(j)\in I^c\}$. These three products will be parts of $B_{\lambda_1}(\boldsymbol{\xi};I)$, $B_{\lambda_2}(\boldsymbol{\xi};I^c)$ and $G(\boldsymbol{\xi};I,I^c)$ respectively. Note that $\{i<j\mid \sigma(i)\in I^c,\sigma(j)\in I\}=\emptyset$ since $\sigma\left(\{1,\ldots,s_1-1\}\right)=I$ and $\sigma\left(\{s_1,\ldots,N\}\right)=I^c$. More precisely, 
    \begin{equation}
        \begin{split}
            &\prod_{1\leq i < j \leq N}1+\xi_{\sigma(i)}\xi_{\sigma(j)}-2\Delta\xi_{\sigma(i)} = \prod_{\substack{i<j\\\sigma(i),\sigma(j)\in I}}\left(1+\xi_{\sigma(i)}\xi_{\sigma(j)}-2\Delta\xi_{\sigma(i)}\right)\\ &\times\prod_{\substack{i<j\\
            \sigma(i),\sigma(j)\in I^c}}\left(1+\xi_{\sigma(i)}\xi_{\sigma(j)}-2\Delta\xi_{\sigma(i)}\right) 
             \prod_{\substack{i<j\\
            \sigma(i)\in I \\ \sigma(j)\in I^c}}\left(1+\xi_{\sigma(i)}\xi_{\sigma(j)}-2\Delta\xi_{\sigma(i)}\right) \\
            &= \prod_{1\leq \alpha<\beta\leq |I|} \left(1+\xi_{i_{\lambda_1(\alpha)}}\xi_{i_{\lambda_1(\beta)}}-2\Delta\xi_{i_{\lambda_1(\alpha)}}\right) \prod_{1\leq \alpha<\beta\leq |I^c|} \left(1+\xi_{i_{\lambda_2(\alpha)}}\xi_{i_{\lambda_2(\beta)}}-2\Delta\xi_{i_{\lambda_2(\alpha)}}\right)\\&\times\prod_{\substack{\alpha\in I\\ \beta \in I^c}}\left(1+\xi_{\alpha}\xi_{\beta}-2\Delta\xi_{\alpha} \right).
        \end{split}
    \end{equation}
    Then all together we have
    \begin{equation}
        \begin{split}
        B_\sigma(\boldsymbol{\xi}) &= \mathrm{sgn}(\lambda_1) \prod_{1\leq \alpha<\beta\leq |I|} \left(1+\xi_{i_{\lambda_1(\alpha)}}\xi_{i_{\lambda_1(\beta)}}-2\Delta\xi_{i_{\lambda_1(\alpha)}}\right) \prod_{\alpha=1}^{|I|}\xi_{i_{\lambda_1(\alpha)}}^\alpha \\
        &\times \prod_{\substack{\alpha\in I\\ \beta \in I^c}}\left(1+\xi_{\alpha}\xi_{\beta}-2\Delta\xi_{\alpha} \right) \prod_{\substack{\alpha>\beta\\\alpha\in I ,\space \beta\in I^c}}(-1) \\
        &\times \mathrm{sgn}(\lambda_2) \prod_{1\leq \alpha<\beta\leq |I^c|} \left(1+\xi_{i_{\lambda_2(\alpha)}}\xi_{i_{\lambda_2(\beta)}}-2\Delta\xi_{i_{\lambda_2(\alpha)}}\right) \prod_{\beta=1}^{|I^c|}\xi_{j_{\lambda_2(\beta)}}^\beta \left(\prod_{k=|I|+1}^N\xi_k^{|I|}\right) \\
        &=B_{\lambda_1}(\boldsymbol{\xi};I)G(\boldsymbol{\xi};I,I^c)B_{\lambda_2}(\boldsymbol{\xi};I^c)\left(F^\sigma_{s_1,N}(\boldsymbol{\xi}) \right)^{|I|}.
        \end{split}
    \end{equation}
\end{proof}

\begin{proof}[Proof of \Cref{c:LSW20L6.3}]
    Note: this proof closely follows the proof for Lemma 6.3 in \cite{liu_integral_2020}. We will proceed by induction. The base case $n=1$ holds as the left hand side of \eqref{c:LSW20L6.3} evaluates to $-\Gamma(\xi_1,\zeta_1)=\frac{1}{\xi_1\zeta_1-1}$ and the right hand side evaluates to $\frac{1}{\xi_1\zeta_1}\cdot \frac{1}{1-\xi_1^{-1}\zeta_1^{-1}} = \frac{1}{\xi_1\zeta_1-1}$. Now let $n>1$ and suppose \eqref{c:LSW20L6.3} holds for $1, \ldots,n-1$. Then, rewriting the left hand side in terms of sums over the sets $I_{1,1}$ and $I_{2,1}$, we have 
    \begin{equation}\label{e:inductionproofeq}
    \begin{split}
    &\sum_{k=1}^n(-1)^k\sum_{\substack{I_{1,1}\cup\cdots\cup I_{1,k}=\{1,\ldots,n\}\\I_{2,1}\cup\cdots\cup I_{2,k}=\{1,\ldots,n\}}} G(\boldsymbol{\xi};I_{1,1},\ldots,I_{1,k})G(\boldsymbol{\zeta};I_{2,1},\ldots,I_{2,k})\\&\times\prod_{\alpha=1}^kF_{1,e_{\alpha}}(\boldsymbol{\xi},\boldsymbol{\zeta};I_{1,\alpha}, I_{2,\alpha})^{\left(\sum_{i=1}^{\alpha-1}e_i \right)}\Gamma(\boldsymbol{\xi},\boldsymbol{\zeta}; I_{1, \alpha}, I_{2, \alpha})\\
    &=(-1)\sum_{s=1}^n\sum_{\substack{I_{1,1},I_{2,1}\subset[n]\\|I_{1,1}|=|I_{2,1}|=s}}G(\boldsymbol{\xi};I_{1,1},I_{1,1}^c)G(\boldsymbol{\zeta};I_{2,1},I_{2,1}^c)\Gamma(\boldsymbol{\xi},\boldsymbol{\zeta};I_{1,1},I_{2,1})H_{n-s}(\boldsymbol{\xi},\boldsymbol{\zeta};I_{1,1}^c,I_{2,1}^c)
    \end{split}
    \end{equation}
    with 
    \begin{equation}
    \begin{split}
        &H_{n-s}(\boldsymbol{\xi},\boldsymbol{\zeta};I_{1,1}^c,I_{2,1}^c)\\
        &=\sum_{k=2}^{n-s+1}(-1)^{k-1}\sum_{\substack{I_{1,2}\cup\cdots\cup I_{1,k}=[n-s]\\I_{2,2}\cup\cdots\cup I_{2,k}=[n-s]}}G(\boldsymbol{\xi};I_{1,2},\ldots,I_{1,k})G(\boldsymbol{\zeta};I_{2,2},\ldots,I_{2,k}) \\
        &\hspace{40mm}\times\prod_{\alpha=2}^k F_{1,e_\alpha}(\boldsymbol{\xi},\boldsymbol{\zeta};I_{1,\alpha},I_{2,\alpha})^{\left(\sum_{i=1}^{\alpha-1}e_i\right)}\Gamma(\boldsymbol{\xi},\boldsymbol{\zeta};I_{1,\alpha},I_{2,\alpha})
    \end{split}
    \end{equation}
    then by the induction hypothesis 
    \begin{equation}
        H_{n-s}(\boldsymbol{\xi},\boldsymbol{\zeta};I_{1,1}^c,I_{2,1}^c)=\left(F_{1,n-s}(\boldsymbol{\xi},\boldsymbol{\zeta};I_{1,1}^c,I_{2,1}^c)\right)^{2n-s-3} \Gamma(\boldsymbol{\xi}^{-1},\boldsymbol{\zeta}^{-1};I_{1,1}^c,I_{2,1}^c).
    \end{equation}

Note that the right hand side of \eqref{e:inductionproofeq} is a sum from $s=1$ to $n$. For clarity of the argument, let's call what is inside the sum over $s$, $g(s)$. One can check that the right hand side of \eqref{c:LSW20L6.3} is equal to $g(0)$. In other words, our proof amounts to showing that $(-1)\sum_{s=1}^ng(s)=g(0)$. Or equivalently, we must show that $\sum_{s=0}^ng(s)=0$.

Set
\begin{equation}\label{e:inductionproofeqV2}
    \begin{split}
    &f(\boldsymbol{\xi}, \boldsymbol{\zeta})\\
    &:=\sum_{s=0}^n\sum_{\substack{I_{1,1},I_{2,1}\subset[n]\\|I_{1,1}|=|I_{2,1}|=s}}G(\boldsymbol{\xi};I_{1,1},I_{1,1}^c)G(\boldsymbol{\zeta};I_{2,1},I_{2,1}^c)\Gamma(\boldsymbol{\xi},\boldsymbol{\zeta};I_{1,1},I_{2,1})H_{n-s}(\boldsymbol{\xi},\boldsymbol{\zeta};I_{1,1}^c,I_{2,1}^c)\\
    &= (-1)\sum_{s=1}^n\sum_{\substack{I_{1,1},I_{2,1}\subset[n]\\|I_{1,1}|=|I_{2,1}|=s}}\Bigg[G(\boldsymbol{\xi};I_{1,1},I_{1,1}^c)G(\boldsymbol{\zeta};I_{2,1},I_{2,1}^c)\\
    &\hspace{10mm}\times\Gamma(\boldsymbol{\xi},\boldsymbol{\zeta};I_{1,1},I_{2,1})\Gamma(\boldsymbol{\xi}^{-1}, \boldsymbol{\zeta}^{-1}, I_{1,1}^c, I_{2,1}^c) \left(F_{1, n-s}(\boldsymbol{\xi}, \boldsymbol{\zeta}; I_{1,1}^c, I_{2,1}^c)\right)^{2n-s-3}\Bigg]
    \end{split}
    \end{equation}
We want to show that $f(\boldsymbol{\xi}, \boldsymbol{\zeta})=0$. We show this by the following properties:
\begin{itemize}
    \item Polynomial factorization
    \begin{equation}
        f(\boldsymbol{\xi}, \boldsymbol{\zeta}) = \frac{\Delta(\boldsymbol{\xi})\Delta(\boldsymbol{\zeta})}{\prod_{i, j=1}^n (1 - \xi_i \zeta_j)} Q(\boldsymbol{\xi}, \boldsymbol{\zeta})
    \end{equation}
    where $\Delta(\cdot)$ is the Vandremond function and $Q$ is a polynomial on each of the sets of variables, $\boldsymbol{\xi}$ and $\boldsymbol{\zeta}$.
    \item Inverse symmetry 
    \begin{equation}
        f(\boldsymbol{\xi}^{-1}, \boldsymbol{\zeta}^{-1}) = \left(\prod_{i=1}^n\xi_i \zeta_i\right)^{3-2n}f(\boldsymbol{\xi}, \boldsymbol{\zeta})
    \end{equation}
    where $\boldsymbol{\xi}^{-1}= (\xi_1^{-1}, \dots, \xi_n^{-1})$ and similarly for $\boldsymbol{\zeta}^{-1}$.
    \item Residue computation
    \begin{equation}
        \underset{\xi_n = \zeta_i^{-1}}{\mathrm{Res}} f(\boldsymbol{\xi}, \boldsymbol{\zeta}) = 0
    \end{equation}
    for $i =1, \dots, n$.
\end{itemize}

The polynomial factorization follows from the fact that
\begin{equation}
    \Gamma(\boldsymbol{\xi}, \boldsymbol{\zeta}) = \frac{\Delta(\boldsymbol{\xi})\Delta(\boldsymbol{\zeta})}{\prod_{i, j=1}^n (1 - \xi_i \zeta_j)} Q(\boldsymbol{\xi}, \boldsymbol{\zeta})
\end{equation}
where $Q$ is a degree $n-1$ polynomial on each variable and symmetric on each variable set, $\boldsymbol{\xi}$ and $\boldsymbol{\zeta}$. One also needs to check that the expression for $f(\boldsymbol{\xi}, \boldsymbol{\zeta})$ is anti-symmetric for each of the set of variables to obtain the Vandermonde factors. 

The inverse symmetry follows from direct computation by replacing the variables with their inverses, i.e.~$\boldsymbol{\xi} \rightarrow \boldsymbol{\xi}^{-1}$ and $\boldsymbol{\zeta} \rightarrow \boldsymbol{\zeta}^{-1}$, in the expression for $f(\boldsymbol{\xi}, \boldsymbol{\zeta})$ and re-indexing the sum, $s \rightarrow n-s$. 

The inverse symmetry and the polynomial factorization imply the following identity
\begin{equation}
    \left( \prod_{i=1}^n \xi_i \zeta_i\right)^{2(n-1)} Q(\boldsymbol{\xi}^{-1}, \boldsymbol{\zeta}^{-1}) = Q(\boldsymbol{\xi}, \boldsymbol{\zeta}).
\end{equation}
This means that $Q$ is a polynomial of degree $2(n-1)$ with symmetric coefficients so that the coefficient of the term with degree $k$ is equal to the coefficient of the term with degree $2(n-1) - k$ for $k=0, 1, \dots, 2(n-1)$. Moreover, this means that $Q$ is completely determined by $n$ distinct points.  

The residue computation implies that $Q(\boldsymbol{\xi}, \boldsymbol{\zeta}) = 0$ for $n$ points, $\xi_n= \zeta_1^{-1}, \dots, \zeta_n^{-1}$. Thus it follows that $Q(\boldsymbol{\xi}, \boldsymbol{\zeta}) = 0$ and hence $f(\boldsymbol{\xi}, \boldsymbol{\zeta})=0$. The result follows. 
\end{proof}

\subsubsection{On the decomposition of \texorpdfstring{$S_N$}{S	extsubscript{N}}} \label{section: S_N decomp}
Since a permutation in $S_N$ is uniquely determined by its maps restricted to partitions of $\{1,\ldots,N\}$, we can rewrite the symmetric group as \begin{equation}
\begin{split}
    S_N&=\bigcup_{\substack{ I\subset\{1,\ldots,N\}\\
    |I|=s_1-1}
}\{\sigma\in S_N \mid \sigma(\{1,\ldots,s_1-1\})=I\}\\
&=\bigcup_{\substack{
    I\subset\{1,\ldots,N\}\\
    |I|=s_1-1}}S(I),
\end{split}
\end{equation} then we can express a sum over all permutations in $S_N$ as a double sum over subsets $I\subset\{1,\dots,N\}$ with fixed cardinality $|I|=s_1-1$  and permutations in $S(I)$. For instance, 
\begin{equation*}
\begin{split}
    \sum_{\sigma\in S_N}F_{1,s_1-1}^\sigma(\boldsymbol{\xi}) &= \sum_{\substack{
    I\subset\{1,\ldots,N\}\\
    |I|=s_1-1}}\sum_{\sigma\in S(I)} F_{1,s_1-1}^\sigma(\boldsymbol{\xi}) \\
&= \sum_{\substack{
    I\subset\{1,\ldots,N\}\\
    |I|=s_1-1}} F_{i_1,i_{s_1-1}}(\boldsymbol{\xi})\sum_{\sigma\in S(I)}1
\end{split}
\end{equation*}
where $I=\{i_1,\ldots,i_{s_1-1}\}$. The last equality holds since the permutations $\sigma\in S(I)$ map $\{1,\ldots,s_1-1\}$ to $I$. As a result of this decomposition, we simplify the sum in this example from one over the $N!$ permutations of $S_N$ to a sum over the $\binom{N}{s_1-1}$ subsets of $\{1,\ldots,N\}$ of size $s_1-1$.

\section{Proof of \texorpdfstring{\Cref{c:main_v3}}{Conjecture 3.11} for \texorpdfstring{$N=2$}{N=2}}\label{s: N=2 proof}

\begin{assumption}\label{a:N=2 proof assumptions}
    Assume $N=2$, $L>2$, $|\Delta|<\frac{L-1}{2L}$, $\Delta\neq\cos(2\pi k/(L-2))$ with $k\in [L-2]$, and that $\Delta$ is not in a finite set of values as determined in the proof of \Cref{l: xi_1 poles simple}.
\end{assumption}

\begin{theorem}\label{t:N=2 theorem}
     Take \Cref{a:N=2 proof assumptions} to be true. 
     Then \Cref{c:main_v3} holds. That is,
    \begin{equation}
         \mathds{1}(\mathbf{x} = \mathbf{y}) =  \sum_{[\boldsymbol{\xi}] \in \Xi}  \ell(\mathbf{y}, \boldsymbol{\xi}) u(\boldsymbol{\xi}, \mathbf{x}).
    \end{equation}
\end{theorem}
\begin{proof}
The theorem holds by \Cref{l:Indicator function expansion} and $\Cref{l: Bethe function expansion}$.     
\end{proof}

We prove \Cref{c:main_v3} for $N=2$ to elucidate the origin of the general conjecture. The formulas for the $\ell$-function arise from residue computations. Moreover, the residue computations depend on the location of certain poles, which are related to the Bethe roots \eqref{e:bethe_sols}. For $N=2$, we carry out the analysis explicitly in this section culminating in \Cref{l:Indicator function expansion} in \Cref{sub:Indicator function expansion} and \Cref{l: Bethe function expansion} in \Cref{sub:Bethe function expansion}. The general case is more complicated due to technical reasons arising from lack of control of the Bethe roots, and we leave the general case for future work.

The basic strategy is to show that \eqref{e:identity} holds via a nested contour integral:
\begin{equation}
    \mathds{1}(\mathbf{x} = \mathbf{y}) =\frac{1}{(2\pi i)^2}\oint_{C_1}d\xi_1 \oint_{C_2} d\xi_2 \ I(\boldsymbol{\xi};\mathbf{x},\mathbf{y}) =  \sum_{[\boldsymbol{\xi}] \in \Xi}  \ell(\mathbf{y}, \boldsymbol{\xi}) u(\boldsymbol{\xi}, \mathbf{x})
\end{equation}
for contours $C_1$ and $C_2$ and some integrand $I(\boldsymbol{\xi}; \boldsymbol{x}, \boldsymbol{y})$. We define the contours
\begin{equation}\label{e:contour}
\begin{split}
    C_j &= C_{r_j} \sqcup C_{R_j} \\
\end{split}
\end{equation}
where $C_{r_j}$ and $C_{R_j}$ are circles centered at the origin oriented counterclockwise and clockwise, respectively, with the radii $r(C_{r_j})=r_j \ll1$ and $r(C_{R_j})=R_j \gg1$. In particular, let $0<\epsilon \ll \frac12$. We set $r_1=\epsilon$, $R_1=\frac1\epsilon$, $r_2=2\epsilon$, $R_2=\frac{1}{2\epsilon}$. Moreover, for the rest of \Cref{s: N=2 proof}, take \Cref{a:N=2 proof assumptions} to be true. We define
\begin{equation}\label{e:integrand}
\begin{split}
    I(\boldsymbol{\xi};\mathbf{x},\mathbf{y})&= \frac{u(\xi_1,\xi_2;x_1,x_2)\xi_1^{-y_1-1}\xi_2^{-y_2-1}}{\left(\xi_1^L - S(\xi_1,\xi_2) \right)\left(\xi_2^L - S(\xi_2,\xi_1) \right)} \\
    &=  \frac{\xi_1^{x_1-y_1-1}\xi_2^{x_2-y_2-1}+S(\xi_2,\xi_1)\xi_1^{x_2-y_1-1}\xi_2^{x_1-y_2-1}}{\left(\xi_1^L - S(\xi_1,\xi_2) \right)\left(\xi_2^L - S(\xi_2,\xi_1) \right)}
\end{split}
\end{equation}

where $S(\xi_i,\xi_j) = -\frac{1+\xi_i\xi_j-2\Delta\xi_i}{1+\xi_i\xi_j-2\Delta\xi_j}$. To prove the equality \eqref{e:identity} we will show that integrating $I(\boldsymbol{\xi};\mathbf{x},\mathbf{y})$ over the $C_j$ contours with respect to $\xi_1$ and $\xi_2$ (i.e. considering poles 'inside' the contours $C_1$ \& $C_2$) will yield $\mathds{1}(x=_L y)$ (equality up to periodicity), and deforming the $C_{r_j}$ contours to $C_{R_j}$ in $I$ and computing the resulting residues (i.e. considering the poles 'outside' of the contours $C$) will yield $\sum_{\boldsymbol{\xi}\in\Xi}\ell(\mathbf{y},\boldsymbol\xi)u(\boldsymbol{\xi},\mathbf{x})$. 

\subsection{Indicator function expansion}\label{sub:Indicator function expansion}

We start by writing the indicator function on configurations of the periodic XXZ spin chain as a nested contour integral. The result follows by computing the residues of the integral at the origin and at infinity.

We extend the configuration space for the XXZ spin-1/2 chain on a ring. We set
\begin{equation}
    \mathcal{X}' = \{(x_1, x_2) \in \mathbb{Z}^2 \mid x_1< x_2 < x_1+L\}.
\end{equation}
The configuration space $\mathcal{X}$, given by \eqref{e:coordinate_config}, restricts $x_1$ and $x_2$ to be in the set $[L]= \{0, 1, \dots, L-1\}$. Given two configurations $\boldsymbol{x} = (x_1, x_2), \boldsymbol{y}=(y_1, y_2) \in \mathcal{X}'$, we denote
\begin{equation}
    \boldsymbol{x}=_{L} \boldsymbol{y} \Leftrightarrow \begin{cases}
        x_1-y_1 = x_2-y_2=k\, L\\
        x_1-y_2 = x_2-(y_1+ L)= k\, L
    \end{cases}
\end{equation}
for some $k \in \mathbb{Z}$ and we say that two configurations $\boldsymbol{x} = (x_1, x_2), \boldsymbol{y}=(y_1, y_2) \in \mathcal{X}'$ are \emph{equivalent}\footnote{The $d\xi_i$ are always assumed to have a factor of $\frac{1}{2\pi i}$}.

\begin{lemma}\label{l:Indicator function expansion}
Let $\boldsymbol{x},\boldsymbol{y} \in \mathcal{X}'$ and take \Cref{a:N=2 proof assumptions} to be true. Then, we have the following contour integral representation on the indicator function of equivalent configurations,
\begin{equation}\label{e:N=2ProofIndicator}
    \mathds{1}(\boldsymbol{x}=_{L} \boldsymbol{y}) = \oint_{C_1} d\xi_1 \oint_{C_2} d\xi_2 \, I(\boldsymbol{\xi};\mathbf{x},\mathbf{y}),
\end{equation}
where the contours $C_j$ are given by \eqref{e:contour} and the integrand $I$ is given by \eqref{e:integrand}.
\end{lemma}

\begin{proof}
We denote

\begin{equation}
\begin{split}
    &\mathds{1}(x=_L y) \coloneq \sum_{k\in\mathbb{Z}}\mathds{1}\left(x_1+kL=y_1,x_2+kL=y_2 \right) + \mathds{1}\left(x_2+(k-1)L=y_1,x_1+kL=y_2 \right).
\end{split}
\end{equation}
We claim that the non-permuted terms in the integrand will yield the indicator functions where both $x_1$ and $x_2$ are shifted by $kL$, and we claim that the permuted terms will yield the indicator functions where just $x_1$ and just $x_2$ are shifted by a factor of $L$ apart. That is,   
\begin{equation}\label{e:N=2ProofIdentityPart}
    \oint_{C_1}d\xi_1\oint_{C_2}d\xi_2 \frac{\xi_1^{x_1-y_1-1}\xi_2^{x_2-y_2-1}}{\left(\xi_1^L - S(\xi_1,\xi_2) \right)\left(\xi_2^L - S(\xi_2,\xi_1) \right)} = \sum_{k\in\mathbb{Z}}\mathds{1}(\mathbf{x}+kL=\mathbf{y})
\end{equation}
and 
\begin{equation}\begin{split}
    \oint_{C_1}d\xi_1\oint_{C_2}d\xi_2 \frac{S(\xi_2,\xi_1)\xi_1^{x_2-y_1-1}\xi_2^{x_1-y_2-1}}{\left(\xi_1^L - S(\xi_1,\xi_2) \right)\left(\xi_2^L - S(\xi_2,\xi_1) \right)} = \sum_{k\in\mathbb{Z}}\mathds{1}\left(x_2+(k-1)L=y_1,x_1+kL=y_2 \right).
    \end{split}
\end{equation}
Moreover, integrating with respect to the $C_{r_j}$ and $C_{R_j}$ contours separately for $\xi_1$ and $\xi_2$ will yield the terms of the sum over non-negative and negative $k$ terms respectively. 

Let's begin with the nested integrals over the $C_{r_j}$ contours and non-permuted terms. Because $r(C_{r_1})=r_j<1$, we have that $\left|\xi_i^L S(\xi_j,\xi_i) \right| \approx (r_j)^L<1$, so we can rewrite the denominator in the left hand side of \eqref{e:N=2ProofIdentityPart} as geometric sums\footnote{Note: the property $S(\xi_1,\xi_2)=(S(\xi_2,\xi_1))^{-1}$ will be frequently used without specific mention in this proof sketch.}. That is, 
\begin{equation}\label{e:N=2proofGeometricSum}
\begin{split}
    &\oint_{C_{r_1}}d\xi_1\oint_{C_{r_2}}d\xi_2 \frac{\xi_1^{x_1-y_1-1}\xi_2^{x_2-y_2-1}}{\left(\xi_1^L - S(\xi_1,\xi_2) \right)\left(\xi_2^L - S(\xi_2,\xi_1) \right)}\\ = &\sum_{k,m=0}^\infty \oint d\xi_1\oint d\xi_2\hspace{0.1cm}\xi_1^{x_1-y_1-1+kL}\xi_2^{x_2-y_2-1+mL}S(\xi_2,\xi_1)^{k-m}.
\end{split}
\end{equation}

To handle the sum, we will consider three cases: when $k>m$, when $k<m$, and $k=m$. For the first two cases, we will use the two possible change of variables for $\eta=\xi_1\xi_2$. Given $k<m$, we let $\xi_1=\eta/\xi_2$. Then we have that the corresponding terms from $\eqref{e:N=2proofGeometricSum}$ are equal to 
\begin{equation}
    \begin{split}
        \sum_{m>k}^\infty\oint_{C_{r_1\cdot r_2}}d\eta\oint_{C_{r_2}}d\xi_2 \ \eta^{x_1-y_1-1+kL}\xi_2^{x_2-x_1+y_1-y_2-1+(m-k)(L-1)}\left( \frac{\xi_2+\eta\xi_2-2\Delta\eta}{1+\eta-2\Delta\xi_2} \right)^{m-k}.
    \end{split}
\end{equation}
Focusing on $\xi_2$, our only possible pole inside the $C_{r_2}$ contour is at $\xi_2=0$ if its exponent is negative. Inspecting the exponent, we find that 
\begin{equation}
    x_2-x_1+y_2-y_1-1+(m-k)(L-1)\geq -(L-1)+(m-k)(L-1)\geq 0
\end{equation}
since $L>1$ and $m>k$. Thus there are no poles with respect to the the $\xi_2$ variable and the integrals in the $k<m$ case are zero. For $k>m$, we let $\xi_2=\eta/\xi_1$ and have a similar story. The corresponding terms in this case from \eqref{e:N=2proofGeometricSum} are equal to 
\begin{equation}
    \begin{split}
        \sum_{k>m}^\infty\oint_{C_{r_1\cdot r_2}}d\eta\oint_{C_{r_2}}d\xi_2 \ \xi_1^{x_1-x_2+y_2-y_1-1+(k-m)(L-1)}\eta^{x_2-y_2-1+mL}\left( \frac{\xi_1+\eta\xi_1-2\Delta\eta}{1+\eta-2\Delta\xi_1} \right)^{k-m}.
    \end{split}
\end{equation}

With a possible pole at $\xi_1=0$, we inspect the exponent of $\xi_1$ and find that 
\begin{equation}
    x_1-x_2+y_2-y_1-1+(k-m)(L-1)\geq -(L-1)+(k-m)(L-1)\geq0.
\end{equation}
Thus there are no poles with respect to the $\xi_1$ variable and the integrals in the $k>m$ case are zero. Now from \eqref{e:N=2proofGeometricSum} we are left with 
\begin{equation}
    \sum_{k=0}^{\infty}\oint_{C_{r_1}} d\xi_1\oint_{C_{r_2}} d\xi_2 \ \xi_1^{x_1-y_1-1+kL}\xi_2^{x_2-y_2-1+kL}.
\end{equation}
For a fixed $k\in\mathbb{Z}_{\geq 0}$, if $x_j+kL>y_j$ for $j=1$ or $j=2$, then there are no poles in the above integrand in $\xi_1$ or $\xi_2$ respectively, and the integral is $0$. If $x_j+kL = y_j$ for $j=1$ and $j=2$, have simple poles in $\xi_1$ and $\xi_2$ and a residue computation shows that
\begin{equation}
    \oint_{C_{r_1}} d\xi_1\oint_{C_{r_2}} d\xi_2 \ \xi_1^{x_1-y_1-1+kL}\xi_2^{x_2-y_2-1+kL} = \oint d\xi_1\oint d\xi_2 \ \xi_1^{-1}\xi_2^{-1} =1.
\end{equation}
If $x_j+kL<y_j$ for $j=1$ or $j=2$, there is a pole of order $-(x_j-y_j-1+kL)$ in $\xi_j$ and a residue computation shows that
\begin{equation}
     \oint_{C_{r_1}} d\xi_1\oint_{C_{r_2}} d\xi_2 \ \xi_1^{x_1-y_1-1+kL}\xi_2^{x_2-y_2-1+kL} = 0.
\end{equation}
Thus, we have that 
\begin{equation}
    \oint_{C_{r_1}}d\xi_1\oint_{C_{r_2}}d\xi_2 \frac{\xi_1^{x_1-y_1-1}\xi_2^{x_2-y_2-1}}{\left(\xi_1^L - S(\xi_1,\xi_2) \right)\left(\xi_2^L - S(\xi_2,\xi_1) \right)} = \sum_{k=0}^\infty \mathds{1}(x+kL = y)
\end{equation}
where $x+kL = (x_1+kL, x_2+kL)$.

Turning to the permuted terms, since $|\xi_j|< 1$ in the $C_{r_j}$ contours, we have
\begin{equation}\begin{split}
&\oint_{C_{r_1}}d\xi_1\oint_{C_{r_2}}d\xi_2 \frac{S(\xi_2,\xi_1)\xi_1^{x_2-y_1-1}\xi_2^{x_1-y_2-1}}{\left(\xi_1^L - S(\xi_1,\xi_2) \right)\left(\xi_2^L - S(\xi_2,\xi_1) \right)}\\
&=\sum_{k,m=0}^\infty \oint_{C_{r_1}}d\xi_1\oint_{C_{r_2}}d\xi_2 \xi_1^{x_2-y_1-1}\xi_2^{x_1-y_2-1}S(\xi_2,\xi_1)^{k-m+1}.
\end{split}
\end{equation}
Now we will split the sum into three possible cases: $k-m+1<0$, $k-m+1>0$, and $k-m+1=0$. For $k-m+1<0$, we use the change of variables $\xi_1=\frac{\eta}{\xi_2}$ which gives us
\begin{equation}
    \begin{split}
        &\sum_{k-m+1<0}\oint_{C_{r_1\cdot r_2}}d\eta\oint_{C_{r_2}}d\xi_2 \ \eta^{x_2-y_1-1+kL}\xi_2^{x_1-x_2+y_1-y_2-1+(m-k)L}\left(\frac{1+\eta-2\Delta\eta/\xi_2}{1+\eta-2\Delta\xi_2} \right)^{m-k-1} \\
        &= \sum_{k-m+1<0}\oint_{C_{r_1\cdot r_2}}d\eta\oint_{C_{r_2}}d\xi_2 \ \eta^{x_2-y_1-1+kL}\xi_2^{x_1-x_2+y_1-y_2+(m-k)L-(m-k)}\left(\frac{\xi_2+\eta\xi_2-2\Delta\eta}{1+\eta-2\Delta\xi_2} \right)^{m-k-1}.
    \end{split}
\end{equation}
Since $m-k>1$ in this case, we have the following lower bound on the $\xi_2$ exponent $x_1-x_2+y_1-y_2+(m-k)(L-1)\geq -2(L-1)+(m-k)(L-1)\geq 0$. Thus, there are no poles in the $\xi_2$ integral and the terms in this case are all zero. For $k-m+1>0$, we use the change of variables $\xi_2=\frac{\eta}{\xi_1}$. Then we have 
\begin{equation}
    \begin{split}
        &\sum_{k-m+1>0}\oint_{C_{r_1}}d\xi_1\oint_{C_2}d\eta \ \xi_1^{x_2-x_1+y_2-y_1-1+(k-m)L}\eta^{x_1--y_2-1+mL}\left(\frac{1+\eta-2\Delta\eta/\xi_1}{1+\eta-2\Delta\xi_1} \right)^{k-m+1} \\
        &=\sum_{k-m+1>0}\oint_{C_{r_1}}d\xi_1\oint_{C_{r_2}}d\eta \ \xi_1^{x_2-x_1+y_2-y_1+(k-m)L-(k-m)-2}\eta^{x_1--y_2-1+mL}\left(\frac{\xi_1+\eta\xi_1-2\Delta\eta}{1+\eta-2\Delta\xi_1} \right)^{k-m+1}.
        \end{split}
\end{equation}
Since $k-m\geq0$ in this case, we have the following lower bound on the $\xi_1$ exponent $x_2-x_1+y_2-y_1-2+(k-m)L-(k-m)\geq (k-m)(L-1) \geq 0$. Hence there are no poles in the $\xi_1$ integral, and the terms in this case are zero. Finally for $k-m+1 = 0$ we have the remaining terms
\begin{equation}
    \sum_{k=0}^\infty\oint_{C_{r_1}}d\xi_1\oint_{C_{r_2}}d\xi_2\ \xi_1^{x_2-y_1-1+kL}\xi_2^{x_1-y_2-1+(k+1)L}.
\end{equation}
If $x_1 + (k+1)L=y_2$ and $x_2+kL=y_1$, then we have 
\begin{equation}
    \oint_{C_{r_1}}d\xi_1\oint_{C_{r_2}}d\xi_2 \ \xi_1^{-1}\xi_2^{-1} = 1.
\end{equation}
If $x_1 + (k+1)L<y_2$ or $x_2+kL<y_1$, then a residue computation yields zero. And if $x_1 + (k+1)L>y_2$ or $x_2+kL>y_1$, then $\xi_1$ or $\xi_2$ do not have a pole and the integral is zero. Thus, 
\begin{equation}
\begin{split}
    \oint_{C_{r_1}}d\xi_1\oint_{C_{r_2}}d\xi_2 \frac{S(\xi_2,\xi_1)\xi_1^{x_2-y_1-1}\xi_2^{x_1-y_2-1}}{\left(\xi_1^L - S(\xi_1,\xi_2) \right)\left(\xi_2^L - S(\xi_2,\xi_1) \right)}&=\sum_{k=0}^\infty\mathds{1}\left(x_2+kL=y_1,\ x_1+(k+1)L=y_2 \right)\\
    &=\sum_{k=1}^\infty\mathds{1}\left(x_2+(k-1)L=y_1,\ x_1+kL=y_2 \right)
    \end{split}
\end{equation}

Our integration over the $C_{R_j}$ contours of $I(\xi;x,y)$ will be nearly identical to that of the $C_{r_j}$ contours. The key difference is that we obtain geometric sums using the fact that $|\xi_j^{-L} S(\xi_j,\xi_i)|\approx (r_j)^L <1$ on $C_{R_j}$ and then we can swap the sum over non-negative integers to one over negative integers, and then the integrands are essentially the same as in the case of integrating over $C_{r_j}$ contours. Beginning with non permuted terms, we have
\begin{equation}
    \begin{split}
        &\oint_{C_{R_1}}d\xi_1\oint_{C_{R_2}}d\xi_2 \frac{\xi_1^{x_1-y_1-1}\xi_2^{x_2-y_2-1}}{\left(\xi_1^L - S(\xi_1,\xi_2) \right)\left(\xi_2^L - S(\xi_2,\xi_1) \right)} \\ 
        &=\sum_{k,m=0}^\infty \oint_{C_{R_1}}d\xi_1\oint_{C_{R_2}}d\xi_2 \ \xi_1^{x_1-y_1-1-(k+1)L}\xi_2^{x_2-y_2-1-(m+1)L}S(\xi_2,\xi_1)^{m-k} \\ 
        &= \sum_{k,m\leq0} \oint_{C_{R_1}}d\xi_1\oint_{C_{R_2}}d\xi_2 \ \xi_1^{x_1-y_1-1+(k-1)L}\xi_2^{x_2-y_2-1+(m-1)L}S(\xi_2,\xi_1)^{k-m}.
    \end{split}
\end{equation}

As we did previously, we consider three cases: $k<m$, $k>m$, and $k=m$. For $k<m$, we use the change of variables $\xi_2=\frac{\eta}{\xi_1}$. Then we have 
\begin{equation}
    \begin{split}
        &\sum_{k<m\leq 0} \oint_{C_{R_1}}d\xi_1\oint_{C_{R_1\cdot R_2}}d\eta \ \xi_1^{x_1-x_2+y_2-y_1-1+(k-m)L} \eta^{x_2-y_2-1+(m-1)L}\left(\frac{1+\eta-2\Delta\xi_1}{1+\eta-2\Delta\eta/\xi_1} \right)^{m-k} \\
        &=\sum_{k<m\leq0} \oint_{C_{R_1}}d\xi_1\oint_{C_{R_1\cdot R_2}}d\eta \ \xi_1^{x_1-x_2+y_2-y_1-1+(k-m)(L-1)} \eta^{x_2-y_2-1+(m-1)L}\left(\frac{1/\xi_1+\eta/\xi_1-2\Delta}{1+\eta-2\Delta\eta/\xi_1} \right)^{m-k}.
    \end{split}
\end{equation}
 Since $r(C_{R_1})>1$, potential poles for $\xi_1$ exist at $+\infty$ when its exponent is positive. Since $k-m<0$ in this case, we have the following upper bound on the exponent of $\xi_1$: $x_1-x_2+y_2-y_1-1+(k-m)(L-1)\leq -2 + L-1 +(k-m)(L-1) \leq -2$. Hence there are no poles with respect to $\xi_1$ and the integrals in this case are zero. For $k>m$ we let $\xi_1=\frac{\eta}{\xi_2}$. Then we have 
\begin{equation}
    \begin{split}
        &\sum_{0\geq k>m} \oint_{C_{R_1 \cdot R_2}}d\eta\oint_{C_{R_2}}d\xi_2 \ \eta^{x_1-y_1-1+kL} \xi_2^{x_2-x_1+y_1-y_2-1+(m-k)L}\left(\frac{1+\eta-2\Delta\xi_2}{1+\eta-2\Delta\eta/\xi_2} \right)^{k-m} \\
        &= \sum_{0\geq k>m} \oint_{C_{R_1\cdot R_1}}d\eta\oint_{C_{R_2}}d\xi_2 \ \eta^{x_1-y_1-1+kL} \xi_2^{x_2-x_1+y_1-y_2-1+(m-k)(L-1)}\left(\frac{1/\xi_2+\eta/\xi_2-2\Delta}{1+\eta-2\Delta\eta/\xi_2} \right)^{k-m}.
    \end{split}
\end{equation}
Since $m-k<0$, we have the following upper bound on the $\xi_2$ exponent: $x_2-x_1+y_1-y_2-1+(m-k)(L-1)\leq L-1 - 2 + (m-k)(L-1) \leq -2$. Thus there are no poles with respect to $\xi_2$ and the integrals in this case are zero. Lastly we have the case where $k=m$. Then we have the remaining terms 
\begin{equation}
    \begin{split}
        \sum_{k\leq 0}\oint_{C_{R_1}}d\xi_1 \oint_{C_{R_2}}d\xi_2 \ \xi_1^{x_1-y_1 - 1 +(k-1)L}\xi_2^{x_2-y_2-1+(k-1)L}.
    \end{split}
\end{equation}

If $x_j + (k-1)L = y_j$ for $j=1,2$, we have 
\begin{equation}
    \oint_{C_{R_1}}d\xi_1\oint_{C_{R_2}}d\xi_2  \ \xi_1^{-1}\xi_2^{-1} = 1.
\end{equation}
If $x_j - y_j + (k-1)L>0$ for $j=1$ or $j=2$ then a quick residue computation yields zero for the integrals in this case. And if $x_j-y_j+(k-1)L<0$ for $j=1$ or $j=2$, then there are no poles for $\xi_1$ or $\xi_2$, and the integrals in this case are also zero. Thus, 
\begin{equation}
    \oint_{C_{R_1}}d\xi_1\oint_{C_{R_2}}d\xi_2 \frac{\xi_1^{x_1-y_1-1}\xi_2^{x_2-y_2-1}}{\left(\xi_1^L - S(\xi_1,\xi_2) \right)\left(\xi_2^L - S(\xi_2,\xi_1) \right)}=\sum_{k\leq -1}\mathds{1}\left(x_1+kL=y_1, \, x_2+kL=y_2\right).
\end{equation}

Following the same process for the permuted terms, we begin with
\begin{equation}
    \begin{split}
        &\oint_{C_{R_1}}d\xi_1\oint_{C_{R_2}}d\xi_2 \frac{S(\xi_2,\xi_1)\xi_1^{x_2-y_1-1}\xi_2^{x_1-y_2-1}}{\left(\xi_1^L - S(\xi_1,\xi_2) \right)\left(\xi_2^L - S(\xi_2,\xi_1) \right)} \\
        &= \sum_{k,m\leq 0} \oint_{C_{R_1}}d\xi_1\oint_{C_{R_2}}d\xi_2 \xi_1^{x_2-y_1-1+(k-1)L}\xi_2^{x_1-y_2-1+(m-1)L}S(\xi_2,\xi_1)^{k-m+1}.
    \end{split}
\end{equation}

We separate the sum into three cases: $k-m+1>0$, $k-m+1<0$, and $k-m+1 = 0$. For $k-m+1>0$ we use the change of variables $\xi_1=\frac{\eta}{\xi_2}$ and have 
\begin{equation}
    \begin{split}
        &\sum_{0\geq k>m-1} \oint_{C_{R_1 \cdot R_2}}d\eta\oint_{C_{R_2}}d\xi_2 \ \eta^{x_2-y_1-1+(k-1)L}\xi_2^{x_1-x_2+y_1-y_2-1+(m-k)L}\left(\frac{1+\eta-2\Delta\xi_2}{1+\eta-2\Delta\eta/\xi_2}\right)^{k-m+1} \\
        &=\sum_{0\geq k>m-1}\oint_{C_{R_1\cdot R_2}}d\eta\oint_{C_{R_2}}d\xi_2 \ \eta^{x_2-y_1-1+(k-1)L}\xi_2^{x_1-x_2+y_1-y_2+(m-k)(L-1)}\left(\frac{1/\xi_2+\eta/\xi_2-2\Delta}{1+\eta-2\Delta\eta/\xi_2}\right)^{k-m+1}.
    \end{split}
\end{equation}

Since $m-k<1$, we have the following upper bound on the $\xi_2$ exponent: $x_1-x_2 +y_1-y_2+(m-k)(L-1)\leq -2 + (m-k)(L-1) \leq -2$. Thus there are no poles with respect to $\xi_2$ and the integrals in this case are zero. For $k-m+1<0$ we use the change of variables $\xi_2=\frac{\eta}{\xi_1}$ and have 
\begin{equation}
    \begin{split}
        &\sum_{0\geq m>k+1}\oint_{C_{R_1}}d\xi_1\oint_{C_{R_1\cdot R_2}}d\eta \ \xi_1^{x_2-x_1+y_2-y_1-1+(k-m)L}\eta^{x_1-y_2-1+(m-1)L}\left(\frac{1+\eta-2\Delta\xi_1}{1+\eta-2\Delta\eta/\xi_1}\right)^{m-k-1} \\
        &=\sum_{0\geq m>k+1}\oint_{C_{R_1}}d\xi_1\oint_{C_{R_1\cdot R_2}}d\eta \ \xi_1^{x_2-x_1+y_2-y_1-2+(k-m)(L-1)}\eta^{x_1-y_2-1+(m-1)L}\left(\frac{1/\xi_1+\eta/\xi_1-2\Delta}{1+\eta-2\Delta\eta/\xi_1}\right)^{m-k-1}.
    \end{split}
\end{equation}
Since $k-m<-1$, we have the following upper bound on the $\xi_1$ exponent: $x_2-x_1+y_2-y_1-2 + (k-m)(L-1) \leq 2(L-1) + (k-m)(L-1)-2 \leq -2$. Hence there are no poles with respect to $\xi_1$ and the integrals in this case are zero. For $k-m+1=0$, we have
\begin{equation}
    \sum_{k\leq -1}\oint_{C_{R_1}}d\xi_1\oint_{C_{R_2}}d\xi_2 \ \xi_1^{x_2-y_1-1+(k-1)L}\xi_2^{x_1-y_2-1+kL}.
\end{equation}

If $x_2+(k-1)L=y_1$ and $x_1+kL=y_2$, then we have 
\begin{equation}
    \oint_{C_{R_1}}d\xi_1\oint_{C_{R_2}}d\xi_2 \ \xi_1^{-1}\xi_2^{-1} = 1.
\end{equation}

Otherwise, we have negative exponents that yield no poles with respect to $\xi_1$ or $\xi_2$, or we have positive exponents on both $\xi_1$ and $\xi_2$ and straightforward residue computations yield zero. Thus, 
\begin{equation}
    \oint_{C_{R_1}}d\xi_1\oint_{C_{R_2}}d\xi_2 \frac{S(\xi_2,\xi_1)\xi_1^{x_2-y_1-1}\xi_2^{x_1-y_2-1}}{\left(\xi_1^L - S(\xi_1,\xi_2) \right)\left(\xi_2^L - S(\xi_2,\xi_1) \right)} = \sum_{k\leq -1}\mathds{1}\left(x_2+(k-1)L=y_1, \ x_1+kL=y_2 \right).
\end{equation}

Up to this point, we have shown 
\begin{equation}
    \begin{split}
        &\oint_{C_{r_1}}d\xi_1\oint_{C_{r_2}}d\xi_2 \ \frac{\xi_1^{x_1-y_1-1}\xi_2^{x_2-y_2-1}+ S(\xi_2,\xi_1)\xi_1^{x_2-y_1-1}\xi_2^{x_1-y_2-1}}{\left( \xi_1^L-S(\xi_1,\xi_2) \right)\left( \xi_2^L-S(\xi_2,\xi_1) \right)} \\
        &= \sum_{k=0}^\infty \mathds{1}\left( x_1+kL=y_1,  \ x_2+kL=y_2 \right) + \sum_{k=1}^\infty\mathds{1}\left( x_2+(k-1)L=y_1,  \ x_1+kL=y_2 \right) \\ \text{and}\\
        &\oint_{C_{R_1}}d\xi_1\oint_{C_{R_2}}d\xi_2 \ \frac{\xi_1^{x_1-y_1-1}\xi_2^{x_2-y_2-1}+ S(\xi_2,\xi_1)\xi_1^{x_2-y_1-1}\xi_2^{x_1-y_2-1}}{\left( \xi_1^L-S(\xi_1,\xi_2) \right)\left( \xi_2^L-S(\xi_2,\xi_1) \right)}\\
        &= \sum_{k\leq -1} \mathds{1}\left( x_1+kL=y_1,  \ x_2+kL=y_2 \right) + \mathds{1}\left( x_2+(k-1)L=y_1,  \ x_1+kL=y_2 \right).
    \end{split}
\end{equation}

Integrating over the $C_{r_1}$ contour for $\xi_1$ and $C_{R_2}$ contour for $\xi_2$, we begin with
\begin{equation}
    \begin{split}
        &\oint_{C_{r_1}}d\xi_1\oint_{C_{R_2}}d\xi_2 \ \frac{\xi_1^{x_1-y_1-1}\xi_2^{x_2-y_2-1}}{\left( \xi_1^L-S(\xi_1,\xi_2) \right)\left( \xi_2^L-S(\xi_2,\xi_1) \right)} \\
        &=-\sum_{k,m=0}^{\infty}\oint_{C_{r_1}}d\xi_1\oint_{C_{R_2}}d\xi_2 \ \xi_1^{x_1-y_1-1+kL}\xi_2^{x_2-y_2-1-(m+1)L}S(\xi_2,\xi_1)^{k+m+1} 
    \end{split}
\end{equation}
since $|\xi_1^LS(\xi_2,\xi_1)|\approx \frac23|\Delta|\epsilon^{L-1}< 1$ and $|\xi_2^{-L}S(\xi_2,\xi_1)|\approx \frac23|\Delta|(2\epsilon)^{L-1}< 1$. For $\xi_2$, we have potential poles at $+\infty$. Using the change of variables $\xi_1=\frac{\eta}{\xi_2}$, we get 
\begin{equation}
    \sum_{k,m=0}^{\infty}\oint_{C_{r=1/2}}d\eta\oint_{C_{R_2}}d\xi_2 \ \eta^{x_1-y_1-1+kL}\xi_2^{x_2-x_1+y_1-y_2-(k+m+1)(L-1)-1}\left(\frac{1/\xi_2+\eta/\xi_2-2\Delta}{1+\eta-2\Delta\eta/\xi_2}\right)^{k+m+1}.
\end{equation}
Then we have the following upper bound on the $\xi_2$ exponent: $x_2-x_1+y_1-y_2-(k+m+1)(L-1)-1\leq L-1-2-(k+m+1)(L-1) \leq -2$. Hence there are no poles with respect to $\xi_2$ and the above integrals are zero. Moving along, we compute the permuted terms
\begin{equation}
    \begin{split}
        &\oint_{C_{r_1}}d\xi_1\oint_{C_{R_2}}d\xi_2 \ \frac{S(\xi_2,\xi_1)\xi_1^{x_2-y_1-1}\xi_2^{x_1-y_2-1}}{\left( \xi_1^L-S(\xi_1,\xi_2) \right)\left( \xi_2^L-S(\xi_2,\xi_1) \right)} \\
        &=-\sum_{k,m=0}^\infty\oint_{C_{r_1}}d\xi_1\oint_{C_{R_2}}d\xi_2 \ \xi_1^{x_2-y_1-1+kL}\xi_2^{x_1-y_2-1-(m+1)L}S(\xi_2,\xi_1)^{k+m+2}.
    \end{split}
\end{equation}
 Using the change of variables $\xi_1=\eta/\xi_2$ we obtain
 \begin{equation}
     \begin{split}
         -\sum_{k,m=0}^\infty\oint_{C_{r=1/2}}d\eta\oint_{C_{R_2}}d\xi_2 \ \eta^{x_2-y_1-1+kL}\xi_2^{x_1-x_2+y_1-y_2-(k+m+1)(L-1)}\left(\frac{1/\xi_2+\eta/\xi_2-2\Delta}{1+\eta-2\Delta\eta/\xi_2}\right)^{k+m+2}.
     \end{split}
 \end{equation}
 Then we have the following upper bound on the $\xi_2$ exponent: $x_1-x_2+y_1-y_2-(k+m+1)(L-1)\leq -2-(k+m+1)(L-1)$. So there are no poles with respect to the $\xi_2$ variable and the above integrals are zero. Lastly we must compute the nested integral with $\xi_1$ over the $C_{R_1}$ contour and $\xi_2$ over the $C_{r_2}$ contour. We have 
 \begin{equation}
     \begin{split}
         &\oint_{C_{R_1}}d\xi_1\oint_{C_{r_2}}d\xi_2 \ \frac{\xi_1^{x_1-y_1-1}\xi_2^{x_2-y_2-1}}{\left( \xi_1^L-S(\xi_1,\xi_2) \right)\left( \xi_2^L-S(\xi_2,\xi_1) \right)} \\
         &= -\sum_{k,m=0}^{\infty}\oint_{C_{R_1}}d\xi_1\oint_{C_{r_2}}d\xi_2 \ \xi_1^{x_1-y_1-1-(k+1)L}\xi_2^{x_2-y_2-1+mL}S(\xi_1,\xi_2)^{k+m+1}
     \end{split}
 \end{equation}
 since $|\xi_1^{-L}S(\xi_1,\xi_2)|\approx \frac23|\Delta|\epsilon^{L-1}<1$ and $|\xi_2^LS(\xi_1,\xi_2)|\approx \frac43|\Delta| (2\epsilon)^{L-1}<1$. Using the change of variables $\xi_2=\eta/\xi_1$ we obtain 
 \begin{equation}
     \begin{split}
         -\sum_{k,m=0}^{\infty}\oint_{C_{R_1}}d\xi_1\oint_{C_{r=2}}d\eta \ \xi_1^{x_1-x_2+y_2-y_1-1-(k+m+1)(L-1)}\eta^{x_2-y_2-1+mL}\left(\frac{1/\xi_1+\eta/\xi_1-2\Delta}{1+\eta-2\Delta\eta/\xi_1}\right)^{k+m+1}.
     \end{split}
 \end{equation}
 On the exponent of $\xi_1$ we have the following upper bound: $x_1-x_2+y_2-y_1-1-(k+m+1)(L-1) \leq -2 +(L-1) -(k+m+1)(L-1)\leq -2$. Thus there are no poles with respect to $\xi_1$ so the above integrals are zero. On to the permuted terms, we have
 \begin{equation}
     \begin{split}
         &\oint_{C_{R_1}}d\xi_1\oint_{C_{r_2}}d\xi_2 \ \frac{S(\xi_2,\xi_1)\xi_1^{x_2-y_1-1}\xi_2^{x_1-y_2-1}}{\left( \xi_1^L-S(\xi_1,\xi_2) \right)\left( \xi_2^L-S(\xi_2,\xi_1) \right)} \\
         &= -\sum_{k,m=0}^{\infty}\oint_{C_{R_1}}d\xi_1\oint_{C_{r_2}}d\xi_2 \ \xi_1^{x_2-y_1-1-(k+1)L}\xi_2^{x_1-y_2-1+mL}S(\xi_1,\xi_2)^{k+m}.
     \end{split}
 \end{equation}
Using the change of variables $\xi_2=\eta/\xi_1$, we obtain
\begin{equation}
    \begin{split}
    \sum_{k,m=0}^{\infty}\oint_{C_{R_1}}d\xi_1\oint_{C_{2}}d\eta \ \xi_1^{x_2-x_1+y_2-y_1-1-(k+m+1)L - 1 + k + m}\eta^{x_1-y_2-1+mL}\left(\frac{1/\xi_1+\eta/\xi_1-2\Delta}{1+\eta-2\Delta \eta/\xi_1} \right)^{k+m}.
    \end{split}
\end{equation}
 We have the following upper bound on the $\xi_1$ exponent: 
 $x_2-x_1+y_2-y_1-(k+m)(L-1)-L-1
 \leq (2-k-m)(L-1)-L-1$.
 This upper bound is only positive for $k=m=0$. Otherwise, the exponent of $\xi_1$ is negative and the above integrals with respect to $\xi_1$ have no poles and are zero. For $k=m=0$, we are left with the nested integral 
 \begin{equation}
     -\oint_{C_{R_1}}d\xi_1\oint_{C_{r_2}}d\xi_2 \ \xi_1^{x_2-y_1-1-L}\xi_2^{x_1-y_2-1}.
 \end{equation}
 Then we deform the $C_{R_1}$ contour to $C_{r_1}$ and obtain 
\begin{equation}
    \oint_{C_{r_1}}d\xi_1\oint_{C_{r_2}}d\xi_2 \ \xi_1^{x_2-y_1-1-L}\xi_2^{x_1-y_2-1} = \mathds{1}(x_2-L=y_1, \ x_1=y_2).
\end{equation}

To piece everything together now, from evaluating our nested integrals with respect to poles inside the $C_1$ and $C_2$ contours, we have the following equalities that we've shown 
\begin{equation}
    \begin{split}
        \oint_{C_{r_1}}d\xi_1\oint_{C_{r_2}}d\xi_2 I(\boldsymbol{\xi};\mathbf{x},\mathbf{y}) &= \sum_{k=0}^\infty \mathds{1}(x+kL=y) + \sum_{k=1}^\infty\mathds{1}(x_2+(k-1)L=y_1,\ x_1+kL=y_2) \\
        \oint_{C_{R_1}}d\xi_1\oint_{C_{R_2}}d\xi_2 I(\boldsymbol{\xi};\mathbf{x},\mathbf{y}) &= \sum_{k\leq -1} \mathds{1}(x+kL=y) + \mathds{1}(x_2+(k-1)L=y_1,\ x_1+kL=y_2) \\
        \oint_{C_{r_1}}d\xi_1\oint_{C_{R_2}}d\xi_2 I(\boldsymbol{\xi};\mathbf{x},\mathbf{y}) &= 0 \\
        \oint_{C_{R_1}}d\xi_1\oint_{C_{r_2}}d\xi_2 I(\boldsymbol{\xi};\mathbf{x},\mathbf{y}) &= \mathds{1}(x_2-L=y_1,\ x_1=y_2).
    \end{split}
\end{equation}
Hence we have shown that 
\begin{equation}
\begin{split}
    \oint_{C_1}d\xi_1\oint_{C_2}d\xi_2 \ I(\boldsymbol{\xi};\mathbf{x},\mathbf{y}) &= \sum_{k\in\mathbb{Z}}\mathds{1}(x+kL=y) + \mathds{1}(x_2+(k-1)L=y_1,\ x_1+kL=y_2)\\
    &= \mathds{1}(x=_L y).
\end{split}
\end{equation}
\end{proof}

\subsection{Bethe function expansion}\label{sub:Bethe function expansion}
Now that we have evaluated the potential poles inside of the $C$ contours and obtained the above result, we shall evaluate the potential poles outside of the $C$ contours and aim to show the following lemma. 

\begin{lemma}\label{l: Bethe function expansion}
Take \Cref{a:N=2 proof assumptions} to be true. Then
    \begin{equation}
      \oint_{C_1}d\xi_1\oint_{C_2}d\xi_2 \ I(\boldsymbol{\xi};\mathbf{x},\mathbf{y})=\sum_{(\xi_1,\xi_2)\in\Xi}
      \ell(\mathbf{y}, \boldsymbol{\xi})u(\boldsymbol{\xi}, \mathbf{x})
\end{equation}
where the nested contour integral on the left hand side is the same as the one in \Cref{l:Indicator function expansion}.
\end{lemma}

\begin{proof}
    The result holds by \Cref{l: contour to residue sums}, \Cref{l: multipart sums}, and \Cref{l:last Bethe expansion equality}.
\end{proof}

To obtain the required lemmas in the above proof, we analyze residues "outside" the contours $C_1$ and $C_2$. That is, residues not near $0$ or $+\infty$. We first note where our possible poles are outside of the $C_1$ and $C_2$ contours. Inspecting $I(\boldsymbol{\xi};\mathbf{x},\mathbf{y})$, we find possible poles when $S(\xi_2,\xi_1)=\infty$, $\xi_1^L-S(\xi_1,\xi_2)=0$, and $\xi_2^L-S(\xi_2,\xi_1)=0$. We rule out the first possible pole $S(\xi_2,\xi_1)=\infty$ since the same term with the same order is in the denominator of $I(\boldsymbol{\xi};\mathbf{x},\mathbf{y})$. For the other two possibilities, we will use the following Lemmas 

\begin{lemma}\label{l: AllPoles}
Take \Cref{a:N=2 proof assumptions} to be true. Then there are $L$ simple poles from $\xi_2^L-S(\xi_2,\xi_1)=0$ with respect to $\xi_2$ which lie outside the contour $C_{r_2}$. I.e. the poles are inside of the annulus with inner radius $2 \epsilon$ and outer radius $\frac{1}{2\epsilon}$, denoted $A_2$.
\end{lemma}

\begin{proof}
    We first rewrite $\xi_2^L-S(\xi_2,\xi_1) =0 $ as 
    \begin{equation}\label{e: xi2 pole}
        \frac{\left(\xi_2^L+1 \right)\left(\xi_1\xi_2 +1 \right)}{1+\xi_1\xi_2-2\Delta\xi_1} + \frac{-2 \Delta\xi_2\left( \xi_1 \xi_2^{L-1} + 1\right)}{1+\xi_1\xi_2-2\Delta\xi_1} = 0.
    \end{equation}
    We call the first and second terms of \cref{e: xi2 pole} $f(\xi_2)$ and $g(\xi_2;\Delta)$, respectively; where $
    \xi_1$ and $\Delta$ are fixed parameters. In this form, we are effectively treating $g$ as a $\Delta$-perturbation of $f$. Now the core idea is using Rouch\'e's theorem to show that $f+g$ has exactly $L$ zeros in $A_2$ from the perturbed roots of unity from $f$ and the $L+1$'th zero lies outside of $A_2$. Since we will eventually integrate over the $C_1$ contour with respect to $\xi_1$, we must consider the cases where $\xi_1\in C_{r_1}$ and $\xi_1\in C_{R_1}$. To bound the locations of the zeros of $f+g$ via Rouch\'e's theorem, we will consider the cases when $\xi_2\in C_{r_2}$ and $\xi_2\in C_{R_2}$. So we have four cases in total to consider. By Rouch\'e's theorem, if $|f|>|g|$ for $\xi_2$ on a circle on which $f$ and $g$ are holomorphic, then $f+g$ has the same difference of poles from zeros in the interior of the circle as $f$. Since $f$ and $g$ have the same lone pole, $1+\xi_1\xi_2-2\Delta\xi_1=0$, then they have the same number of zeros. We note that $f$ has $L$ zeros at the roots of unity and one more zero at $\xi_2 = -\xi_1^{-1} \in C_{R_1}$. For $\xi_1\in C_{r_1}$ and $\xi_2\in C_{r_2}$, we have 
    \begin{equation}
        |f|(1+\xi_1\xi_2-2\Delta\xi_1) \geq \left(1-(2\epsilon)^L \right)\left(1-2\epsilon^2 \right)
    \end{equation}
    and 
    \begin{equation}
        |g|(1+\xi_1\xi_2-2\Delta\xi_1) \leq 4|\Delta|\epsilon\left(1+2^{L-1}\epsilon^L\right).
    \end{equation}

Hence we have that $|f|>|g|$ when 
\begin{equation}
    |\Delta| < \frac{\left(1-(2\epsilon)^L \right)\left(1-2\epsilon^2 \right)}{4\epsilon\left(1+2^{L-1}\epsilon^L\right)} < \frac{1}{4\epsilon}.
\end{equation}
 So in this case $f$ has no zeros inside the $C_{r_2}$ contour and thus neither does $f+g$. For $\xi_1\in C_{r_1}$ and $\xi_2\in C_{R_2}$ we have that 
\begin{equation}
    |f|(1+\xi_1\xi_2-2\Delta\xi_1) \geq \frac12\left| 1-(2\epsilon)^{-L} \right|
\end{equation}
and  
\begin{equation}
    |g|(1+\xi_1\xi_2-2\Delta\xi_1) \leq |\Delta|\epsilon^{-1}\left(1+2^{-L+1}\epsilon^{-L+2} \right). 
\end{equation}
Then $|f|>|g|$ when 
\begin{equation}
    |\Delta| < \frac{\left| (2\epsilon)^L-1 \right|}{2\epsilon\left(1+2^{L-1}\epsilon^{L-2} \right)} < \frac{1}{2\epsilon}.
\end{equation}

Thus for $\xi_1\in C_{r_1}$, $f+g$ has the same number of zeros in the interior of $C_{R_2}$ contour as $f$, which is $L$. So we have that $f+g$ has $L$ zeros in $A_2$ since there are none in the interior of the $C_{r_2}$ contour. For the case when $\xi_1\in C_{R_1}$ and $\xi_2\in C_{r_2}$, we have that 
\begin{equation}
    |f|(1+\xi_1\xi_2-2\Delta\xi_1)\geq \left| 1 - (2\epsilon)^L\right|
\end{equation}
and 
\begin{equation}
    |g|(1+\xi_1\xi_2-2\Delta\xi_1) \leq 4|\Delta|\epsilon\left(2^{L-1}\epsilon^{L-2}+1\right). 
\end{equation}

Then $|f| > |g|$ when 
\begin{equation}
    |\Delta| < \frac{\left|1-(2\epsilon)^L\right|}{4\epsilon\left(1+2^{L-1}\epsilon^{L-2}\right)}< \frac{1}{4\epsilon}
\end{equation}

Lastly when $\xi_1\in C_{R_1}$ and $\xi_2\in C_{R_2}$ we have that 
\begin{equation}
    |f|(1+\xi_1\xi_2-2\Delta\xi_1) \geq \left|1-(2\epsilon)^{-L}\right|\left|1-\frac12\epsilon^{-2}\right|
\end{equation}
and 
\begin{equation}
    |g|(1+\xi_1\xi_2-2\Delta\xi_1)\leq |\Delta|\epsilon^{-1}\left(1+2^{-L+1}\epsilon^{-L}\right).
\end{equation}
Then $|f|>|g|$ when 
\begin{equation}
    |\Delta|<\frac{\left|2\epsilon^2-1\right|\left|(2\epsilon)^L-1\right|}{4\epsilon\left|1+2^{L-1}\epsilon^L\right|}<\frac{1}{4\epsilon},
\end{equation}

which holds since $|\Delta|<\frac{L-1}{L}\epsilon$ and $|\epsilon|<\frac12$. Thus, we ensure that there are exactly $L$ poles from $\xi_2^L-S(\xi_2,\xi_1)=0$ which are located in the annulus $A_2$. 
\end{proof}

\begin{lemma} \label{l: noncontributing pole}
    Take \Cref{a:N=2 proof assumptions} to be true. Then the pole $\xi_1^L-S(\xi_1,\xi_2)=0$ with respect to $\xi_2$ lies inside the contour $C_2$. I.e. the pole is outside of the annulus $A_2$.
\end{lemma}

\begin{proof}[Proof of \cref{l: noncontributing pole}]
    The potential pole $\xi_1^L-S(\xi_1,\xi_2)=0$ with respect to $\xi_2$ can be solved explicitly. In particular, this pole is given by 
\begin{equation}\label{e: xi2 pole not in annulus}
    \xi_2 =-\frac{1+\xi_1^L-2\Delta\xi_1}{\xi_1\left(1+\xi_1^L-2\Delta \xi_1^{L-1}\right)}.
\end{equation}
For $\xi_1\in C_{r_1}$, we have the following approximate bound on \eqref{e: xi2 pole not in annulus}
\begin{equation}
    |\xi_2| \gtrapprox \frac1\epsilon-2|\Delta|>\frac{1}{2\epsilon}
\end{equation}
since $|\Delta|<\frac{L-1}{L}\epsilon$. For $\xi_1\in C_{R_1}$, \begin{equation}
    |\xi_2| \approx \frac{\epsilon}{1-2|\Delta|\epsilon}<2\epsilon.
\end{equation}
Hence the pole \eqref{e: xi2 pole not in annulus} is inside the contours $C_{r_2}$ and $C_{R_2}$, and their residues do not contribute in these computations. 
\end{proof}

\begin{lemma}\label{l: simple poles}
    Take \Cref{a:N=2 proof assumptions} to be true. Then the $L$ poles from \Cref{l: AllPoles} are simple (i.e. have a multiplicity of one). 
\end{lemma}

\begin{proof}
    Assume $z_2$ satisfies  $z_2^L-S(z_2,\xi_1)=0$. To determine if $z_2$ has a multiplicity greater than one, we inspect whether $Lz_2^{L-1}-\frac{\partial S(z_2,\xi_1)}{\partial z_2}=0$. Substituting $z_2^L=S(z_2,\xi_1)$ we obtain 
    \begin{equation}
        Lz_2^{-1}S(z_2,\xi_1)-\frac{\partial S(z_2,\xi_1)}{\partial z_2} = 0. 
    \end{equation}
     After clearing denominators, this equation is quadratic in $z_2$. 
    To bound the location of the solutions to this quadratic equation, we inspect the leading order behavior of a series expansion in the $\xi_1$ variable for $\xi_1\in C_{r_1}$ and $\xi_1\in C_{R_1}$. For $\xi_1 \in C_{r_1}$ we obtain 
    \begin{equation}
       |z_2| \approx \frac{L-1}{L}\frac{1}{\epsilon} 
    \end{equation}
    and
    \begin{equation}
      |z_2| \approx \frac{L}{2|\Delta|(L-1)} 
    \end{equation}

    For $\xi_1\in C_{R_1}$ we obtain 
    \begin{equation}
        |z_2| \approx \frac{2|\Delta|(L-1)}{L} 
    \end{equation}
    and 
    \begin{equation}
        |z_2| \approx \frac{L}{L-1}\epsilon.
    \end{equation}
    
To ensure the poles are simple, we want any double roots to be outside of the annulus $A_2$. That is, we want $|z_2|< 2\epsilon$ or $|z_2|> \frac{1}{2\epsilon}$. Since $|\Delta|< \frac{L-1}{L}\epsilon$, this condition is satisfied and thus the $L$ poles in $A_2$ must be simple. 

\end{proof}

We note that the approximations in the prior two proofs can easily be made exact but present them in this way for clarity of exposition as the exact forms do not elucidate the results any further. 

By \Cref{l: AllPoles} and \Cref{l: noncontributing pole}, integrating $I(\boldsymbol{\xi};\textbf{x},\textbf{y})$ with respect to $\xi_2$ over $C_2$ by deforming the $C_{r_2}$ contour to $C_{R_2}$ yields the following residue sum
\begin{equation}\label{e: xi2 residue sum}
    \sum_{z_2\in Q(\xi_1)} \frac{u(\xi_1,z_2; x_1,x_2)\xi_1^{-y_1-1}z_2^{-y_2-1}}{\left(\xi_1^L-S(\xi_1,z_2)\right)\partial_{\xi}\left(\xi^L-S(\xi,\xi_1)\right)_{\xi=z_2}}
\end{equation}
where $Q(\xi)=\left\{ z  \mid z^L-S(z,\xi) = 0 \right\}$. Now when we integrate with respect to $\xi_1$ deforming the $C_{r_1}$ contour to $C_{R_1}$, we have to consider poles from two sources: $\xi_1^L-S(\xi_1,z_2)=0$ and $\partial_{\xi}\left(\xi^L-S(\xi,\xi_1)\right)_{\xi=z_2}=0$. We need to consider the latter source now because the roots $z_2\in Q(\xi_1)$ depend on $\xi_1$, so we introduce the possibility of higher multiplicity roots when deforming the $C_{r_1}$ contour. To take care of this we will show the following lemma, adapted from \cite{li_contour_2023} (Lemma 4.15). 

\begin{lemma}\label{l: double roots don't contribute}
    Take \Cref{a:N=2 proof assumptions} to be true. Let $q_\xi(w)=w^L-S(w,\xi)$ and suppose $\xi_0\in \mathbb{C}$ such that $\partial_{w}^j q_{\xi_0}(w)=0$ for $j=1,\ldots,k-1$ and with $w\in Q(\xi_0)$. 
    Let 
    \begin{equation}
        f(\xi) \coloneq \sum_{\lambda\in Q(\xi)} \frac{u(\xi,\lambda;\textbf{x})}{q_\xi'(\lambda)}g(\xi,\lambda;\textbf{y})
    \end{equation}
    where $g$ is analytic in a neighborhood of $\xi_0$. Then $f(\xi)$ is holomorphic on a neighborhood of $\xi_0$. 
\end{lemma}

\begin{proof}
    To show $f$ is holomorphic on a neighborhood of $\xi_0$, we must first show it is well-defined as a function on $\xi\in\mathbb{C}$. A priori the roots of $q_\xi(w)$ are defined on a branched cover of $\mathbb{C}$ since the coefficients of $q_\xi(w)$ depend on $\xi$. However, we do not need to choose branch cuts for the roots of $q_\xi(w)$, because it is a symmetric function on the roots. More specifically, choose a base point $\xi\neq\xi_0$ and let $\gamma:[0,1]\mapsto\mathbb{C}$ be a closed curve around $\xi_0$ such that $\gamma(0)=\xi$ and $\gamma(1)=\xi$. Denote the roots of $q_\xi(w)$ by $\lambda_i$ and fix a labeling of the roots. That is, $Q(\xi)=\{\lambda_1,\ldots,\lambda_L\}$. Let $\lambda_i(t)$ be the analytic continuation of $\lambda_i$ along $\gamma$. We have that $\lambda_i(0)=\lambda_i$ and $\lambda_i(1)=\lambda_{\sigma(i)}$ for some $\sigma\in S_L$. But since $f(\xi)$ is defined as a sum over the roots $\lambda_i$ with each summand independent of the other roots, we find that $f(\gamma(0))=f(\gamma(1))$. Moreover, the denominator of $f$ only vanishes for a small neighborhood of $\xi_0$. Thus $f$ is independent of the analytic continuation on $\gamma$ and well-defined on a punctured disc $B_0=\left\{\xi\mid 0<|\xi-\xi_0|<\delta_0 \right\}$ for some $\delta_0>0$ small enough. 

    Moreover, we have that $f(\xi)$ is locally holomorphic on $B_0$. So then for any closed curve $\gamma\subset B_0$ containing $\xi_0$, we can deform $\gamma$ to a small circle centered at $\xi_0$,
    \begin{equation} \label{e: double root integral}
        \oint_{\gamma}f(\xi)d\xi = \oint_{|\xi-\xi_0|=\delta}f(\xi)d\xi
    \end{equation}
     for $0<\delta<\delta_0$. Now to show that $f$ is holomorphic, we will uses Morera's theorem \cite{Gamelin_2001}, where it suffices to show the limit as $\delta\to0$ of the right hand side of \eqref{e: double root integral} is zero. Denote the $k$ segments of the circle $|\xi-\xi_0|=\delta$ in between the $k$-roots of unity by $\gamma_j(\xi_0,\delta)$ for $j=1,\ldots,k$. Moving forward we use the decomposition 
     \begin{equation}
         \oint_{|\xi-\xi_0|=\delta}f(\xi)d\xi=\sum_{j=1}^k\oint_{\gamma_j(\xi_0,\delta)} f(\xi)d\xi 
     \end{equation}
    in order to express the higher multiplicity roots of $q_\xi(w)$ as a single valued series expansion. We want to show 
    \begin{equation} \label{e: deltalimit integral}
        \lim_{\delta\to0}\oint_{\gamma_j(\xi_0,\delta)} f(\xi)d\xi=0.
    \end{equation} for all $j=1,\ldots,k$.
    Consider $\xi\in\gamma_j(\xi_0,\delta)$ for some $j\in\{1,\ldots,k\}$ and $\delta$ arbitrarily small. In this region we may fix a labeling of the roots $Q(\xi)=\{\lambda_1,\ldots,\lambda_L\}$ such that each $\lambda_i=\lambda_i(\xi)$ is a single valued holomorphic function. The function $f$ is independent of the labeling since it is symmetric with respect to the roots $\lambda_i$. The roots are pairwise distinct and may only have multiplicity higher than one when $\xi=\xi_0$. For roots of multiplicity $k\geq2$, by \Cref{l: Bethe Puiseux} we can distinguish $k$ roots by \begin{equation} \label{e: rootExpansion}
        \lambda_i=\lambda_0+\alpha_i(\xi-\xi_0)^{\frac1k} +\mathcal{O}\left((\xi-\xi_0)^{\frac2k}\right)
    \end{equation}
    for $i=1,\ldots,k$, some constant $\alpha_i\in \mathbb{C}$ and $\xi\in\gamma_j(\xi_0,
    \delta)$. Now we will show that \cref{e: deltalimit integral} holds. Note that we are only looking in a neighborhood of higher multiplicity poles. So the terms of $f$ containing $\lambda_j$ for $j>k$ are $\mathcal{O}(1)$. By \cref{e: rootExpansion} we have 
    \begin{equation}
        q_\xi'(\lambda_j) = C_j(\xi-\xi_0)^{1-\frac1k} +\mathcal{O}(\xi-\xi_0) 
    \end{equation}
    for some $C_j\in \mathbb{C}$ independent of $\xi$. Then, \begin{equation}
        |f(\xi)|\leq C|\xi-\xi_0|^{\frac1k-1}
    \end{equation}
    for some constant $C$ since $u$ and $g$ are holomorphic in a small enough neighborhood of $\xi_0$. Moreover, we have   
    \begin{equation}
        \begin{split}
            \left| \oint_{\gamma_j(\xi_0,\delta)}f(\xi)d\xi\right| \leq \oint_{\gamma_j(\xi_0,\delta)} \left| f(\xi) \right|d\xi \leq C\,\text{length}(\gamma_j)|\xi-\xi_0|^{\frac1k-1}.
        \end{split}
    \end{equation}
    And since the length of $\gamma_j(\xi_0,\delta)$ is $\mathcal{O}(|\xi-\xi_0|)$, we have that \begin{equation}
        \lim_{\delta\to0}\, \text{length}(\gamma_j)|\xi-\xi_0|^{\frac1k-1} = 0.
    \end{equation}
    This establishes \eqref{e: deltalimit integral} and together with \eqref{e: double root integral} means $\oint_\gamma f(\xi)d\xi=0$. Then by Morera's theorem the lemma is true.  
\end{proof}

\begin{lemma}\label{l: xi_1 poles simple}
    The poles from $\xi_1^L-S(\xi_1,z_2)=0$ are simple except for finitely many values of $\Delta$. 
\end{lemma}
\begin{proof}
     If the poles were not simple, we would have the following system of equations
    \begin{equation}\label{e: system of eqs nonsimple xi1}
        \begin{split}
             z_2^L-S(z_2,\xi_1)=0 \\
            \xi_1^L-S(\xi_1,z_2(\xi_1))=0 \\
            \partial_{\xi_1}\left(\xi_1^L-S(\xi_1,z_2(\xi_1)\right)=0 
        \end{split}
    \end{equation}
    The first two equations imply that $(\xi_1 z_2)^L = 1$ which means $\xi_1 z_2 = \eta_k$ where $\eta_k$ is an $L$'th root of unity. That is, $\eta_k = e^{2\pi i k/L}$ for $0\leq k\leq L-1$. We claim that for any such $k$, there are finitely many solutions to the above system of equations. Fix such a $k$.
    Using $\xi_1 z_2  = \eta_k$, we will write $\xi_1^L-S(\xi_1,z_2)=0$ and its derivative with respect to $\xi_1$ into two equations that does not involve $z_2$. Note that since $L$ is odd, $\eta_k\neq -1$.
    We first have that \begin{equation}\label{e: xi1 BE derivative}
       \partial_{\xi_1}\left(\xi_1^L-S(\xi_1,z_2(\xi_1)\right)= L\xi_1^{L-1}-\frac{\partial S(\xi_1,z_2)}{\partial \xi_1} - \frac{\partial S(\xi_1,z_2)}{\partial z_2} \frac{\partial z_2}{\partial \xi_1}=0.
    \end{equation} 
    Since $\xi_1 z_2 = \eta_k$, we have that $\frac{\partial z_2}{\partial \xi_1} = -\eta_k \xi_1^{-2}$. Substituting this into \ref{e: xi1 BE derivative} along with $z_2=\eta_k \,\xi_1^{-1}$ and some algebra, we obtain 
    \begin{equation}
        g_1(\xi_1,\Delta) \coloneq  a(\Delta)\xi_1^2 + b(\Delta)\xi_1 + c(\Delta) = 0
    \end{equation}
where 
\begin{equation}
    \begin{split}
        a(\Delta) &=-2\Delta(L(1+\eta_k) + 1) \\ 
        b(\Delta)&= (L(1+\eta_k)^2 + (2\Delta)^2(L+1)\eta_k) \\
        c(\Delta) &= \eta_k(1-2\Delta(L+1)\eta_k). 
    \end{split}
\end{equation}
Substituting $z_2 = \eta_k \,\xi_1^{-1}$ into $\xi_1^L-S(\xi_1,z_2)=0$ and simplifying, we obtain 
\begin{equation}
    g_2(\xi_1,\Delta) \coloneq (1+\eta_k)\xi_1^L -2\Delta\eta_k \,\xi_1^{L-1}-2\Delta \xi_1 +1+\eta_k=0 
\end{equation}
Now the common roots of $g_1$ and $g_2$ are precisely solutions to our system of equations \eqref{e: system of eqs nonsimple xi1} along with $\xi_1 z_2=\eta_k$. So it suffices to show that $g_1$ and $g_2$ have a finite number of common roots $(\xi_1,\Delta)$. To do this we consider $g_1,g_2$ as polynomials in $\xi_1$ with coefficients in $\mathbb{C}[\Delta]$, and we will use facts about resultants to bound the number of roots. Specifically, the resultant $\text{Res}_\Delta(g_1,g_2)$ gives a polynomial in $\Delta$ whose roots are precisely the $\Delta$ values where $g_1$ and $g_2$ have common roots \cite{Cox_Little_OShea_1998}. This means $g_1$ and $g_2$ can only have common roots for infinitely many values of $\Delta$ if the resultant is identically zero. Then checking that $\text{Res}_\Delta(g_1,g_2)\neq 0$ for a specific value of $\Delta$ is sufficient to show $g_1$ and $g_2$ have shared roots for finitely many values of $\Delta$. Pick $\Delta=0$. Then 
\begin{equation}
    g_1(\xi_1) = L(1+\eta_k)^2\xi_1 + \eta_k
\end{equation}
and 
\begin{equation}
    g_2(\xi_1) = (1+\eta_k)\xi_1^L+1+\eta_k.
\end{equation}

From \cite{Cox_Little_OShea_1998}, we can compute the resultant explicitly as 
\begin{equation}
    \text{Res}_{\Delta=0}(g_1,g_2)=\left(L(1+\eta_k)^2\right)^L g_2\left(\frac{-\eta_k}{L(1+\eta_k)^2}\right)
\end{equation}
which simplifies to 
\begin{equation}
    \text{Res}(g_1,g_2) = (1+\eta_k)\left(L^L(1+\eta_k)^{2(L-1)}-1\right).
\end{equation}
Now we show that $\text{Res}(g_1,g_2)\neq0$. We begin by writing 
\begin{equation}\label{e: Resultant eqn zero}
    L^L(1+\eta_k)^{2(L-1)}-1=L^L\left(1+e^{2\pi i k /L}\right)^{2(L-1)}-e^{2\pi i}.
\end{equation}
Since
\begin{equation}
    \begin{split}
        1+e^{i\theta} &= e^{i\theta/2}\left(e^{-i\theta/2} + e^{i\theta/2} \right) \\
        &= e^{i\theta/2}(2\cos(\theta/2)),
    \end{split}
\end{equation}
we may write 
\begin{equation}
    L^L\left(1+e^{2\pi i k /L}\right)^{2(L-1)}-e^{2\pi i}=L^L\left(2\cos(k\, \pi/L) \right)^{2(L-1)}e^{\frac{2\pi(L-1)k}{L} i}-e^{2\pi i}. 
\end{equation}
For the above to be zero, $L$ would have to divide $(L-1)k$. Since $L$ and $L-1$ are coprime, $L$ would have to divide $k$. We know that $0\leq k\leq L-1$, so $k=0$ is the only value which $L$ divides. But for $k=0$, $\eta_0=1$ and our left hand side of \Cref{e: Resultant eqn zero} becomes $L^L\,2^{2(L-1)}$ which is strictly greater than $1$. This means \Cref{e: Resultant eqn zero} cannot be zero and thus $\text{Res}(g_1,g_2)\neq0$ for $\Delta=0$. Therefore the resultant of $g_1$ and $g_2$ is not identically zero and $g_1$ and $g_2$ share roots only for a finite number of $\Delta$ values.
\end{proof}

\begin{lemma} \label{l: contour to residue sums}
Take \Cref{a:N=2 proof assumptions} to be true. Then
    \begin{equation}
        \begin{split}
            &\oint_C d\xi_1 \oint_C d\xi_2 \ \frac{\xi_1^{-y_1-1}\xi_2^{-y_2-1}u(\xi_1,\xi_2; x_1,x_2)}{\left( \xi_1^L-S(\xi_1,\xi_2) \right)\left( \xi_2^L-S(\xi_2,\xi_1) \right)} \\
            &= \sum_{\substack{(z_1,z_2):\\ z_2^L-S(z_2,z_1) =0
            \\ z_1^L-S(z_1,z_2)=0}} \frac{z_1^{-y_1-1}z_2^{-y_2-1}u(z_1,z_2; x_1,x_2)}{\partial_{\xi}\left[\xi^L-S(\xi,z_2(\xi))\right]_{\xi=z_1}\partial_{\xi}\left[\xi^L-S(\xi,z_1)\right]_{\xi=z_2(z_1)}} 
        \end{split}
    \end{equation}
\end{lemma}
\begin{proof}
    By \Cref{l: AllPoles} and \Cref{l: noncontributing pole}, there are $L$ simple poles in $A_2$ from $\xi_2^L-S(\xi_2,\xi_1)$. By \Cref{l: double roots don't contribute}, $I(\boldsymbol{\xi};\textbf{x},\textbf{y})$ is holomorphic in a neighborhood of any higher multiplicity roots of $\xi_2^L-S(\xi_2,\xi_1)$, and thus the only contributing poles with respect to $\xi_1$ are from $\xi_1^L-S(\xi_1,z_2)=0$ which are simple by \Cref{l: xi_1 poles simple}. Thus the lemma holds via residues of simple poles.  
\end{proof}

\begin{lemma} \label{l: multipart sums}

    \begin{equation}
    \begin{split}
        &\sum_{\substack{(z_1,z_2):\\ z_2^L-S(z_2,z_1) =0
            \\ z_1^L-S(z_1,z_2)=0}} \frac{z_1^{-y_1-1}z_2^{-y_2-1}u(z_1,z_2; x_1,x_2)}{\partial_{\xi}\left[\xi^L-S(\xi,z_2(\xi))\right]_{\xi=z_1}\partial_{\xi}\left[\xi^L-S(\xi,z_1)\right]_{\xi=z_2(z_1)}} \\
            &= \sum_{\substack{(z_1,z_2):\\ z_2^L-S(z_2,z_1) =0
            \\ z_1^L-S(z_1,z_2)=0}} \frac{z_1^{-y_1-1}z_2^{-y_2-1}u(z_1,z_2; x_1,x_2)}{\det \Lambda(z_1,z_2)}
    \end{split}
    \end{equation}
   
\end{lemma}

To show the above lemma we will need the following.
\begin{lemma} \label{l: denominator to lambda determinant}
    Denote $B(x,y) = Lx^{-1}S(x,y) - \frac{\partial S(x,y)}{\partial x}$. Then,
    \begin{equation}
        B(z_1,z_2)B(z_2,z_1) - \frac{\partial S(z_1,z_2)}{\partial z_2}\frac{\partial S(z_2,z_1)}{\partial z_1} = \det \Lambda(z_1,z_2)
    \end{equation}
\end{lemma}

\begin{proof}
    \begin{equation}
        \begin{split}
            &B(z_1,z_2)B(z_2,z_1) - \frac{\partial S(z_1,z_2)}{\partial z_2}\frac{\partial S(z_2,z_1)}{\partial z_1} \\
            &= \left( Lz_1^{-1}S(z_1,z_2) - \frac{\partial S(z_1,z_2)}{\partial z_1}\right)\left(Lz_2^{-1}S(z_2,z_1) - \frac{\partial S(z_2,z_1)}{\partial z_2} \right) - \frac{\partial S(z_1,z_2)}{\partial z_2}\frac{\partial S(z_2,z_1)}{\partial z_1} \\
            &= S(z_1,z_2)S(z_2,z_1)\left( Lz_1^{-1} - \frac{\partial S(z_1,z_2)}{\partial z_1}S(z_2,z_1)\right)\left(Lz_2^{-1} - \frac{\partial S(z_2,z_1)}{\partial z_2}S(z_1,z_2) \right) - \frac{\partial (S(z_2,z_1)^{-1})}{\partial z_2}\frac{\partial (S(z_1,z_2)^{-1})}{\partial z_1} \\ 
            &=\left( Lz_1^{-1} - \frac{\partial S(z_1,z_2)}{\partial z_1}S(z_2,z_1)\right)\left(Lz_2^{-1} - \frac{\partial S(z_2,z_1)}{\partial z_2}S(z_1,z_2) \right)\\
            &\hspace{20mm} - \frac{1}{S(z_2,z_1)^2S(z_1,z_2)^2}\frac{\partial S(z_2,z_1)}{\partial z_2}S(z_1,z_2)S(z_2,z_1)\frac{\partial S(z_1,z_2)}{\partial z_1} \\
            &=\left( Lz_1^{-1} - \frac{\partial S(z_1,z_2)}{\partial z_1}S(z_2,z_1)\right)\left(Lz_2^{-1} - \frac{\partial S(z_2,z_1)}{\partial z_2}S(z_1,z_2) \right) - \frac{\partial S(z_1,z_2)}{\partial z_1}S(z_2,z_1)\frac{\partial S(z_2,z_1)}{\partial z_2}S(z_1,z_2) \\
            &= \det \Lambda(z_1,z_2). 
        \end{split}
    \end{equation}
\end{proof}

\begin{proof}[Proof of \Cref{l: multipart sums}]
    We begin with the terms in the denominators. 
    \begin{equation}
        \partial_\xi(\xi^L-S(\xi,z_2(\xi))_{\xi=z_1} = Lz_1^{L-1} - \frac{\partial S(z_1,z_2)}{\partial z_2}\frac{\partial z_2}{\partial z_1} - \frac{\partial S(z_1,z_2)}{\partial z_1}
    \end{equation}
    To obtain $\frac{\partial z_2}{\partial z_1}$ we differentiate $z_2^L-S(z_2,z_1)=0$ with respect to $z_1$. We have
    \begin{equation}
        Lz_2^{-1}S(z_2,z_1)\frac{\partial z_2}{\partial z_1}-\frac{\partial S(z_2,z_1)}{\partial z_2} \frac{\partial z_2}{\partial z_1}-\frac{\partial S(z_2,z_1)}{\partial z_1} = 0
    \end{equation}
    and so 
    \begin{equation}
        \frac{\partial z_2}{\partial z_1} = \left(\frac{\partial S(z_2,z_1)}{\partial z_1}\right)/\left(Lz_2^{-1}S(z_2,z_1)-\frac{\partial S(z_2,z_1)}{\partial z_2} \right) = \left(\frac{\partial S(z_2,z_1)}{\partial z_1}\right)/B(z_2,z_1).
    \end{equation}
    Now we have 
    \begin{equation} \label{e: z1 BE first deriv}
        \begin{split}
            \partial_\xi(\xi^L-S(\xi,z_2(\xi))_{\xi=z_1} &= Lz_1^{-1}S(z_1,z_2) - \frac{\frac{\partial S(z_1,z_2)}{\partial z_2}\frac{\partial S(z_2,z_1)}{\partial z_1}}{B(z_2,z_1)} - \frac{\partial S(z_1,z_2)}{\partial z_1} \\
            &= B(z_1,z_2) - \frac{\frac{\partial S(z_1,z_2)}{\partial z_2}\frac{\partial S(z_2,z_1)}{\partial z_1}}{B(z_2,z_1)}.
        \end{split}
    \end{equation}
    Our other term in the denominator is
    \begin{equation}
        \partial_\xi (\xi^L-S(\xi,z_1))_{\xi=z_2(z_1)} = Lz_2^{-1}S(z_2,z_1) - \frac{\partial S(z_2,z_1)}{\partial z_2} = B(z_2,z_1). 
    \end{equation}
    So then 
    \begin{equation}
       (\partial_\xi(\xi^L-S(\xi,z_2(\xi))_{\xi=z_1}) (\partial_\xi (\xi^L-S(\xi,z_1))_{\xi=z_2(z_1)}) = B(z_1,z_2)B(z_2,z_1) - \frac{\partial S(z_1,z_2)}{\partial z_2}\frac{\partial S(z_2,z_1)}{\partial z_1}.
    \end{equation}
    Thus, by \eqref{l: denominator to lambda determinant}, 
    \begin{equation}
        (\partial_\xi(\xi^L-S(\xi,z_2(\xi))_{\xi=z_1}) (\partial_\xi (\xi^L-S(\xi,z_1))_{\xi=z_2(z_1)})  = \det\Lambda(z_1,z_2). 
    \end{equation}
\end{proof}

To finish showing \Cref{l: Bethe function expansion}, we must establish the last equality.
\begin{lemma}\label{l:last Bethe expansion equality}
    \begin{equation}
    \sum_{\substack{(z_1,z_2):\\ z_2^L-S(z_2,z_1)=0\\ z_1^L-S(z_1,z_2)=0}} \frac{z_1^{-y_1-1}z_2^{-y_2-1}u(z_1,z_2;x_1,x_2)}{\det \Lambda(z_1,z_2)} = \sum_{(z_1,z_2)\in\Xi} \ell(\textbf{y},\textbf{z})u(\textbf{z},\textbf{x}).
\end{equation}
\end{lemma}

\begin{proof}
Symmetrizing the summand on the left hand side we have

\begin{equation}
    \sum_{\substack{(z_1,z_2):\\ z_2^L-S(z_2,z_1)=0\\ z_1^L-S(z_1,z_2)=0}} \frac{z_1^{-y_1-1}z_2^{-y_2-1}u(z_1,z_2;x_1,x_2)}{\det \Lambda(z_1,z_2)} = \sum_{(z_1,z_2)\in\Xi} \sum_{\sigma\in S_2} \frac{u(\sigma\cdot\textbf{z};\textbf{x})}{z_{\sigma(1)}^{y_1+1}z_{\sigma(2)}^{y_2+1} \det(\Lambda(z_1,z_2))}.
\end{equation}

For the only nontrivial $\sigma\in S_2$, we have \begin{equation}
\begin{split}
    u(\sigma\cdot\textbf{z};\textbf{x}) &= z_2^{x_1}z_1^{x_2}-\frac{1+z_1z_2-2\Delta z_1}{1+z_1z_2-2\Delta z_2} z_1^{x_1}z_2^{x_2} \\
    &=-\frac{1+z_1z_2-2\Delta z_1}{1+z_1z_2-2\Delta z_2} u(\textbf{z};\textbf{x}).
\end{split}
\end{equation}

So then we have 
\begin{equation}
\begin{split}
    \sum_{(z_1,z_2)\in\Xi} \sum_{\sigma\in S_2} \frac{u(\sigma\cdot\textbf{z};\textbf{x})}{z_{\sigma(1)}^{y_1+1}z_{\sigma(2)}^{y_2+1} \det(\Lambda(z_1,z_2))} &= \sum_{(z_1,z_2)\in\Xi}\left[\sum_{\sigma\in S_2}\frac{u(z_1,z_2;x_1,x_2)}{A_\sigma (z_1,z_2)z_{\sigma(1)}^{y_1+1}z_{\sigma(2)}^{y_2+1} \det(\Lambda(z_1,z_2))}\right] \\
    &= \sum_{(z_1,z_2)\in \Xi} \ell(\textbf{y},\textbf{z})u(\textbf{z};\textbf{x}).
    \end{split}
\end{equation}
\end{proof}
Thus by \Cref{l: contour to residue sums}, \Cref{l: multipart sums} and \Cref{l:last Bethe expansion equality}, we have shown \Cref{l: Bethe function expansion}.

\subsubsection{Puiseux Expansion of Bethe polynomials}
\begin{lemma}\label{l: puiseux expansion}
    Let $K\in\mathbb{N}_{\geq 2}$. For a polynomial $F\in\mathbb{C}[x,y]$ assume that $(x_0,y_0)$ is a root such that $\partial_{x}^{k-1}F(x_0,y_0)=0$ for $k=1,\ldots,K$ and $\partial_{x}^{K}F(x_0,y_0)\neq0$. If $\partial_{y}F(x_0,y_0)\neq0$, then $F$ has the Puiseux expansion \begin{equation}
        x-x_0=\alpha(y-y_0)^{1/K}+O\left((y-y_0)^{2/K}\right).
    \end{equation}
\end{lemma}
\begin{proof}
Let $(\overline{x},\overline{y})$ be local coordinates given by $\overline{x}=x-x_0$ and $\overline{y}=y-y_0$. In these coordinates, we can write $F$ as 
\begin{equation}
    F(\overline{x},\overline{y})=\sum_{j=0}^{d_{\overline{x}}}h_j(\overline{y})\overline{x}^j
\end{equation}
where $d_{\overline{x}}$ is the degree of $\overline{x}$ in $F$ and 
\begin{equation}            h_j(\overline{y})=\sum_{m=0}^{d_{\overline{y}}^j}C_m^j\overline{y}^m
\end{equation}
where $d_{\overline{y}}^j$ is the degree of $\overline{y}$ as a coefficient of $\overline{x}^j$. Then \begin{equation}\partial_{\overline{x}}^{k-1}F(0,0)=h_{k-1}(0)=C_{0}^{k-1}=0\end{equation} and \begin{equation}\partial_{\overline{x}}^{K}F(0,0)=h_{K}(0)=C_{0}^{K}\neq0.\end{equation} Moreover, we have \begin{equation}
    \partial_{\overline{y}}F(\overline{x},\overline{y})=\sum_{j=0}^{d_{\overline{x}}}h'_j(\overline{y})\overline{x}^j 
\end{equation}
and 
\begin{equation}
      \partial_{\overline{y}}F(0,\overline{y})=h'_0(\overline{y}).
\end{equation}
Now suppose that $\partial_{\overline{y}}F(0,0)=C_1^0\neq0$. Then the Newton Polygon of $F(\overline{x},\overline{y})$ near $(\overline{x},\overline{y})=(0,0)$ consists of a single line segment from $(0,1)$ to $(K-1,0)$. So then the Puiseux expansion \cite{Casas-Alvero_2000} of $F(\overline{x},\overline{y})$ at $(0,0)$ is \begin{equation}
    \overline{x}=\alpha  \overline{y}^{1/K} +O\left(\overline{y}^{2/K}\right). 
\end{equation}
for some nonzero $\alpha\in\mathbb{C}$.
Then in terms of $x$ and $y$ near $(x_0,y_0)$, we have \begin{equation}
    x-x_0=\alpha (y-y_0)^{1/K}+O\left((y-y_0)^{2/K}\right).
\end{equation}
\end{proof}

\begin{lemma}\label{l: Bethe Puiseux}
    For $\Delta\neq \cos(2\pi k/(L-2))$, $k\in[L-2]$, the polynomial \begin{equation}
        q_\xi(w) = 1-2\Delta w+w^L+(w-2\Delta w^L + w^{L+1})\xi
    \end{equation}
    has a Puiseux expansion of $\lambda=\lambda_0+\alpha(\xi-\xi_0)^{1/k}$ near roots $\lambda_0\in Q(\xi_0)$ of multiplicity $k$, where $\alpha\in\mathbb{C}\setminus\{0\}$.  
\end{lemma}
\begin{proof}
By \Cref{l: puiseux expansion}, checking that there are no solutions to the system of equations $q_\xi(w)=0$ and $\partial_\xi q_\xi(w)=0$ is sufficient. Writing this system explicitly we have
\begin{equation}
    \begin{split}
        1-2\Delta w + w^L+\left(w-2\Delta w^L+ w^{L+1}\right)\xi =0  \\
        w - 2\Delta w^L + w^{L+1}=0
    \end{split}
\end{equation}
which implies the following two equations 
\begin{equation}
    \begin{split}
        1-2\Delta w + w^L=0 \\ 
        1-2\Delta w^{L-1}+ w^L=0.
    \end{split}
\end{equation}
We get equality between these two equations when $w^{L-1}=w$. However substituting $w$ for $w^{L-1}$ in the above equations gives us 
\begin{equation}
    w=\Delta\pm\sqrt{\Delta^2-1}.
\end{equation}
If $w^{L-1}=w$, then $w=0$ or $w=e^{2\pi i k/(L-2)}$ for $k\in[L-2]$. It is clear that $w=0$ is not possible. For $0<|\Delta|<\frac{L-1}{L}\frac12$, we have $w = \Delta \pm i\sqrt{1-\Delta^2}$. Then for $w=e^{2\pi i k/(L-2)}$, we would need $\Delta = \cos(2\pi k/(L-2))$\footnote{Interestingly, these are zeroes of the Chebyshev polynomials of the first kind. Similar values, zeroes of the Cehbyshev polynomials of the second kind, are presented in \cite{Gutkin2000-gr} as critical values of $\Delta$ in the case of the $XXZ$ spin-1/2 chain on $\mathbb{Z}$.}, which is false by assumption. So a $w$ such that $w^{L-1}=w$ is not possible in this setting, and thus $\partial_{\xi}q_\xi(w)\neq 0$. Hence the conditions of \Cref{l: puiseux expansion} are satisfied, and we are done. 
\end{proof}

\section{Conclusion}

We have developed a constructive coordinate--energy correspondence for the periodic Heisenberg--Ising XXZ spin-$\tfrac12$ chain.  While the coordinate Bethe Ansatz gives the transformation from particle configurations to Bethe eigenvectors, our inverse coordinate Bethe Ansatz formula provides a candidate transformation in the opposite direction.  We proved this inverse formula at the free-fermion point $\Delta=0$ for arbitrary particle number and, under the stated nondegeneracy and small-anisotropy assumptions, for the two-particle sector.  The numerical verifications for additional values of $N$, $L$, and $\Delta$ provide further evidence for the general conjecture.  Whenever the inverse formula holds, the Bethe vectors form a complete basis, yielding an explicit finite-volume solution of the Schr\"odinger equation and exact formulas for configuration probabilities.

Using this spectral representation, we also derived a determinant-based formula for the one-point occupation probability.  The simplification, obtained through identities related to the Izergin--Korepin determinant, replaces the original sum over configurations and permutations by a more structured expression that is better suited to exact computation and asymptotic analysis.  The main remaining problem is to prove the inverse transformation for general particle number and a broader range of anisotropy parameters, including a precise characterization of the exceptional values and possible degeneracies of the Gaudin-type determinant.  Such a result would place completeness, finite-volume dynamics, and observable formulas within a single explicit framework and could provide a starting point for studying large-scale spin transport and possible KPZ-type fluctuations on the ring.

\section*{Acknowledgments}
The authors thank and acknowledge the students from OSU's URSA Engage program: Mohammad Faks, Christie Chang, Johannes Huurman, Crystal Lee, and Mason Spears for their contributions to the code implementation. The authors also acknowledge the OSU's College of Science \textit{SciRIS Stage 2 and 3} grant, OSU's \textit{URSA Engage} program, and the Simons Foundation \textit{Pivot Fellowship}. 

\printbibliography
\end{document}